\documentclass[a4paper, 11pt]{article}
\usepackage{graphicx} % Required for inserting images
\usepackage{natbib}

\usepackage{url}
\usepackage{hyperref}

\setcitestyle{authoryear, open={(}, close={)}}

\date{} 					

\usepackage{hyphenat} 
\usepackage[textwidth=151mm, hcentering, vmargin={24mm,32.5mm}, foot=20mm]{geometry}

\usepackage{authblk}

\author[1]{Veeti~Ahvonen}
\author[2]{Maurice~Funk}
\author[1]{Damian~Heiman}
\author[1]{Antti~Kuusisto}
\author[2]{Carsten~Lutz}

\affil[1]{Tampere University, Finland}
\affil[2]{Leipzig University, Germany}

\usepackage{tikz}
\usepackage{multicol}

\usepackage{arydshln}
\usepackage{makecell}

\usepackage{colonequals}
\usepackage{amsmath, amssymb, amsthm}
\usepackage{mathtools}
\usepackage{thm-restate}
\usepackage{graphicx} % Required for inserting images

\usepackage{verbatim}

\usepackage[utf8]{inputenc}
\usepackage[T1]{fontenc}

\usepackage{enumitem}
\usepackage{tikz}
\usepackage{lipsum}
\usepackage{pdfpages}
\usepackage{multicol} 
\usepackage{xspace} 
\usepackage{stmaryrd}

\PassOptionsToPackage{numbers, compress}{natbib}

\usepackage{makecell}
\usepackage{mathtools}
\usepackage{colonequals}

\usepackage{marginnote}

\usepackage{bm}

\theoremstyle{plain}
\newtheorem{theorem}{Theorem}
\newtheorem{lemma}[theorem]{Lemma}
\newtheorem{proposition}[theorem]{Proposition}
\newtheorem{corollary}[theorem]{Corollary}

\theoremstyle{definition}

\newtheorem{example}[theorem]{Example}
\newtheorem{definition}[theorem]{Definition}
\newcommand{\N}{\mathbb N}
\newcommand{\Z}{\mathbb Z}

\newcommand{\R}{\mathbb R}

\newcommand{\bb}{\mathbf{b}}

\newcommand{\be}{\mathbf{e}}
\newcommand{\bbf}{\mathbf{f}}

\newcommand{\bs}{\mathbf{s}}

\newcommand{\bu}{\mathbf{u}}
\newcommand{\bv}{\mathbf{v}}
\newcommand{\bw}{\mathbf{w}}
\newcommand{\bx}{\mathbf{x}}

\newcommand{\cC}{\mathcal{C}}

\newcommand{\cF}{\mathcal{F}}
\newcommand{\cG}{\mathcal{G}}
\newcommand{\cH}{\mathcal{H}}
\newcommand{\cI}{\mathcal{I}}

\newcommand{\cL}{\mathcal{L}}
\newcommand{\cM}{\mathcal{M}}

\newcommand{\cP}{\mathcal{P}}

\newcommand{\cT}{\mathcal{T}}

\newcommand{\fG}{\mathfrak{G}}

\newcommand{\ordo}{\mathcal{O}}

\newcommand{\softmax}{\mathrm{softmax}}

\newcommand{\FO}{\mathrm{FO}}
\newcommand{\GML}{\mathrm{GML}}

\newcommand{\SA}{\mathrm{SA}}
\newcommand{\T}{\mathsf{T}}

\newcommand{\MP}{\mathrm{MP}}

\newcommand{\FF}{\mathrm{FF}}
\newcommand{\AGG}{\mathrm{AGG}}
\newcommand{\COM}{\mathrm{COM}}
\newcommand{\MLP}{\mathrm{MLP}}
\newcommand{\ReLU}{\mathrm{ReLU}}
\newcommand{\SUM}{\mathrm{SUM}}

\newcommand*{\abs}[1]{\lvert#1\rvert}   

\newcommand{\DiamondG}{\langle \mathsf{G} \rangle}

\newcommand{\GMLGC}{\mathrm{GML+GC}}
\newcommand{\GMLG}{\mathrm{GML+G}}
\newcommand{\PLGC}{\mathrm{PL+GC}}
\newcommand{\PLG}{\mathrm{PL+G}}
\newcommand{\PL}{\mathrm{PL}}
\newcommand{\MLGC}{\mathrm{ML+GC}}
\newcommand{\MLG}{\mathrm{ML+G}}
\newcommand{\ML}{\mathrm{ML}}
\newcommand{\GMSC}{\mathrm{GMSC}}

\newcommand{\AH}{\mathrm{AH}}
\newcommand{\UH}{\mathrm{UH}}

\newcommand{\GNN}{\mathrm{GNN}}
\newcommand{\GNNG}{\mathrm{GNN+G}}
\newcommand{\GNNGC}{\mathrm{GNN+GC}}
\newcommand{\GT}{\mathrm{GT}}
\newcommand{\GPS}{\mathrm{GPS}}

\newcommand{\GNNF}{\mathrm{GNN[F]}}
\newcommand{\GNNGF}{\mathrm{GNN+G[F]}}
\newcommand{\GNNGCF}{\mathrm{GNN+GC[F]}}
\newcommand{\GTF}{\mathrm{GT[F]}}
\newcommand{\GPSF}{\mathrm{GPS[F]}}
\newcommand{\MLPF}{\mathrm{MLP[F]}}

\newcommand{\LAB}{\mathrm{LAB}}

\newcommand{\lpto}{\to^{\abs{\cdot}}}

\newcommand{\round}{\mathit{round}}

\newcommand{\B}{\mathbb B}

\newcommand{\NaN}{\mathrm{NaN}}

\newcommand{\Gmc}{\ensuremath{\mathcal{G}}\xspace}
\newcommand{\Hmc}{\ensuremath{\mathcal{H}}\xspace}
\newcommand{\Imc}{\ensuremath{\mathcal{I}}\xspace}

\newcommand{\mn}[1]{\ensuremath{\mathsf{#1}}}

\newcommand{\neigh}{\operatorname{Neigh}}
\newcommand{\pred}{\operatorname{Pred}}

\title{%
Expressive Power of Graph Transformers via Logic
}

\usepackage{enumitem}
\setitemize{noitemsep,topsep=0pt,parsep=3pt,partopsep=0pt}
\setenumerate{noitemsep,topsep=0pt,parsep=3pt,partopsep=0pt}

\date{\today}
\usepackage{amsmath}

\begin{document}
\maketitle
\begin{abstract}
Transformers are the basis of modern large language models, but relatively little is known about their precise expressive power on graphs. We study the expressive power of graph transformers (GTs) by Dwivedi and Bresson (2020) and GPS-networks by Rampásek et al. (2022), both under soft-attention and average hard-attention. Our study covers two scenarios: the theoretical setting with real numbers and the more practical case with floats. With reals, we show that in restriction to vertex properties definable in first-order logic (FO), GPS-networks have the same expressive power as graded modal logic (GML) with the global modality. With floats, GPS-networks turn out to be equally expressive as GML with the counting global modality. The latter result is absolute, not restricting to properties definable in a background logic. We also obtain similar characterizations for GTs in terms of propositional logic with the global modality (for reals) and the counting global modality (for floats).
\end{abstract}

\section{Introduction}\label{Introduction}

Transformers have emerged as a powerful machine learning architecture serving as the basis of modern large language models such as GPTs \citep{NIPS2017_3f5ee243} and finding success also, e.g., in computer vision \citep{DBLP:conf/iclr/DosovitskiyB0WZ21}.
Recently, transformers have received significant attention in the field of graph learning, traditionally 
dominated by graph neural networks~(GNNs)~\citep{ScarselliGTHM09} and related formalisms like graph convolutional networks~\citep{kipf2017semi}. This shift
is driven by well‑known challenges GNNs face in handling long‑range interactions, including issues such as over-squashing~\citep{DBLP:conf/iclr/0002Y21} and over-smoothing~\citep{DBLP:conf/aaai/LiHW18}. Whereas GNNs rely primarily on local message passing, transformers can attend globally to any vertex in the graph. 
The literature now includes many 
graph learning models incorporating transformers.
An important distinction is between `pure' transformer models, which ignore the graph structure and
result in `bags-of-vertices' models \citep{NEURIPS2019_9d63484a, DBLP:conf/nips/KreuzerBHLT21}, and hybrids 
that combine transformers and GNN-style message passing \citep{Rampasek}.

To understand the limitations of learning models and their relationships, an expanding literature characterizes the expressive power of such models using logical
formalisms. While transformers on words as used in GPTs connect to versions of linear temporal logic and first-order logic with counting \citep{li2025characterizingexpressivitytransformerlanguage,chiang2023}, GNNs relate to variants of graded modal logic (GML) \citep{Barcelo_GNNs, benediktandothers}.

In 
this article, we provide logical 
characterizations of graph learning models that incorporate transformers. Our characterizations are uniform in that we
do not impose a constant bound on graph size. We are
primarily interested in models that 
combine GNN message passing layers with transformer layers, and focus
in particular on the rather general GPS-networks of \citet{Rampasek}. 
In addition, we also
consider pure bags-of-vertices graph transformers (GTs) \citep{NEURIPS2019_9d63484a,Dwivedi}. 
For both models, we study the case where
features are vectors of real numbers,
as in most theoretical studies, and also
the case where they are floats, as in real-life implementations. 
We study both soft-attention and average hard-attention in the transformer layers. We 
focus on these models
in their `naked' form, without positional (or structural) encodings. Such encodings---often based on the graph Laplacian, homomorphism counts, and notions of graph centrality---enrich each vertex with information regarding its position in the graph. While they play an
important role in graph learning with transformers, there is an uncomfortably large zoo of them. 
Therefore, we believe that to characterize expressive power, it is natural to begin with the naked case, providing a foundation for analyzing models with encodings.
We focus on vertex classification as a basic learning task, but many of our results also generalize to graph classification tasks, see the end of Section~\ref{sect: float characterization} and 
Appendix~\ref{appendix: graph classification}.

To survey our  results, we start
with the case of real numbers.  Our first main result is that in restriction to vertex properties expressible in first-order logic ($\FO$), $\GPS$-networks based on reals have the same expressive power as $\GML$ with the (non-counting) global modality ($\GMLG$), stated in Theorem \ref{thm:real-GPS-GMLG}. As  all our results, this applies to both soft-attention and average hard-attention, 
assuming sum aggregation in message-passing layers as in~\citep{Barcelo_GNNs}.
While it is unsurprising that adding transformer layers to a GNN corresponds to adding a global feature to the logic $\GML$, it was far
from clear that this feature is the non-counting global modality, rather than, say, its counting version. 
Our result implies that $\GPS$-networks cannot globally count in an absolute way, as in `the graph contains at least 10 vertices labeled $p$'. In contrast, they can
globally count in a relative way, as in
`the graph contains more vertices labeled $p$ than vertices labeled $q$'. This, however, is not expressible in $\FO$. The 
proof of our result is non-trivial and requires the introduction of a new type of bisimilarity (global-ratio graded bisimilarity $\sim_{G\%}$) and the proof of a new van
Benthem/Rosen-style result that essentially states: if an $\FO$-formula $\varphi(x)$ is invariant under $\sim_{G\%}$, then it is equivalent to a $\GMLG$-formula. We also
prove that, relative to $\FO$, real-based $\GT$s have
the same expressive power as propositional
logic $\PL$ with only the non-counting global modality
($\PLG$),
stated as Theorem~\ref{thm:real-GT-PLG}.

We next discuss our results regarding floating-point numbers. Our main result
for this case is that float-based $\GPS$-networks have the same expressive power as 
$\GML$ with the counting global modality ($\GMLGC$), stated as Theorem~\ref{theorem: GMLGC = SGPS = AHGPS = GNNG}. In contrast to the
case of the reals, this characterization is absolute rather than relative to $\FO$. 
It applies to any reasonable aggregation function including sum, max and mean.
We consider it remarkable that transitioning from reals to floats results in incomparable expressive power: while
relative global counting is no longer possible, absolute global counting becomes expressible. 
Our proof techniques leverage the underflow phenomenon of floats.
Via techniques from \citep{ahvonen_neurips}, it follows that with \emph{floats}, $\GPS$-networks have the same expressive power as GNNs with counting global readout ($\GNNGC$).
We also show that float-based $\GT$s have the same expressive power as $\PL$ with the counting global modality
($\PLGC$), stated in Theorem~\ref{theorem: PLGC = SGT = AHGT}. 
We note that the logic $\PLGC$ was recently used to study links between entropy and description complexity \citep{JAAKKOLA2025103615}.
We also briefly discuss our characterizations with floats in restriction to word-shaped graphs, how the float results generalize for graph classification and non-Boolean classification tasks, and how they could be modified to cover positional encodings.
Finally, we emphasize that the paper focuses exclusively on theoretical analysis and does not involve any experimental evaluation.

\paragraph{Related Work.}

To our knowledge, the expressive power of graph transformers and $\GPS$-networks has not yet been studied from the perspective of logic. 
Regarding  transformers \emph{over words}, 
the closest to our work are \citep{jerad2025uniquehardattentiontale, li2025characterizingexpressivitytransformerlanguage}, which characterized fixed-precision transformers that use `causal masking' with soft-attention, average hard-attention and unique hard-attention, via the past fragment of linear temporal logic (LTL). 
In contrast, our characterizations exclude masking and focus on transformers on graphs rather than words.
Similar characterizations via variations of LTL are given, for instance, in \citep{Yang0A24, yang2025kneedeepcrasptransformerdepth}.

Regarding logic-based lower and upper bounds for the expressive power of transformers over words, we mention \citep{chiang2023}, which established the logic $\mathsf{FOC[+;MOD]}$ as an upper bound for the expressivity of fixed-precision transformer encoders and as a lower bound for transformer encoders working with reals.
\citet{DBLP:conf/nips/MerrillS23} showed that first-order logic with majority quantifiers is an upper bound for log-precision transformers. 
\citet{yang2024counting} gave a lower bound for future-masked soft-attention transformers with unbounded input size via the past fragment of Minimal Tense Logic with counting terms ($\mathsf{K_t}$[\#]).
\citet{BarceloKLP24} established two lower bounds for hard-attention transformers working with reals: first-order logic extended with unary numerical predicates (FO($\mathsf{Mon}$)) in the case of unique hard-attention and linear temporal logic extended with unary numerical predicates and counting formulas (LTL(C,+)) for average hard-attention. 
For a more comprehensive study of precise characterizations and upper/lower bounds on the expressive power of different transformer architectures on words, see the survey \citep{StroblMW0A24}.

Regarding non-logic-based studies of the expressive power of graph transformers, \citet{DBLP:conf/nips/KreuzerBHLT21} showed that graph transformers with positional encodings are universal function approximators in the non-uniform setting.
In the uniform setting, \citet{grohe_transformers} proved that $\GNNGC$s and $\GPS$-networks using reals have incomparable expressive power w.r.t. function approximation. 
By contrast, this paper proves that with floats, $\GNNGC$s and $\GPS$-networks are equally expressive w.r.t. both graph and vertex properties, and also w.r.t. non-Boolean graph classification which amounts to expressing functions over floating-point numbers on graphs.
Recently, \citet{Schwartzentruber} studied the complexity of verifying float-based GNNs, and showed that it is PSPACE-complete.

\citet{Barcelo_GNNs} pioneered work on the expressive power of graph neural networks by characterizing aggregate-combine $\GNN$s with reals in restriction to 
$\FO$ via graded modal logic, and the same $\GNN$s extended with a global readout mechanism (essentially the same as our $\GNNGC$s) with the two-variable fragment of $\FO$ with counting quantifiers. 
\citet{DBLP:conf/lics/Grohe23} connected GNNs 
to the circuit complexity class $\mathsf{TC}^0$,
utilizing dyadic rationals.  
\citet{ahvonen_neurips} gave logical characterizations of recurrent and constant-iteration GNNs with both reals and floats, making similar assumptions to ours on float operations such as summation.
\citet{Pfluger_Tena_Cucala_Kostylev_2024} characterized recurrent GNNs that use reals in terms of the graded two-way $\mu$-calculus relative to a logic LocMMFP.
We also mention \citep{benediktandothers}, which characterized GNNs with bounded activation functions via logics involving Presburger quantifiers.

\section{Preliminaries}\label{section: preliminaries}

We let $\Z_+$ denote the set of positive integers and $\N$ the set of non-negative integers.
For $n \in \Z_+$, define $[n] \colonequals \{1, \ldots, n\}$.
Also, set $\mathbb{B} \colonequals \{0,1\}$.  
For a set $S$, we let $\cM(S)$ denote the set of multisets over~$S$, i.e., the set of functions \mbox{$S \to \N$}. We use double curly brackets $\{\{ \cdots \}\}$ to denote multisets.
For a multiset~$M$, $M_{|k}$ denotes the \textbf{k-restriction} of~$M$, i.e., the multiset given by $M_{|k}(x) = \min\{M(x),k\}$.

For 
$\bx \in X^n$ and $i \in [n]$, 
let $\bx_i$ denote the $i$th component of $\bx$.
For a matrix $M \in X^{n\times m}$, we use  $M_{i,*}$, $M_{*,j}$ and $M_{i,j}$ to denote, respectively, the $i$th row (from the top), the $j$th column (from the left), and the 
$j$th entry in the $i$th row of $M$.
For a sequence $(M^{(1)}, \dots, M^{(k)})$ of matrices in $X^{n \times m}$, their \textbf{concatenation} is the matrix $M \in X^{n \times k m}$ such that $M^{(\ell)}_{i,j} = M_{i,(\ell - 1) m +  j}$ for all $\ell \in [k]$.
For non-empty sets $X$ and~$Y$, let $X^+$ denote the set of non-empty sequences over $X$, while $f \colon X^+ \lpto\ Y^+$ is a notation for functions that
map each sequence in $X^+$ to a sequence of the same length in $Y^+$.

\subsection{Graphs and Feature Maps}

For a finite domain $D \neq \emptyset$, a dimension $d \in \N$ and a non-empty set $W$, a \textbf{feature map} is a function $f \colon D \to W^d$ that maps each 
$x \in D$ to a \textbf{feature vector} $f_x \in W^d$.
Typically, $D$ consists of graph vertices and $W$ is $\R$ or a set of  floating-point numbers.
If 
$D$ is ordered by $<^D$, then 
we can identify
$f$
with
a \textbf{feature matrix}  $M \in W^{|D| \times d}$, 
where the row $M_{i,*}$ contains the feature vector of the $i$th element of $D$ w.r.t. $<^D$, 
the column $M_{*,j}$ containing the $j$th vector components.

We work with vertex-labeled directed graphs, allowing self-loops, and simply refer to them as \emph{graphs}.
Let $\mathrm{LAB}$ denote a countably infinite set of \textbf{(vertex) label symbols}.
We assume an ordering $<^{\mathrm{LAB}}$ of $\mathrm{LAB}$, also inducing an ordering $<^L$ of every $L \subseteq \LAB$.
Finite subsets of $\LAB$ are denoted by $\Pi$.
Given $\Pi \subseteq \mathrm{LAB}$,
a \textbf{$\Pi$-labeled graph} is
a tuple $\cG = (V, E, \lambda)$, where $V$ is a finite non-empty set of \textbf{vertices}, $E \subseteq V \times V$ is a set of \textbf{edges} and $\lambda \colon V \to 2^\Pi$ a \textbf{vertex labeling function}. Here, $2^\Pi$ denotes the power set of $\Pi$.
For convenience, we set $V(\cG) \colonequals V$, $E(\cG) \colonequals E$, and $\lambda(\cG) \colonequals \lambda$. A \textbf{pointed graph} is a pair $(\cG, v)$ with $v \in V$.
The set of \textbf{out-neighbors} of $v \in V(\cG)$ is $\neigh_\cG(v) \colonequals \{\, u \mid (v, u) \in E\,\}$. 
We may identify 
$\lambda$ with a feature map $\lambda' \colon V \to \B^{\abs{\Pi}}$ where 
$\lambda'(v)_i = 1$ if the $i$th vertex label symbol 
(w.r.t. $<^{\mathrm{\Pi}}$) is in $\lambda(v)$ and else $\lambda'(v)_i = 0$.
Thus, $\Pi$-labeled graphs can be seen as graphs where each vertex is labeled with a single vector from $\B^{\abs{\Pi}}$.
We assume w.l.o.g. that 
for any graph $\cG$, $V(\cG) = [n]$ for some $n \in \Z_+$.
Hence,
we can identify
feature maps of graphs 
with 
feature matrices and use labeling functions, feature maps and feature matrices interchangeably.

\subsection{Graph Transformers and GNNs}\label{sec: GTs and GNNs}

We next discuss the computing architectures relevant to this article: graph transformers, $\GPS$-networks, and $\GNN$s.
We view them as vertex classifiers that produce Boolean classifications. 
For this section, fix an arbitrary $\Pi$-labeled graph $\cG = (V, E, \lambda)$ with $\abs{\Pi} = \ell$ and $|V| = n$. 
In what follows, we will often speak of
the input/hidden/output dimension of learning models and their components. For 
brevity, we abbreviate these to I/H/O  dimension.

\subsubsection{Basic Components.}

A multilayer perceptron (or feedforward neural network) can be intuitively understood as a structure that propagates an input through a series of layers, where in each layer every ``neuron'' computes a weighted sum of its incoming signals, applies an activation function, and passes the result forward to the next layer. The formal details, along with some auxiliary notions, are given below.

A \textbf{perceptron layer} $P$ of I/O dimension $(d_I, d_O)$ consists of a \textbf{weight matrix} $W \in \R^{d_O \times d_I}$, a \textbf{bias term} $b \in \R^{d_O \times 1}$ and an \textbf{activation function} $\alpha \colon \R \to \R$.
Given an input vector $\bx \in \R^{d_{I}}$, $P$ computes the vector $P(\bx) \colonequals \alpha \big( b + W\bx \big)$, where $\alpha$ is applied element-wise.
A \textbf{multilayer perceptron} ($\MLP$) $F$ of I/O dimension 
$(d_I, d_O)$ is a sequence $(P^{(1)}, \dots, P^{(m)})$ of perceptron layers, where each $P^{(i)}$ has I/O dimension $(d_{i-1}, d_{i})$, where $d_0 = d_I$ and $d_m = d_O$ and $P^{(m)}$ uses the identity activation function. 
Given a vector $\bx \in \R^{d_I}$, $F$ computes the vector $F(\bx) \colonequals P^{(m)}(\cdots P^{(2)}(P^{(1)}(\bx)) \cdots)$.
For a matrix $X \in \R^{n \times d_I}$, we let $F(X)$ denote the $\R^{n \times d_O}$-matrix, where $F$ is applied row-wise for $X$.
An MLP
is \textbf{$\alpha$-activated}
if every layer uses $\alpha$, except the last, which always uses the identity function. 
Unless otherwise stated, 
$\MLP$s are $\mathrm{ReLU}$-activated, where $\mathrm{ReLU}(x) = \max(0, x)$.
An $\MLP$ is \textbf{simple} if it 
is $\ReLU$-activated and has only two perceptron layers.

Next, we introduce aggregation functions and readout gadgets, which are essential components of graph transformers, GPS-networks and graph neural networks.
An \textbf{aggregation function} of dimension $d_I$ is a function $\AGG \colon \cM(\R^{d_I}) \to \R^{d_I}$ which typically is (point-wise) sum, max or mean. It is \textbf{set-based} if $\AGG(M) = \AGG(M_{|1})$ for all $M \in \cM(\R^{d_I})$.
A \textbf{readout gadget} of I/O dimension $(d_I, d_O)$ is a tuple $R \colonequals (F, \AGG)$, where $F$ and $\AGG$ are as above. Given a matrix $X \in \R^{n \times d_I}$, it computes the matrix $R(X) \in \R^{n \times d_O}$ where each row is the same, defined by
    $R(X)_{i,*} \colonequals F \big( \AGG ( \{\!\{\, X_{j,*} \mid j \in [n] \,\}\!\} ) \big)$.

\subsubsection{Graph Neural Networks.}

A (constant-iteration) graph neural network can be intuitively viewed as a distributed system in which vertices of the input graph exchange information synchronously for a fixed number of rounds. In each round, every vertex updates its feature vector based on its previous feature vector and an aggregated representation of the feature vectors of its out-neighbors. The formal details follow below.

A \textbf{message-passing layer} of dimension $d$ is a pair $L = (\COM, \AGG)$, where $\COM$ is an $\MLP$ of I/O dimension $(2d, d)$ and $\AGG$ is an aggregation function of dimension $d$. 
A \textbf{message-passing layer with counting global readout} of dimension $d$ is a pair $(L, R)$, where $L$ is defined as above and $R$ is a readout gadget of I/O dimension $(d,d)$. 
A \textbf{message-passing layer with non-counting global readout} $(L, R)$ of dimension~$d$ is defined analogously, but the aggregation function of $R$ is set-based. 

A \textbf{graph neural network} ($\GNN$) over $(\Pi, d)$ 
is a tuple $G = (P, L^{(1)}, \ldots, L^{(k)}, C)$ where $P$ is an $\MLP$ of I/O dimension $(\ell, d)$, each $L^{(i)} = (\COM^{(i)}, \AGG^{(i)})$ is a message passing layer of dimension $d$, and $C$ is an $\MLP$ of I/O dimension $(d,1)$ 
that induces a function $\R^{d} \to \B$ 
called a \textbf{(Boolean vertex) classification head}. 
The $\MLP$ $C$ does not have to be $\ReLU$-activated, and can use, e.g., the Heaviside function $\sigma$, defined such that $\sigma(x) = 1$ if $x > 0$ and $\sigma(x) = 0$ if $x \leq 0$.
Over a graph  $\cG$, $G$ computes  a sequence $\lambda^{(0)}, \ldots, \lambda^{(k)}$ of feature maps and a final feature map $G(\cG)$ as follows: $\lambda^{( 0 )} \colonequals P(\lambda)$ and $\lambda^{( i+1)} \colonequals \lambda^{( i)} + L^{(i+1)}\big(\lambda^{( i)} \big)$, where for each $v \in V$, $L^{(i+1)}\big(\lambda^{( i)} \big)$ maps $v$ to the feature vector
\[
    \COM^{(i+1)} \Big( \lambda^{( i)}_v, \AGG^{(i+1)}  \big( \{\!\{ \lambda^{( i)}_u \mid (v,u) \in E \}\!\} \big) \Big).
\]
Finally, $G(\cG) \colonequals C(\lambda^{( k )})$.\footnote{Note that messages flow in the direction opposite to the edges. However, this is only a convention and could be reversed by using modal logics that operate over the inverse edge relation.} The final feature map $G(\cG)$ labels each vertex with $0$ or $1$, and we say that $G$ \textbf{accepts} a vertex $v$ of $\cG$ if $G(\cG)$ maps $v$ to $1$.
Note that we follow the convention of \citep{Rampasek,grohe_transformers} by including skip connections around message-passing layers, 
which refers to the fact that 
$\lambda^{( i+1)}$
is not simply defined as $L^{(i+1)}\big(\lambda^{( i)} \big)$. It is easy to see that GNNs have the same expressive power with and without skip connections.

We define \textbf{graph neural networks with counting global readout} ($\GNNGC$s) and \textbf{graph neural networks with non-counting global readout} ($\GNNG$s) analogously, except each $L^{(i)}$ is a message-passing layer $(\hat{L}^{(i)}, R^{(i)})$ with counting and non-counting global readout, respectively. They behave analogously to GNNs, except that $\lambda^{( i+1)} \colonequals  \hat{\lambda}^{( i+1)} + R^{(i+1)}(\hat{\lambda}^{( i+1)})$ where $\hat{\lambda}^{( i+1 )} \colonequals \lambda^{( i)} + \hat{L}^{(i+1)}(\lambda^{( i)}) $.
In other words, each vertex's feature vector additionally reflects the global information aggregated from all vertices in the graph.

\begin{example}
    We can define a $\GNN$ over $(\{p\}, 1)$ that accepts precisely the vertices that satisfy the following property $\cP$: `a vertex has at least $5$ out-neighbors that have the label $p$'. The initial $\MLP$ assigns the feature vector $(1)$ to each vertex labeled $p$, and $(0)$ to others. The first (and only) message-passing layer uses summation as the aggregation function; if the sum is at least $5$, the vertex has the property $\cP$. Similarly, a $\GNNG$ can express that a graph contains a vertex with the property $\cP$ by adding a second layer that sums over all vertices in the graph, and a $\GNNGC$ can express, e.g., that a graph contains exactly $3$ vertices with the property $\cP$. 
\end{example}

\subsubsection{Self-Attention and Graph Transformers.}

A graph transformer is intuitively quite similar to a $\GNN$; each vertex of a graph updates its feature vector in synchronous rounds. The key difference is that instead of aggregation and combination functions, a global attention mechanism is used to obtain information from all vertices. 
A $\GPS$-network is then a variation of graph transformers that combines the attention mechanism with the update mechanism of $\GNN$s.
The formal details are given below.

A \textbf{self-attention head} $H$ of I/H dimension 
$(d, d_h)$ over $\R$ is defined w.r.t.
an \textbf{attention-function $\alpha \colon \R^+ \lpto \R^+$} 
and
three $\R^{d \times d_h}$-matrices: the \textbf{query-matrix} $W_Q$, the \textbf{key-matrix} $W_K$ and the \textbf{value-matrix} $W_V$.
Given a matrix $X \in \R^{n \times d}$, it computes the  $n\times d_h$-matrix 
\begin{equation*} 
H(X) \colonequals \alpha\left( \frac{(X W_Q) (X W_K)^{\T}}{\sqrt{d_h}} \right) (XW_V),
\end{equation*}
where $\alpha$ is applied row-wise. 
A \textbf{self-attention module} of dimension $d$ over $\R$ is a tuple $\SA~\colonequals~(H^{(1)}, \dots, H^{(k)}, W_O)$, where 
each $H^{(i)}$ is a self-attention head
of I/H dimension $(d, d_h)$ and $W_O \in \R^{kd_h \times d}$ is an \textbf{output matrix}. Let $\cH(X)$ denote the 
concatenation of $H^{(1)}(X), \dots, H^{(k)}(X)$. 
Now, $\SA$ computes the matrix $\SA(X) \colonequals \cH(X)W_O$. 
For brevity, we may omit `self' from `self-attention'.

A \textbf{transformer layer} of dimension $d$ is a pair $(\SA, \FF)$, where $\SA$ is an attention module and $\FF$ is an $\MLP$ of dimension $d$. 
A \textbf{GPS-layer} of dimension $d$ is a tuple $(\SA, \MP, \FF)$, where $\MP$ is a message-passing layer.

A \textbf{graph transformer} ($\GT$) over $(\Pi, d)$ is a tuple $T = (P, L^{(1)}, \ldots, L^{(k)}, C)$, where $P$ and $C$ are as for GNNs and each $L^{(i)}$ is a transformer layer $(\SA^{(i)}, \FF^{(i)})$ of dimension~$d$. 
A \textbf{GPS-network} $N$ over $(\Pi, d)$ is defined like a $\GT$ except that each $L^{(i)}$ is a $\GPS$-layer
$(\SA^{(i)}, \MP^{(i)}, \FF^{(i)})$
of dimension $d$. 
This definition does not include the optional normalization layers. For more about normalization layers, see e.g. \citep{Rampasek}.
Analogously to a GNN, a $\GPS$-network $N$ computes over a graph $\cG$ a sequence of feature maps and a final feature map $N(\cG)$ as follows: $\lambda^{( 0 )} \colonequals P(\lambda)$,
\[
\begin{aligned}
    \lambda^{(i+1)}_{B} &\colonequals \lambda^{(i)} + B^{(i+1)}\big( \lambda^{( i)} \big), \text{ where } B \in \{\SA, \MP\}, \\
    \lambda^{(i+1)}_{\SA+\MP} &\colonequals \lambda^{(i+1)}_{\SA} + \lambda^{(i+1)}_{\MP}, \\
    \lambda^{(i+1)} &\colonequals   \lambda^{(i+1)}_{\SA+\MP} + \FF^{(i+1)} \big( \lambda^{(i+1)}_{\SA+\MP} \big).
\end{aligned}
\]
Finally, $N(\cG) \colonequals C(\lambda^{( k )})$.
A $\GT$ $T$ computes a feature map $T(\cG)$ analogously to $N$, but  without the modules $\MP^{(i)}$:
\[
\lambda^{(i+1)} \colonequals   \lambda^{(i+1)}_\SA + \FF^{(i+1)}\big(\lambda^{(i+1)}_\SA \big).
\]
Acceptance for GTs and GPS-networks is defined analogously to GNNs.

We focus on the two most commonly used attention functions.
For $\bx \in \R^\ell$, let $\cI_\bx = \{\, i \in [\ell] \mid \bx_i = \max(\bx) \,\}$ where $\max$ returns the largest entry in vector $\bx$.
We define the \textbf{average hard} ($\AH$) and \textbf{softmax} functions:
\begin{enumerate}
    \item\label{average hard} $\mathrm{AH}(\bx)_i \colonequals \frac{1}{\abs{\cI_\bx}}$ if $i \in \cI_\bx$ and $\mathrm{AH}(\bx)_i \colonequals 0$ otherwise,
    \item\label{softmax} $\softmax(\bx)_i \colonequals \frac{e^{x_i-b}}{\sum_{j \in [\ell]} e^{x_j-b}}$, where $b = \max(\bx)$.\footnote{This is also known as the “stable” or “safe” softmax due to its numerical stability \citep{blanchard2019accuratecomputationlogsumexpsoftmax}, in contrast to the version of $\softmax$ without the biases $-b$.} 
\end{enumerate}
\begin{example}
    For $\bx = (5,7,1,7) \in \R^4$,   
     $ \softmax(\bx) \approx (0.063, 0.468, 0.001, 0.468)$ and $\mathrm{AH}(\bx) = (0, \frac{1}{2}, 0, \frac{1}{2})$.
\end{example}
Attention heads that use
$\AH$ or $\softmax$
are called 
average hard-attention heads and soft-attention heads, respectively. The same naming applies to attention modules, transformer layers, graph transformers, $\GPS$-layers and $\GPS$-networks. 
Later in Example~\ref{ex:reals:relativecounting} in Section~\ref{sec:real-gps}, we show that there exists a soft-attention graph transformer that expresses the property: `at least half of the vertices in the studied graph have the label symbol $p$'.

\subsection{Logics}\label{section: logics}

We define the logics used in this paper.
Let $\Pi$ be a finite set of vertex label symbols. 
With a first-order (FO) formula $\varphi$ over $\Pi$, we mean a formula of first-order logic over the vocabulary that contains a unary relation symbol for each $p \in \Pi$ and a binary edge relation symbol $E$ (equality is included).
A \textbf{$\Pi$-formula of graded modal logic with the counting global modality} ($\GMLGC$)  is defined by the  grammar
$\varphi \coloncolonequals \top \mid p \mid \neg \varphi \mid \varphi \land \varphi \mid \Diamond_{\geq k} \varphi \mid \DiamondG_{\geq k} \varphi$,
where $p \in \Pi$ and $k \in \N$. 
We use $\lor$, $\rightarrow$ and $\leftrightarrow$ as abbreviations in the usual way, and for $D \in \{\Diamond, \DiamondG\}$, we define that $D_{< k} \varphi \colonequals \neg D_{\geq k} \varphi$ and $D_{= k} \varphi \colonequals D_{\geq k} \varphi \land D_{< k+1} \varphi$.

The semantics of $\GMLGC$ is defined over pointed graphs. 
In the field of modal logic, $\Pi$-labeled graphs are often called
Kripke models.
For a $\Pi$-formula $\varphi$ of $\GMLGC$ and a pointed
$\Pi$-labeled graph $(\cG, v)$, 
the truth of $\varphi$ in
$(\cG, v)$ (denoted by $\cG, v \models \varphi)$ is defined as follows. $\cG, v \models
\top$ holds always. For $p \in
\Pi$, $\cG, v \models p$ iff $p \in \lambda(v)$. The cases $\neg$ and $\land$ are defined in the
usual way. For diamonds,
$$
\begin{array}{rcl}
\cG, v \models \Diamond_{\geq k} \varphi& \!\!\text{iff}\!\! &\abs{\{ u \in \neigh_\cG(v) \mid \cG, u \models \varphi \}} \geq k \\
\cG, v \models \DiamondG_{\geq k} \varphi & \!\!\text{iff}\!\! & \abs{\{ u \in V(\Gmc) \mid \cG, u \models \varphi \}} \geq k.
\end{array}
$$

\textbf{Graded modal logic with global modality} ($\GMLG$) is the fragment of $\GMLGC$ where diamonds $\DiamondG_{\geq k}$ are allowed only if $k=1$. For simplicity, we let $ \DiamondG \colonequals \DiamondG_{\geq 1}$.
\textbf{Graded modal logic} ($\GML$) is the fragment of $\GMLG$ without diamonds $\DiamondG$. \textbf{Modal logic} ($\ML$) is the fragment of $\GML$ where we allow diamonds $\Diamond_{\geq k}$ only if $k=1$, and we let $\Diamond \colonequals \Diamond_{\geq 1}$.
\textbf{Propositional logic} ($\PL$) is the fragment of $\ML$ without diamonds $\Diamond$.
The logics 
$\MLGC$, $\MLG$, $\PLGC$ and $\PLG$ are  defined in the expected way.

\begin{example}\label{example: expressibility in logic}
    The property `no vertex is a dead-end' is expressed by the $\GMLGC$-formula $\DiamondG_{=0} \Diamond_{=0} \top$. 
    No $\GMLGC$-formula can express the property `at least half of the vertices in the graph have label $p$'. A formal proof via graded bisimulations is straightforward.
\end{example}

\subsection{Equivalence of Vertex Classifiers}\label{sec: equivalence}

A \textbf{(vertex) property over $\Pi$} is a mapping $f$ that assigns to each $\Pi$-labeled graph $\cG$ a feature map $\lambda' \colon V(\cG) \to \{0,1\}$ and is invariant under isomorphisms. 
A \textbf{vertex classifier} is
any object $C$ that defines a
vertex property.
Note that each of our computing models is a vertex classifier.
Each $\Pi$-formula $\varphi$ of any logic introduced above also corresponds to a vertex classifier (where for $\FO$, $\varphi$ must have a single free variable) which maps each $\Pi$-labeled graph $\cG$ 
to the feature map 
$\lambda_{\varphi}$ 
with $\lambda_\varphi(v) = 1$ if $\cG, v  \models \varphi$ and
$\lambda_\varphi(v) = 0$ otherwise.

 We say that vertex classifiers 
 $C_1$ and $C_2$ are  \textbf{equivalent} if they define the same vertex property.
Two classes $\cC_1$ and $\cC_2$ of vertex classifiers \textbf{have the same expressive power} if for each $C_1 \in \cC_1$ there is an equivalent $C_2 \in \cC_2$, and vice versa. We say that $\cC_1$ and $\cC_2$ \textbf{have the same expressive power relative to $\FO$},
if for each vertex property $f$ definable by a formula $\varphi(x) \in \FO$, there is some $C_1 \in \cC_1$ that defines $f$ if and only if there is some $C_2 \in \cC_2$ that defines $f$.

\section{Characterizing Real-Based Transformers} \label{sec:real-gps}

We provide characterizations of the expressive power of $\GPS$-networks
and of $\GT$s, over the reals and relative to FO, in terms of the
logics $\GMLG$ and $\PLG$, respectively. The characterizations apply to both
soft-attention
and average hard-attention.
We start with $\GPS$-networks.

\begin{theorem}\label{thm:real-GPS-GMLG}
Relative to $\FO$, the following have the same expressive power:
$\GMLG$, soft-attention $\GPS$-networks, and average hard-attention $\GPS$-networks.
\end{theorem}
We discuss Theorem~\ref{thm:real-GPS-GMLG} before giving a proof. An
interesting comparison is to the results of~\citet{Barcelo_GNNs}, who prove that relative
to $\FO$, GNNs without transformer layers and without
global readout have the same expressive power as 
$\GML$.
They also show that, when counting global readout is admitted, GNNs can express
all of $\GMLGC$.\footnote{
They actually show that GNNs with counting global readout capture all of $C_2$---the two-variable fragment of $\FO$ with counting quantifiers---but only on undirected graphs; this fails for directed graphs, as $\GNNGC$s cannot express, for instance, the $C_2$-formula $\exists y \, E(x,y) \wedge E(y,x)$.
} 
Relative to $\FO$, $\GPS$-networks thus sit properly in the middle between
$\GNN$s and $\GNNGC$s: they can express global
properties such as $P_1 =$ `the graph contains a vertex labeled~$p$', but cannot express absolute global counting,  such as $P_2 =$ `the graph contains at least 2 vertices
labeled~$p$'.

Let us also discuss  absolute expressive power, 
dropping $\FO$ as a background logic. $\GNN$s are by definition
a special case of $\GPS$-networks. Conversely, 
Property~$P_1$ witnesses  that
$\GPS$-networks are strictly more expressive than GNNs, also in an absolute sense. Likewise, Property~$P_2$
shows that some $\GNNGC$s do not have an equivalent GPS-network. It remains open whether every GPS-network is
equivalent to a $\GNNGC$. It is proved in 
\citep{grohe_transformers} that this is not the case for graph classification, but the result does not immediately transfer to vertex classification.

While $\GPS$-networks cannot express
global properties that involve absolute counting, 
they 
can express global properties with relative counting,
and so can
$\GT$s. 
This is not visible in Theorem~\ref{thm:real-GPS-GMLG} because
such properties do not fall within $\FO$.
We demonstrate relative counting in the example below.

\begin{example}
\label{ex:reals:relativecounting}
There is a $1$-layer soft-attention $\GT$ $(P, L, C)$ that accepts precisely the vertices of those graphs that satisfy the property: `at least half of the vertices in the graph have the label $p$'.
The initial $\MLP$ $P$  
maps each feature vector to a $2$-dimensional feature vector, where the first component encodes the labeling by $p$ and the other is $0$. 
The soft-attention module of $L$ then has 
$W_Q = W_K = [0, 0]^\T$, $W_V = [1, 0]^\T$ and $W_O = [0, 1]$ 
and 
the $\MLP$ outputs a zero matrix.
%,
After the last skip connection, the second column of the matrix consists of the value $x$, where $x$ tells the ratio of how many vertices have the label $p$; the final classification head $C$ outputs $1$ if $x \geq 0.5$ and $0$ otherwise.
\end{example}

We now survey the proof of Theorem~\ref{thm:real-GPS-GMLG}, full 
details are in 
Appendix~\ref{app:real-gps}.
The
easier direction is to show that every $\GMLG$-formula can be
translated into an equivalent $\GPS$-network.  We extend the
corresponding construction of \citep{Barcelo_GNNs} for $\GML$, using
self-attention heads to handle subformulae of the form~$\DiamondG \varphi$.
\begin{restatable}{lemma}{lemgmlgtorealgps}\label{lem:gmlg-to-real-gps}
  For every $\GMLG$-formula, there is an equivalent $\GPS$-network.
  This applies to both soft-attention and average hard-attention.
\end{restatable}
A notable difference to the proof of \citep{Barcelo_GNNs} is that we
use a step function as an activation function, rather than truncated
ReLU. Intuitively, this is because truth values are represented as 0 and 1 in
feature vectors, but both soft-attention and average hard-attention
may deliver an arbitrarily small (positive) fractional value and there
seems to be no way to `rectify' this into a 1 without using a step function. This is similar
to what happens in GNNs that
use arithmetic mean as the aggregation 
function \citep{AAAIMean}. 

The difficult direction in the proof of
Theorem~\ref{thm:real-GPS-GMLG} is to show that every $\GPS$-network that
expresses an $\FO$-property is equivalent to a $\GMLG$-formula. In
\citep{Barcelo_GNNs}, this direction is proved by first showing that
GNNs are invariant under graded bisimulation and then applying a van
Benthem/Rosen-style result from finite model theory~\citep{otto2019}
which says that every $\FO$-formula invariant under graded bisimulation
is equivalent to a GML-formula. $\GPS$-networks, however, are not
invariant under graded bisimulations because these do not preserve
global properties. We thus introduce a stronger version of graded
bisimilarity that also takes into account the multiplicities with
which graded bisimulation types are realized, and prove a
corresponding van Benthem/Rosen theorem.

Let $\Pi$ be a finite set of vertex label symbols.
A \textbf{graded bisimulation} between two $\Pi$-labeled graphs
$\cG_1 = (V_1, E_1, \lambda_1)$ and $\cG_2 = (V_2, E_2, \lambda_2)$
is a binary relation $Z \subseteq V_1 \times V_2$ that satisfies the following conditions:
\begin{description}
    \item[atom] for all $(v_1, v_2) \in Z$, $\lambda_1(v_1) = \lambda_2(v_2)$.
    
    \item[graded forth] for all pairs $(u_1, u_2) \in Z$, for all $k \geq 1$: for all
    pairwise distinct vertices $v_1, \ldots, v_k \in \neigh_{\cG_1}(u_1)$ there are
    pairwise distinct vertices $v'_1, \ldots, v'_k \in \neigh_{\cG_2}(u_2)$ with
    $(v_1, v'_1), \ldots, (v_k, v'_k) \in Z$. 
    \item[graded back] for all pairs $(u_1, u_2) \in Z$, for all $k \geq 1$: for all
    pairwise distinct vertices $v'_1, \ldots, v'_k \in \neigh_{\cG_2}(u_2)$ there are
    pairwise distinct $v_1, \ldots, v_k \in \neigh_{\cG_1}(u_1)$ with $(v_1, v'_1),
    \ldots, (v_k, v'_k) \in Z$.
\end{description}
We write $(\cG_1, v_1) \sim (\cG_2, v_2)$ if there is a graded bisimulation $Z$
between $\cG_1$ and $\cG_2$ with \mbox{$(v_1, v_2) \in Z$}.

A \textbf{graded bisimulation type over $\Pi$} is a maximal set $t$ of $\Pi$-labeled pointed graphs such
that $(\Gmc_1,v_1) \sim (\Gmc_2,v_2)$ for all $(\Gmc_1,v_1),(\Gmc_2,v_2) \in t$.
For a pointed $\Pi$-labeled
graph $(\Gmc,v)$, we use $\mn{tp}_\Gmc(v)$ to denote the unique graded
bisimulation type $t$ over $\Pi$ such that $(\Gmc,v) \in t$.

    Pointed graphs $(\Gmc_1,v_1)$, $(\Gmc_2,v_2)$ are \textbf{global-ratio graded bisimilar}, written $(\cG_1, v_1)
    \sim_{G\%} (\cG_2, v_2)$,
    if $(\cG_1, v_1)
    \sim
    (\cG_2, v_2)$ and there exists
    a rational number $q > 0$ such that for every graded bisimulation type $t$,
    $$
       |\{ v \in V_1 \mid \mn{tp}_{\Gmc_1}(v)=t\}|
       = q \cdot  |\{ v \in V_2 \mid \mn{tp}_{\Gmc_2}(v)=t\}|.
    $$
Note that the ratios between graded bisimulation types above are closely related to relative counting as
 in Example~\ref{ex:reals:relativecounting}.

\begin{example}\label{ex:bisim}
We illustrate that global-ratio graded bisimilarity provides a middle ground between graded
bisimilarity and isomorphism.

First consider the $\{p\}$-labeled graphs $\cG_1 = ( \{v_1\}, \emptyset, \lambda_1)$ such that $\lambda_1(v_1) = \{p\}$ and $\cG_2 = ( \{ u_1, u_2\}, \emptyset, \lambda_2)$ such that $\lambda_2(u_1) = \{p\}$ and $\lambda_2(u_2) = \emptyset$. Then $(\cG_1, v_1) \sim (\cG_2, u_1)$, but $(\cG_1, v_1) \not \sim_{G\%} (\cG_2, u_1)$ as the graded bisimulation type of $(\cG_2, u_2)$ does not occur in $\cG_1$.

Now consider the graph $\cG_3 = (\{ w_1, w_2 \}, \emptyset, \lambda_3)$ with $\lambda_3(w_1) = \lambda_3(w_2) = \{p\}$. Then $(\cG_1, v_1) \sim_{G\%} (\cG_3, w_1)$, taking $q = \frac{1}{2}$, but $\cG_1$ and $\cG_3$ are not isomorphic.
\end{example}

A vertex classifier such as a $\GPS$-network or an $\FO$-formula $\varphi(x)$ is \textbf{invariant under $\sim_{G\%}$} if for all pointed
graphs $(\cG_1, v_1)$ and $(\cG_2, v_2)$, 
$(\cG_1, v_1) \sim_{G\%} (\cG_2, v_2)$ implies that
$\cG_1 \models \varphi(v_1)$ if and 
only if $\cG_2 \models \varphi(v_2)$.
A layer-by-layer analysis of $\GPS$-networks shows the following.
\begin{restatable}{lemma}{lemgraphtransformerinvariance}\label{lem:gps-invariant-under-grg-bisim}
    Let $N$ be a soft-attention or average hard-attention $\GPS$-network.
    Then $N$ is invariant under $\sim_{G\%}$.
\end{restatable}
It follows that GPS-networks  cannot distinguish 
$(G_1, v_1)$ and $(G_3, w_1)$ from Example~\ref{ex:bisim}. 
In contrast, $(G_1, v_1)$ and $(G_2, u_1)$ can be distinguished: by Lemma~\ref{lem:gmlg-to-real-gps}, 
one can employ a GPS-network equivalent to the $\GMLG$-formula $\DiamondG \neg p$.

We now give the announced van Benthem/Rosen theorem.\\[-4mm]
\begin{restatable}{theorem}{lemratiobisiminvariance}
    \label{lem:ratio-bisim-invariance}
    For every $\FO$-formula $\varphi(x)$ over $\Pi$, the following are equivalent:
    \begin{enumerate}
        \item $\varphi$ is invariant under $\sim_{G\%}$;
        \item $\varphi$ is equivalent to a $\GMLG$-formula over all
          (finite!) $\Pi$-labeled pointed graphs.
        \end{enumerate}
\end{restatable}
The direction ``$2 \Rightarrow 1$'' of
Theorem~\ref{lem:ratio-bisim-invariance} is easy to prove by a
straightforward induction on the structure of $\GMLG$-formulae. The difficult part is  to show the  ``$1 \Rightarrow 2$'' direction, that is, every FO
formula $\varphi(x)$  invariant under $\sim_{G\%}$ is
equivalent to a $\GMLG$-formula. 
To achieve this, we combine and
extend techniques from~\citep{otto2004} and~\citep{otto2019}.  The
former provides a van Benthem/Rosen theorem that links global (ungraded)
bisimulation to $\MLG$, 
and the latter a theorem of the same
kind that links graded bisimulation to $\GML$. 

Our proof consists
of a sequence of results saying that if an $\FO$-formula $\varphi(x)$ is
invariant under $\sim_{G\%}$, then it is also invariant under certain
other notions of bisimulation that become increasingly
weaker.\footnote{While this provides a good intuition, it is not
  strictly true. For technical reasons, the intermediate notions of
  bisimulation sometimes get stronger in certain respects,
  e.g.\ they may use up-and-down features as used for modal logics with
  the converse modality.}
  We finally arrive at a notion of bisimulation that is called global $c$-graded $\ell$-bisimulation, where $c$ is a counting bound and $\ell$ a depth bound, both derived from $\varphi(x)$.
  Importantly, this notion of bisimulation has only a finite number of bisimulation types,
  and each type can be distinguished from
  the others using a characteristic $\GMLG$-formula.
  We can thus 
  construct the desired $\GMLG$-formula by
  taking the disjunction of all characteristic formulae for bisimulation types in which  $\varphi(x)$ is true.
To make sure that the ratio-property of $\sim_{G\%}$ is respected,
we replace several constructions from~\citep{otto2004} with
more careful ones.

Combining Theorem~\ref{lem:ratio-bisim-invariance} and
Lemma~\ref{lem:gps-invariant-under-grg-bisim} completes the proof sketch of
Theorem~\ref{thm:real-GPS-GMLG}.
Moreover,
as a special case, the construction in the proof of
Lemma~\ref{lem:gmlg-to-real-gps} shows that $\PLG$-formulae can be translated
into $\GT$s with both soft-attention and average hard-attention. A minor extension of our
techniques used to prove Theorem~\ref{lem:ratio-bisim-invariance}, then shows the
following.

\begin{restatable}{theorem}{thmrealgtplg}\label{thm:real-GT-PLG}
Relative to $\FO$, the following have the same expressive power:
$\PLG$, soft-attention $\GT$s, average hard-attention $\GT$s.
\end{restatable}

\section{Characterizing Float-Based Transformers}\label{Characterizing float-based transformers}

We give characterizations of $\GPS$-networks and $\GT$s based on floating-point numbers
via the logics $\GMLGC$ and $\PLGC$, respectively. 
Before the characterizations, we introduce floats and float-based $\GPS$-networks and $\GT$s.

\subsection{Floating-Point Numbers and Arithmetic}

We define the concepts of floating-point numbers based on the IEEE 754 standard \citep{IEEE754-standard}.
Let $p, q \in \Z_+$.
A \textbf{floating-point number} (over $p$ and $q$) is a string of the form 
\[
b_0 b_1 \cdots b_{p + q} \in \{0,1\}^{p+q+1}.
\]
The bit $b_0$ is called the \textbf{sign}, the string $\be = b_1 \cdots b_q$ the \textbf{exponent} and $\bs = b_{q+1} \cdots b_{p+q}$ the \textbf{significand}.
Let $a = 2^{p-1}$, $b = 2^{q-1}-1$, 
and let $e$ and $s$ be the non-negative integers represented in binary by $\be$ and~$\bs$. Then the above floating-point number is interpreted as the real number 
$
(-1)^{b_0}\frac{s}{a} 2^{e -b}.
$
As an exception, the float with $\bs = 0^p$, $\be = 1^q$ and $b_0=0$ (resp. $b_0 = 1$) corresponds to $\infty$ (resp. $-\infty$).  
When the context is clear, we identify a float with the real number (or $\infty$, $-\infty$) that it represents.
A float is \textbf{normalized} if $b_{q+1} \neq 0$, and \textbf{subnormalized} if $b_{q+1} = 0$ and $\be = 0^q$. 
A \textbf{floating-point format} $\cF(p, q)$ over $p$ and $q$ consists of all normalized and subnormalized floating-point numbers over $p$ and $q$ and the symbols $\infty$, $-\infty$, and $\NaN$ (`not-a-number'), and when clear we may write $\cF$ instead of $\cF(p,q)$.
Next, we discuss basic arithmetic operations over floating-point formats:
addition~$+$, subtraction~$-$, multiplication~$\cdot$, division~$\div$ 
and square root~$\sqrt{x}$.
The definition of each of the operations is to first "compute" to unlimited precision in real arithmetic (extended with $\infty$ and $-\infty$) 
and then rounding to the nearest float in the format, with ties rounding to the float with an even least significant bit.
Undefined results, such as $\frac{\infty}{\infty}$, are mapped to $\NaN$.
If any input is $\NaN$, the output is $\NaN$, i.e., our $\NaN$ is \emph{silent} and propagated through the computation. 
The exponential function $\exp(x)$ over floats is not a basic operation and is implemented in a standard way, using basic operations, range reductions and polynomial approximations. 
Background and a discussion on these concepts is in
Appendix~\ref{appendix: floats}.

\subsection{Float-Based Transformers}\label{section: Transformers with floats}

We introduce float-based $\GT$s, $\GPS$-networks and $\GNN$s. To define them, we replace reals with floats, but we must also carefully specify how float operations are performed. One 
reason is that many float operations (e.g. sum) are not associative due to rounding errors between operations. Thus, switching the order of operations can affect the outcome.
\begin{example}
    Consider the sum of the real numbers $-1$, $1$ and $4$ representable in the format $\cF(2,3)$: here we have $(-1 + 1) + 4 = 0 + 4 = 4$ but $-1 + (1 + 4) = -1 + 4 = 3$. Note that in the latter equation, the precise sum of $1$ and $4$ would be $5$, which is not representable in the format $\cF(2,3)$ and is thereby rounded to the nearest number; both $4$ and $6$ are equally near, and $4$ has the even least significant bit.
\end{example}
The $\softmax$ function, the sum aggregation function and some matrix multiplications in attention heads take a sum over the features of vertices in the studied graph, and are thus affected by this non-associativity issue.
In the worst case, this can violate the isomorphism invariance of these learning models, which is undesirable. 
For example, in typical real-life implementations, the set $V$ of vertices in the studied graph is associated with some implementation-related, implicit linear order $<^V$ (that is not part of the actual graph). 
Then isomorphism invariance can be violated if the sum aggregation sums in the order $<^V$.
Hence, it is better to order the floats instead of the vertices.
We make the natural assumption that floats are always summed in increasing order, which results in models that are 
isomorphism invariant. This is further justified by numerical stability  \citep{wilkinson, floats_robertazzi, higham_article}.

Given a floating-point format $\cF$, we let $\SUM_{\cF}$ 
denote the operation that maps a multiset $N$ of floats to the sum $f_1 + \cdots + f_{\ell}$ where each $f_i$ appears $N(f_i)$ times and the floats appear and are summed in increasing order.  We recall from \citep{ahvonen_neurips} the following important result on \emph{boundedness} of float sums.
\begin{proposition}\label{proposition: floating-point saturation}
    For all floating-point formats $\cF$, there exists a $k \in \N$ such that for all multisets $M$ over floats in $\cF$, we have $\SUM_\cF(M) = \SUM_\cF(M_{|k})$.
\end{proposition} 
To see that the above proposition holds, observe that $\SUM_\cF$ repeatedly adds the same float the number of times it appears in the sum.
In the format $\cF(2, 2)$, the number $\frac{1}{2} \cdot 2^{-1} = 0.25$ is exactly representable.
Summing $0.25$ repeatedly in this format gives $0.25$ after one addition, $0.50$ after two additions, $0.75$ after three additions, $1.0$ after four additions, and $1.0$ after five additions, since $1.25$ is not in $\cF(2,2)$ and rounds to $1.0$. 
Thus, the sum ``saturates beyond a threshold''. Using this phenomenon, it is easy to obtain a proof of Proposition~\ref{proposition: floating-point saturation}

We say that an aggregation function $\AGG$ is \textbf{bounded} 
if there exists a $k \in \N$ such that $\AGG(M) = \AGG(M_{|k})$ for all $M$.
Apart from sum, also mean aggregation is similarly bounded. 
Furthermore, we assume that $\softmax$ is implemented for a floating-point format $\cF$ by using the above sum $\SUM_\cF$ in the denominator, and the remaining operations are carried out in the natural order, i.e.,
we first calculate the bias $b$, then values $x_j - b$, then the exponents, 
and finally the division.
Likewise, we assume $\AH$ is implemented for $\cF$ by
calculating the denominator in $\frac{1}{\abs{\cI_\bx}}$ using the same approach as
\citep{li2025characterizingexpressivitytransformerlanguage}, i.e.,
calculating it as $\SUM_\cF ( M )$, where $M$ is the multiset over $\cF$ consisting of precisely $\abs{\cI_\bx}$ instances of the float $1$,
and then performing the division.

\subsubsection{Float-Based Learning Models.}

\emph{Floating-point $\GT$s, denoted by $\GTF$,  are defined in the same way as $\GT$s based on reals, except that they use floats in feature vectors and float operations where the order of operations is as specified above. Likewise for $\GPS$-networks, $\GNN$s, $\MLP$s, etc.
We further assume that these models always use aggregation functions that are bounded. This is a natural assumption as sum, max and mean are all bounded on floats by the above findings.}

We call these learning models \textbf{simple} when the $\MLP$s are simple\footnote{An exception is the final Boolean vertex classifier, which is otherwise a simple $\MLP$ but uses the Heaviside function.}
and the aggregation functions are $\SUM_{\cF}$. In fact, $\GT$s and $\GPS$-networks were originally defined based on simple $\MLP$s \citep{Dwivedi,Rampasek}. 
We do not fix a single float format for all $\GTF$s, $\GPSF$-networks, $\GNNF$s, etc.; instead, each of them is associated with \emph{some} float format. In our translations, the format is assumed arbitrary when translating them into logics, but can be chosen freely in the other direction.

\subsection{Characterizations}\label{sect: float characterization}

Next, we provide logical characterizations for $\GTF$s and $\GPSF$-networks with both soft and average hard-attention. The characterizations are absolute, i.e., they do not require relativizing to a background logic such as FO. Our float-based $\GTF$s and $\GPSF$-networks also do not require step function activated $\MLP$s aside from the classification heads.

First, we make an observation about float-based multiplication relevant to our translation techniques. When multiplying two floating-point numbers that are very close to zero, \textbf{underflow} occurs: the exact result is so small that all significant bits are lost, and the output is $0$. For instance, underflow can occur in attention heads in some matrix multiplications. 
The following proposition demonstrates this phenomenon. In the proposition and the proof that follows, we identify each float with the real number that it represents.

\begin{restatable}{proposition}{floatunderflow}\label{proposition: floating-point underflow 2}
    Let $\cF$ be a floating-point format, let $f$ be the smallest positive float in $\cF$ and let $k$ be some even integer such that $\frac{k}{2}$ is accurately representable in $\cF$. For all $F \in \cF$, $\abs{F} \leq \abs{\frac{1}{k}}$ if and only if $F \cdot (\frac{k}{2} f) = 0$.
\end{restatable}

For the proof, note that since $\frac{k}{2}$ is accurately representable in $\cF$, 
then so is also the precise product $\frac{k}{2}f$.
Now, we observe that $\abs{F} \leq \abs{\frac{1}{k}}$ if and only if the precise product $F \cdot (\frac{k}{2} f)$ belongs to the closed interval $[-\frac{1}{2}f, \frac{1}{2}f]$. 
All numbers in this interval round to $0$.

Now, we give our logical characterization for $\GTF$s. Recall that float-based computing models by definition use bounded aggregation functions, and as explained, this is a natural assumption.
By `constant local aggregation functions', 
we intuitively mean that in message-passing layers, vertices cannot distinguish if a message was received from an out-neighbor or from any other vertex.

\begin{restatable}{theorem}{thmPLGCSGTAHGT}\label{theorem: PLGC = SGT = AHGT}
    The following have the same expressive power: 
        $\PLGC$, 
        soft-attention $\GTF$s and
        average hard-attention $\GTF$s (and
        $\GNNGCF$s with constant local aggregation functions).
        This also holds when the $\GTF$s and $\GNNGCF$s are simple.
\end{restatable}

    We provide a more detailed proof for Theorem~\ref{theorem: PLGC = SGT = AHGT} in 
    Appendix~\ref{Appendix: PLGC = SGT = AHGT},
    but we \emph{sketch} the proof here.
    In the direction from $\GTF$s to logic, 
    the general idea is that for each vertex $v$ we simulate its feature vector $\bx_v$ after each transformer layer by simulating each bit of $\bx_v$ by a single formula, i.e., a sequence of formulae simulates the whole vector $\bx_v$. As a last step, we combine these formulae recursively into a single formula that simulates the output of the classification head.
    There are two key insights for simulating bits. 
    First, each `local step' of a $\GTF$ where a vertex does not need to know the features of any other vertices (e.g. $\MLP$s or matrix products $XW_Q$, $XW_K$ and $XW_V$) can be expressed as a function $f_\cF \colon \cF^n \to \cF^m$.
    As floats are bit strings, we can identify $f_\cF$ with a \emph{partial} function $f_{\mathbb{B}} \colon \{0,1\}^{kn} \to \{0,1\}^{km}$, where $k$ is the number of bits in $\cF$.
    $\PL$ 
    is expressively complete for expressing Boolean combinations, i.e., 
    each function $g \colon \{0,1\}^n \to \{0,1\}$ has an equivalent $\PL$-formula as $g(\bx)$ is simply a Boolean combination of the values in $\bx$. 
    Thus, we can construct an equivalent $\PL$-formula for each output bit of $f_{\mathbb{B}}$, and the full function $f_{\mathbb{B}}$ can be simulated by a sequence of formulae.
    Second, for the remaining `non-local' steps, it suffices to know the features of other vertices in the `global sense', i.e., the edges of the graph are not used. 
    Due to Proposition~\ref{proposition: floating-point saturation}, 
    the float sums appearing in attention heads are bounded for some $k$, i.e., after $k$ copies of a float $F$, further instances of $F$ do not affect the sum. 
    Since the attention heads sum over the features of all vertices,
    it suffices for a vertex to be able to distinguish a bounded number of each possible feature vector appearing in the graph, and we can count up to this bound with the counting global modality.

    For the converse, to translate a $\PLGC$-formula $\varphi$ into a simple $\GTF$, we use a similar strategy as with reals: we compute the truth values of the subformulae of $\varphi$ one at a time, using multiple transformer layers per subformula.
    Again the truth value of each subformula is represented as $0$ or $1$ in the feature vector for each vertex, i.e., in the feature matrix there is a column for each subformula that is used to encode the truth value of the subformula in each vertex. There are also some auxiliary columns in the feature matrix.
    The translation involves making use of the properties of floats and floating-point operations.
    The operators $\neg$ and $\land$ are easy to handle by using the $\MLP$s of the transformer layers.
    The hardest part is to simulate modalities $\DiamondG_{\geq k}$ by using $\MLP$s \emph{and} attention modules. 
    Assume that we are computing the truth value of a subformula of the form $\DiamondG_{\geq k} \psi$ and the truth value of $\psi$ has already been encoded into a column $i$ of the current feature matrix. Then we build an attention head that checks if the number $\ell$ of $1$s in the column $i$ is at least $k$. Intuitively, we construct a query matrix $W_Q$ and a key matrix $W_K$ such that the matrix $\softmax\big( (X W_Q) (X W_K)^{\T} /\sqrt{d_h} \big)$ has an entry in every row that has the value $\frac{1}{\ell}$ (or a suitable rounded value). A similar construction is also possible by using the average-hard function. Then due to Proposition~\ref{proposition: floating-point underflow 2}, we can construct a value matrix which uses underflow to check if $\frac{1}{\ell} \leq \frac{1}{k}$. After that, by using MLPs, we can distinguish when $\ell \geq k$ and when $\ell < k$. 
    This completes the proof sketch.

Before characterizing $\GPSF$-networks, we prove a helpful characterization of float-based $\GNN$s.

\begin{restatable}{theorem}{thmGNNGML}\label{theorem: GNN GML}
    The following pairs have the same expressive power (denoted by $\equiv$): 
    \[
    \GNNF \equiv \GML, \GNNGF \equiv \GMLG \text{ and } \GNNGCF \equiv \GMLGC.
    \]
    This also holds when each type of $\GNNF$ is simple.
\end{restatable}
    This theorem follows from the proof techniques of Theorem 3.2 of \citep{ahvonen_neurips}, which showed that (simple) recurrent float $\GNN$s are equally expressive as a recursive rule-based bisimulation invariant logic called the graded modal substitution calculus ($\mathrm{GMSC}$). 
    Unlike the float $\GNN$s in that paper, our $\GNNF$s are not recurrent, meaning they only scan the neighborhood of a vertex up to some fixed depth. The corresponding constant-iteration fragment of $\GMSC$ is $\GML$. The techniques in \citep{ahvonen_neurips} generalize for global readouts and modalities, see  
    Appendix~\ref{appendix: GML = GNNs}
    for the
    technical details of the proof.

We now characterize float-based $\GPS$-networks.

\begin{restatable}{theorem}{thmGMLGCSGPSAHGPSGNNG}\label{theorem: GMLGC = SGPS = AHGPS = GNNG}
    The following have the same expressive power: 
        $\GMLGC$, 
        soft-attention $\GPSF$-networks, 
        average hard-attention $\GPSF$-networks and 
        $\GNNGCF$s.
    This also holds in the case where the $\GPSF$-networks and $\GNNGCF$s are simple.
\end{restatable}
    The result follows from Theorems \ref{theorem: PLGC = SGT = AHGT} and \ref{theorem: GNN GML}. 
    Importantly, any transformer layer and message-passing layer can be simulated by a $\GPS$-layer of a higher dimension by appending the inputs and outputs of the transformer and message-passing layer with zeros on the $\GPS$ side.
    The technical details of the proof are in 
    Appendix~\ref{Appendix: GMLGC = SGPS = AHGPS = GNNG}.

We make some final observations.
As seen in Example~\ref{ex:reals:relativecounting}, `relative global counting' is expressible by real-based $\GT$s. However, 
the same construction does not work for $\GTF$s as, due to Proposition~\ref{proposition: floating-point saturation}, the softmax-function and average hard function lose accuracy in a drastic way with large graphs.
For the same reason, Lemma~\ref{lem:gps-invariant-under-grg-bisim} fails with floats; $\GTF$s and $\GPSF$-networks are not invariant under the bisimilarity $\sim_{G\%}$.
However, Theorems \ref{theorem: PLGC = SGT = AHGT}, \ref{theorem: GNN GML} and \ref{theorem: GMLGC = SGPS = AHGPS = GNNG} show that with float-based $\GT$s, $\GPS$-networks and $\GNN$s, `absolute counting' is possible (locally or globally depending on the model), since the matching logics can count.

We note that our results with floats immediately hold when restricted to word-shaped graphs, i.e., graphs where the domain is a prefix $[n]$ of positive integers, the edge relation is the successor relation over $[n]$, and $\lambda_\bw(v)$ is a singleton for each vertex. For example, a $\GT$ over word-shaped graphs is just an `encoder-only transformer without causal masking'. A popular example is BERT \citep{devlin2019bertpretrainingdeepbidirectional}, which is such a model inspired by \citep{NIPS2017_3f5ee243}. In 
Appendix~\ref{appendix: Transformers on Words},
we study unique hard-attention graph transformers over word-shaped graphs.
In Appendix~\ref{appendix: graph classification}, we also consider generalizations of our float results for graph classification tasks and non-Boolean classification tasks.

We also briefly discuss how our float results could be modified to cover positional encodings. Often, each $\GNNF$, $\GTF$, or $\GPSF$-network $A = (P, L^{(1)}, \ldots, L^{(k)}, C)$ (with input dimension $\ell$) is associated with a \textbf{positional encoding} (or $\mathrm{PE}$) $\pi$ over $\cF$, i.e., a mapping that assigns to each graph $\cG$ a function $\pi(\cG) \colon V(\cG) \to \cF^\ell$. 
For example, a popular $\mathrm{PE}$ is LapPE \citep{Rampasek}.
Now, $A$ with $\pi$ computes over $\cG$ a sequence of feature maps similarly to $A$ without $\pi$ (see Section~\ref{sec: GTs and GNNs}), but for each vertex $v$ in $\cG$, we define $\lambda^{(0)}_v \colonequals P(\lambda)_v + \pi(\cG)_v$. 
Our characterizations with floats can be modified to cover PEs
by simply adding proposition symbols to the logic that encode the PE of each vertex, see 
Appendix~\ref{appendix: Positional Encodings}
for more details.
Positional encodings are an important future topic that warrants further study.

\section{Conclusion}
\label{sect:concl}

We have given logical characterizations for GPS-networks and graph transformers,  based on reals and on floats. 
As future work, it would be interesting to lift all our characterizations from vertex to graph classification, and to more comprehensively study the expressive power of $\GPS$-networks and $\GT$s enriched with common forms of positional encodings such as graph Laplacians.
Our results in the float case in fact already lift to graph classification tasks and also to non-Boolean classification, and they also hold when restricted to word-shaped graphs; 
this is covered in 
Appendices~\ref{appendix: graph classification} and \ref{appendix: PEs and words}.
Another interesting
open question is whether, in the case of the
reals, every $\GPS$-network can be expressed as a $\GNNGC$.

\section{Acknowledgments}
Veeti Ahvonen was supported by the Vilho, Yrjö and Kalle Väisälä Foundation.
Damian Heiman was supported by the Magnus Ehrnrooth Foundation.
Antti Kuusisto was supported by the project \emph{Perspectives on computational logic}, funded by the Research Council of Finland, project number 369424. Carsten Lutz was supported by  DFG project LU 1417/4-1.

\bibliography{literature}

\appendix

\section{Preliminaries}

\subsection{Word-shaped graphs}\label{appendix: words}

For a word $\bw = w_1 \cdots w_n \in \Pi^+$, the \textbf{word-shaped} graph of $\bw$ is $\cG_\bw = (V_\bw, E_\bw, \lambda_\bw)$ where 
$V_\bw = [n]$, $E_\bw$ is the successor relation over $[n]$, and 
$\lambda_\bw(i) = \{w_i\}$ for all $i \in V_\bw$.

When graph transformers are restricted to word-shaped graphs, a $\GT$ becomes an ordinary ``encoder-only transformer without causal masking''. For example, the popular BERT \citep{devlin2019bertpretrainingdeepbidirectional} inspired by \citep{NIPS2017_3f5ee243} is such a model. 

\subsection{Transformers with unique hard-attention}\label{appendix: Transformers with unique hard-attention}

In Appendices~\ref{appendix: Proofs for float section} and \ref{appendix: PEs and words}, we also study transformers that use unique hard-attention instead of soft-attention or average hard-attention.
Given $\bx \in \R^p$, we let $\cI_\bx = \{\, i \in [p] \mid \bx_i = \max(\bx) \,\}$. The unique hard function $\UH \colon \R^+ \lpto \R^+$ is defined by
\[
\mathrm{UH}(\bx)_i \colonequals 
\begin{cases}
    1, &\text{if $i \in \min(\cI_\bx)$} \\
    0, &\text{otherwise}.
\end{cases}
\]
For example, given $\bx = (-1, 5, 10, 0, 10, 5)$, we have $\mathrm{UH}(\bx) = (0, 0, 1, 0, 0, 0)$.
For the implementation of $\UH$ with floats over a floating-point format $\cF$, we assume the function is exactly the same, i.e., the leftmost component with the maximum value $\max(\bx)$ is given the floating-point value $1 \in \cF$ and others are given $0 \in \cF$.
Analogously to the average hard function and the softmax function, attention heads that use the unique hard function are called unique hard-attention heads. The same naming applies to modules, graph transformers and $\GPS$-networks.

Based on the literature, graph transformers and $\GPS$-networks that use unique hard-attention are not really used in real-life applications. One of the reasons is that they are not invariant under isomorphism.

\begin{proposition}
    There is a unique hard-attention $\GT$ (and $\GTF$) and a unique hard-attention $\GPS$-network (and $\GPSF$-network) that is not isomorphism invariant. 
\end{proposition}
\begin{proof}
    Let $\Pi = \{p\}$.
    Consider the following isomorphic $\Pi$-labeled graphs:
    \begin{itemize}
        \item $\cG_1 = ([2], \emptyset, \lambda_1)$, where 
        $\lambda_1(1) = \{p\}$ and $\lambda_1(2) = \emptyset$,
        \item $\cG_2 = ([2], \emptyset, \lambda_2)$, where 
        $\lambda_2(1) = \emptyset$ and $\lambda_2(2) = \{p\}$.
    \end{itemize}
    Now, consider a unique hard-attention head $H$ with $W_Q = W_K = [0]$ and $W_V = [1]$. Now, $H(X) = X$ (with reals and floats). By using $H$, it is easy to design a $\GT$ $G$ such that $G(\cG_1) = \lambda_1$ and $G(\cG_2) = \lambda_2$ even though $\cG_1$ and $\cG_2$ are isomorphic. An analogous result applies for $\GTF$s, $\GPS$-networks and $\GPSF$-networks.
\end{proof}
However, over word-shaped graphs (or simply over words), transformers based on unique hard-attention preserve isomorphism invariance, and their theoretical properties, such as expressive power, have been studied, for example, in \citep{HaoAF22, BarceloKLP24, Yang0A24}. Moreover, the unique hard-attention function could be used with vertex-ordered graphs (but note that we do not consider such graphs in this paper).

\subsection{Graph transformers and graph neural networks}

In the main body of this paper (i.e., outside of the appendix), we defined our graph neural networks and graph transformers as Boolean vertex classifiers. In a more general sense, we could allow our computing models to have an arbitrary (finite) output dimension. A graph transformer $(P, L^{(1)}, \ldots, L^{(k)}, C)$ \textbf{of I/H/O dimension $(p, d, q)$} (over the set $\Pi$ of vertex label symbols) refers to a graph transformer where $P$ is an $\MLP$ of I/O dimension $(p,d)$ and each $L^{(i)}$ is a transformer layer $d$ as defined in the main body of this paper and $C$ is an $\MLP$ (possibly not ReLU-activated) of I/O dimension $(d, q)$. Analogously, we define $\GPS$-networks \textbf{of I/H/O dimension $(p, d, q)$} over $\Pi$ and $\GNN$s \textbf{of I/H/O dimension $(p, d, q)$} over $\Pi$. For brevity, we can exclude $\Pi$ and leave it implicit when it is clear from the context. For example, when comparing two $\GPS$-networks with the same input dimension, we can assume that they are defined over the same set of vertex label symbols. 

\section{Proofs for Section~\ref{sec:real-gps}}
\label{app:real-gps}

\subsection{Proof of Lemma~\ref{lem:gmlg-to-real-gps}}

A \textbf{basic $\GPS$-layer} of dimension $d$ is a tuple $G = (\sigma, A, C, b,
H)$, where $\sigma \colon \R \to \R$ is an activation function, $A, C \in \R^{d \times d}$
are matrices, $b$ is a $d \times 1$ matrix, and $H$ is a
self-attention head of I/H dimension $(d, d)$.
Analogously to $\GPS$-layers as defined in Section~\ref{section: preliminaries}, $G$ computes the output feature map $\lambda'$ based on
an input feature map $\lambda$ as follows:
for each $v \in V$,
\[
    \lambda'(v) \colonequals \sigma \Big( \lambda(v) C + \big(\sum_{\mathclap{v' \in \neigh_{\cG}(v)}} \lambda(v') \big) A + H(\lambda)_{v, *} + b \Big).
\]
A \textbf{basic} $\GPS$-network is a tuple $(P, L^{(1)}, \ldots, L^{(k)}, C)$ that is defined like a GPS-network, except all layers are basic GPS-layers. Below, we assume that both $\GPS$-networks are defined over the same set $\Pi$ of vertex label symbols.

\begin{lemma}\label{lem:basic-gps-to-gps}
    For every basic $\GPS$-network $G$ of I/H/O dimension $(p, d, 1)$,
    there is a (non-basic) $\GPS$-network $\hat G$ of I/H/O dimension $(p, 2d, 1)$
    such that for all labeled graphs $\cG$, $G(\cG) = \hat G(\cG)$.
\end{lemma}
\begin{proof}
    Let $G = (P, L^{(1)}, \ldots, L^{(k)}, C)$ be a basic $\GPS$-network of I/H/O
    dimension $(p, d, 1)$.
    In constructing an equivalent non-basic $\GPS$-network $\hat G$, the main challenge is
    dealing with the skip-connections. For this, we use $d$ additional  
    hidden dimensions and maintain that the feature maps computed by $\hat G$
    contain the feature maps computed by $G$ in the first $d$ dimensions, and  0 in the other $d$ dimensions.
    The equivalent non-basic $\GPS$-network is then $\hat G = (\hat P, \hat L^{(1)}, \ldots, \hat L^{(k)}, \hat C)$,
    where 
    \begin{itemize}
        \item $\hat P$ is obtained from $P$ by adding $d$ output dimensions that have the value $0$,
        \item $\hat L^{(i)}$ is obtained from $L^{(i)}$ by
            \begin{itemize}
                \item Constructing $\mathrm{MP}^{(i)}$ to output $0$ in the first $d$ dimensions
                and also to output 
                $\lambda^{(i)}(v) C + \sum_{v' \in \neigh_{\cG}(v)} \lambda^{(i)}(v') A + b$
                in the second $d$ dimensions, by choosing sum as the aggregation and combination functions.
                \item Constructing $\mathrm{SA}^{(i)}$ to output $0$ in the first $d$ dimensions and
                $H(\lambda^{(i)})_v$ in the second $d$ dimensions by using $H$ as a single attention head.
                \item Constructing $\mathrm{FF}^{(i)}$ such that, on input
                $\mathbf{x}_1$ in the first $d$ dimension and $\mathbf{x}_2$ in
                the second $d$ dimensions, outputs $\sigma(\mathbf{x}_2) -
                \mathbf{x}_1$ in the first $d$ dimensions and $0$ in the second
                $d$ dimensions.
            \end{itemize}
        \item $\hat C$ is obtained from $C$ by adding $d$ input dimensions, which are ignored. \qedhere
    \end{itemize}
\end{proof}

\lemgmlgtorealgps*

\begin{proof}

Let $\Pi$ be a finite set of vertex label symbols and let $\varphi$ be a $\GMLG$
formula over $\Pi$. For constructing a $\GPS$-network that is equivalent to
$\varphi$, we extend the construction of \citep{Barcelo_GNNs} to $\GMLG$, which
will result in a basic $\GPS$-network that uses step-activation. Lemma~\ref{lem:basic-gps-to-gps} then
shows that an equivalent $\GPS$-network exists.

Let $\varphi_1, \ldots, \varphi_d$ be the subformulae of $\varphi$ ordered such that
if $\varphi_i$ is a subformula of $\varphi_j$, then $i \leq j$. Hence, $\varphi_d =
\varphi$. Let $p = |\Pi|$. We use a $\GPS$-network $N = (P, L^{(1)}, \ldots, L^{(d)},
C)$ over $(\Pi, d)$, where $L^{(1)}, \ldots, L^{(d)}$ are $d$ basic $\GPS$-layers. 
In fact, all basic $\GPS$-layers of $N$ will be exactly identical.\footnote{A GNN with
all layers identical is called \emph{homogeneous} in~\citep{Barcelo_GNNs}.} We
aim to achieve that, for every pointed $\Pi$-labeled graph $(\cG, v)$ with $\cG = (V, E, \lambda)$,
on which $N$ computes feature maps $\lambda^{(0)}, \ldots, \lambda^{(d)}$,
\begin{itemize}
 \item[($*$)] $\lambda^{(i)}(u)_j = 1$ if $\cG, u \models \varphi_j$ and
 $\lambda^{(i)}(u)_j = 0$ otherwise, for all $i,j$ with $1 \leq j \leq i \leq d$
 and all $u \in V$.
\end{itemize}
The final classification layer $C$ will then accept only vertices $u$ with $\lambda^{(d)}(u)_d = 1$.

The initial MLP $P$ is of dimension $(p, d)$ and is a projection such that, for
every vertex label symbol $s \in \Pi$, $1 \leq j \leq d$ and $u \in V$,
$\lambda^{(0)}(u)_j = 1$ if $\varphi_j = s$ and $s \in \lambda(u)$, and
$\lambda^{(0)}(u)_j = 0$ otherwise.

The basic $\GPS$-layers are constructed as follows. The activation function $\sigma$ is the Heaviside step function, that is,
\[
    \sigma(x) = \begin{cases} 1 & x > 0 \\ 0 & x \leq 0. \end{cases}
\]
In the self-attention head $H$, choose $\alpha = \mathrm{softmax}$, and set all
entries in $W_Q$ and $W_K$ to $0$. The latter ensures that in the computation of
$H(X)$, every vertex pays attention $\frac{1}{|V|}$ to every vertex.
Set the entries of $A$, $C$, $b$, and $W_V$, depending on the
subformulae of $\varphi$, as follows. In column $j$,
\begin{enumerate}
    \item if $\varphi_j \in \Pi$, then set $C_{jj} = 1$, 
    \item if $\varphi_j = \varphi_{k} \land \varphi_{k'}$, set $C_{k, j} = C_{k', j} =  1$ and $b_j = -1$,
    \item if $\varphi_j = \neg \varphi_k$, set $C_{k, j} = -1$ and $b_j = 1$,
    \item if $\varphi_j = \Diamond_{\geq c} \varphi_k$, set $A_{k, j} = 1$ and $b_j = -c + 1$,
    \item if $\varphi_j = \DiamondG \varphi_k$, set $(W_V)_{k, j} = 1$,
\end{enumerate}
and set all other values of column $j$ in $A$, $C$, $W_V$ and $b$ to $0$. Using
induction on $j$, one can now show that ($*$) is satisfied. The arguments for
Cases~1 to~4 can be found in \citep{Barcelo_GNNs}. Thus, we only explicitly treat
Case~5.

Thus, let $\varphi_j = \DiamondG \varphi_k$ and assume that $(*)$ holds for
$\varphi_k$, that is, $\lambda^{(i)}(u)_k = 1$ if $\cG, u \models \varphi_k$ and
$\lambda^{(i)}(u)_k = 0$ otherwise, for all $i$ with $k \leq i \leq d$ and all
$u \in V$.
Now consider any $u \in V$.
Recall that
\[
    \lambda^{(i + 1)}(u) \colonequals \sigma \Big( \lambda^{(i)}(u) C + \sum_{\mathclap{u' \in \neigh_{\cG}(v)}} \lambda^{(i)}(u') A + H(\lambda^{(i)})_{u, *} + b \Big).
\]
We are interested in $\lambda^{(i + 1)}(u)_j$. As $C_{k, j} = A_{k, j} =  b_j = 0$ for all $k$,
\[
    \lambda^{(i + 1)}(u)_j \colonequals \sigma \big( H(\lambda^{(i)})_{u, j} \big).
\]
By choice of $W_K$ and $W_Q$, $H(\lambda^{(i)})$ computes
\[
    \frac{\sum_{v \in V} \lambda^{(i)}(v)}{|V|} W_V.
\]
By $(*)$ all $\lambda^{(i)}(v)$ are either $0$ or $1$. Hence, by choice of $W_V$, $H(\lambda^{(i)})_{u,
j} > 0$ if and only if there is a vertex $u' \in V$ with $\lambda^{(i)}(u')_k = 1$.
By choice of $\sigma$ and $(*)$ it thus follows that
$\lambda^{(i + 1)}(u)_j = 1$ if $\cG, u \models \DiamondG \varphi_k$ and
$\lambda^{(i + 1)}(u)_j = 0$ otherwise, as required.

Observe that since all entries of $W_Q$ and $W_K$ are $0$, the same argument also applies if one chooses $\mathrm{AH}$ as the
attention function.
\end{proof}

\subsection{Proof of Lemma~\ref{lem:gps-invariant-under-grg-bisim}}

\lemgraphtransformerinvariance*

\begin{proof}
Let $N = (P, L^{(1)}, \ldots, L^{(k)}, C)$ be a GPS-network, and $(\cG_1, v_1), (\cG_2, v_2)$ $\Pi$-labeled
pointed graphs with $\cG_1 = (V_1, E_1, \lambda_1)$, $\cG_2 = (V_2, E_2,
\lambda_2)$, such that $(\cG_1, v_1) \sim_{G\%} (\cG_2, v_2)$.

For $j \in \{1, 2\}$, let $\lambda_j^{(0)},\dots,\lambda_j^{(k)}$  be the
feature maps of dimension $d$ computed by the layers of $N$ on $\cG_j$.
To prove the lemma, it
suffices to show  
 the following for  $0 \leq i \leq k$:
\begin{itemize}
    \item[($*$)]for all $u \in V_1$ and $v \in V_2$,  $\mn{tp}_{\cG_1}(u) = \mn{tp}_{\cG_2}(v)$ implies $\lambda^{(i)}_1(u) = \lambda^{(i)}_2(v)$.
\end{itemize}
We prove ($*$) by induction on $i$.

\smallskip
For the induction start, where $i = 0$, recall that $\mn{tp}_{\cG_1}(u) = \mn{tp}_{\cG_2}(v)$ implies that
$(\cG_1, u) \sim (\cG_2, v)$. Then ($*$) follows from Condition~\textbf{atom} of graded bisimulations and the definition of $\lambda^{(0)}_j$, $j \in \{1,2\}$.

\smallskip

For the induction step, assume that ($*$) holds for $i$ and let $u \in V_1$ and $v \in V_2$ such that
$\mn{tp}_{\cG_1}(u) = \mn{tp}_{\cG_2}(v)$. We have to show that 
$\lambda^{(i + 1)}_1(u) = \lambda^{(i + 1)}_2(v)$.
Recall that, for  $j \in \{1, 2\}$,
\[
\begin{aligned}
    \lambda^{(i+1)}_{j, B} &\colonequals \lambda^{(i)}_j + B^{(i+1)}\big( \lambda^{( i)}_j \big), \text{ where } B \in \{\SA, \MP\}, \\
    \lambda^{(i+1)}_{j, \SA+\MP} &\colonequals \lambda^{(i+1)}_{j, \SA} + \lambda^{(i+1)}_{j, \MP},\ \text{and}\\
    \lambda^{(i+1)}_{j} &\colonequals   \lambda^{(i+1)}_{j, \SA+\MP} + \FF^{(i+1)} \big( \lambda^{(i+1)}_{j, \SA+\MP} \big).
\end{aligned}
\]
As $+$ and $\mathrm{FF}^{(i + 1)}$ are applied row-wise and $\lambda^{(i)}_1(u) = \lambda^{(i)}_2(v)$ by the induction hypothesis, it suffices to show
that for all $u \in V_1$ and $v \in V_2$, $\mn{tp}_{\Gmc_1}(u) = \mn{tp}_{\Gmc_2}(v)$ implies
$\MP^{(i+1)} \big( \lambda^{(i)}_1 \big)(u) = \MP^{(i+1)} \big( \lambda^{(i)}_2 \big)(v)$
and
$\SA^{(i+1)} \big( \lambda^{(i)}_1 \big)(u) = \SA^{(i+1)} \big( \lambda^{(i)}_2 \big)(v)$. 

We begin with $\mathrm{MP}^{(i + 1)}$. Recall that, for $j \in \{1, 2\}$,
\begin{align*}
    &\MP^{(i+1)} \big( \lambda^{(i)}_j \big)(u) =  \COM^{(i+1)} \Big( \lambda^{( i)}_j(u), \gamma_j^{(i)}(u) \Big) \text{, where}\\
    &\gamma_j^{(i)}(u) \colonequals  \AGG^{(i+1)}  \big( \{\!\{ \lambda^{( i)}_j(u') \mid (u,u') \in E_j \}\!\} \big).
\end{align*}
From the induction hypothesis and the Conditions~\textbf{graded forth} and
\textbf{graded back} of graded bisimulations, it follows that 
\[
    \{\!\{ \lambda^{( i)}_1(u') \mid (u,u') \in E_1 \}\!\}  = \{\!\{ \lambda^{( i)}_2(v') \mid (v, v') \in E_2 \}\!\}.
\]
Thus, 
as 
$\lambda^{( i)}_1(u)  = \lambda^{(i)}_2(v)$ by induction hypothesis,
it further follows that
\[
\MP^{(i+1)} \big( \lambda^{(i)}_1 \big)(u) = \MP^{(i + 1)} \big( \lambda^{(i)}_2 \big)(v).
\]

\medskip

Now consider the self-attention module $\mathrm{SA}^{(i + 1)} = (H^{(1)},
\ldots, H^{(k')}, W_O)$ and recall that, for $j \in \{1, 2\}$, $\SA^{(i+1)}
\big( \lambda^{(i)}_j \big) = \cH\big( \lambda^{(i)}_j \big) W_O$ where  $\cH
\big( \lambda^{(i)}_j\big)$ is the concatenation of the matrices
$H^{(1)}\big(\lambda^{(i)}_j\big), \ldots, H^{(k')}\big(\lambda^{(i)}_j\big)$.
It thus suffices to show that for all $\ell \in [k']$, all  $u \in V_1$ and
$v \in V_2$, 
\[\mn{tp}_{\Gmc_1}(u) = \mn{tp}_{\Gmc_2}(v) \text{ implies }
H^{(\ell)}(\lambda^{(i)}_1)_{u, *} = H^{(\ell)}(\lambda^{(i)}_2)_{v, *}.\]
Thus, let $\ell \in [k']$ and
\[
H^{(\ell)}(X) \colonequals \alpha\left( \frac{(X W_Q) (X W_K)^{\T}}{\sqrt{d_h}} \right) (X W_V).
\] 
Consider the matrices
\[
        A = \frac{(\lambda^{(i)}_1 W_Q) (\lambda^{(i)}_1 W_K)^\T}{\sqrt{d_h}}
\]
and
\[
        B = \frac{(\lambda^{(i)}_2 W_Q) (\lambda^{(i)}_2 W_K)^\T}{\sqrt{d_h}}.
\]
To show that $H^{(\ell)}(\lambda^{(i)}_1)_{u, *} = H^{(\ell)}(\lambda^{(i)}_2)_{v, *}$, we must show  
that $\big(\alpha(A) \lambda^{(i)}_1 \big)_{u, *} = \big(\alpha( B) \lambda^{(i)}_2\big)_{v, *}$, for $\alpha \in \{\softmax, \AH\}$.

For $q > 0$ the rational number that shows global-ratio graded bisimilarity of
$(\cG_1, v_1)$ and $(\cG_2, v_2)$, the induction hypothesis implies that for all $\mathbf{x}
\in \mathbb{R}^{d}$:
\begin{equation*}
    |\{ u' \in V_1 \mid \lambda^{(i)}_1(u') = \mathbf{x} \} | = q \cdot | \{ v' \in V_2 \mid \lambda^{(i)}_2(v') = \mathbf{x} \} |. \tag{$\dag_1$}
\end{equation*}
Furthermore, observe that computing an entry $A_{u, u'}$ depends only on $\lambda^{(i)}_1(u)$
and $\lambda^{(i)}_1(u')$ (and the same holds for $B$ and~$\lambda^{(i)}_2$). It then follows from the
induction hypothesis that, for all $u' \in V_1$ and $v' \in V_2$, $\mn{tp}_{\Gmc_1}(u') = \mn{tp}_{\Gmc_2}(v')$ implies $A_{u, u'} = B_{v, v'}$.
Hence, for all $a \in \mathbb{R}$,
\begin{equation*}
    |\{ u' \in V_1 \mid A_{u, u'} = a \} | = q \cdot | \{ v' \in V_2 \mid B_{v, v'} = a \} |.
\end{equation*}
Therefore, if $\alpha = \softmax$,
\begin{equation*}
    \sum_{u' \in V_1} e^{A_{u, u'} - b} = q \cdot \sum_{v' \in V_2} e^{B_{v, v'} - b},
\end{equation*}
and, for all vertices $u' \in V_1$ and $v' \in V_2$ with $\mn{tp}_{\cG_1}(u') = \mn{tp}_{\cG_2}(v')$,
\begin{align*}
    \softmax(A)_{u, u'} & = \frac{e^{ A_{u, u'} - b}}{\sum_{u'' \in V_1} e^{A_{u, u''} - b} } \\
    & = \frac{e^{ B_{v, v'} - b}}{ q \cdot \sum_{v'' \in V_2} e^{B_{v, v''} - b} }  \\
    & = \frac{1}{q} \cdot \softmax(B)_{v, v'}.\tag{$\dagger_2$} 
\end{align*}
Let $\cT$ be the collection of all graded bisimulation types over $\Pi$. We obtain that
\begin{align*}
    (\softmax(&A) \lambda^{(i)}_1)_{u, *} \\
    & = \sum_{u' \in V_1} \softmax(A)_{u, u'} \* \lambda^{(i)}_1(u') \\
                           & = \sum_{t \in \cT} \sum_{\substack{u' \in V_1\\\mn{tp}_{\Gmc_1}(u') = t}} \softmax(A)_{u, u'} \lambda^{(i)}_1(u') \\
                            & = \sum_{t \in \cT} \Big (q \cdot \sum_{\mathclap{\substack{v' \in V_2\\\mn{tp}_{\Gmc_2}(v') = t}}} \frac{1}{q} \cdot \softmax(B)_{v, v'} \lambda^{(i)}_2(v') \Big )\\
                           & = q \cdot \sum_{v' \in V_2} \frac{1}{q} \cdot \softmax(B)_{v, v'} \lambda^{(i)}_2(v') & & \\
                           & = \sum_{v' \in V_2} \softmax(B)_{v, v'} \lambda^{(i)}_2(v') \\
                           & = (\softmax(B) \lambda^{(i)}_2)_{v, *}.
\end{align*}
Note that the third equation holds by ($\dagger_1$) and ($\dagger_2$).

\medskip

For $\alpha = \AH$, consider the sets of indices with maximal value $\cI_{A_{u,
*}}$ and $\cI_{B_{v, *}}$.
By $(\dag_1)$, $|\cI_{A_{u, *}}| = q \cdot |\cI_{B_{v, *}}|$, and thus, for $u'
\in \cI_{A_{u, *}}$ and $v' \in \cI_{B_{v, *}}$,
\[
    \AH(A)_{u, u'} = \frac{1}{ |\cI_{A_{u, *}}|} = \frac{1}{q \cdot |\cI_{B_{v, *}|}} = \frac{1}{q} \cdot \AH(B)_{v, v'}.
\]
Hence, using the same argument as for $\softmax$, we obtain $(\AH(A) \lambda^{(i)}_1)_{u, *} = (\AH(B) \lambda^{(i)}_2)_{v, *}$.
\end{proof}

\subsection{Proof of Theorem~\ref{lem:ratio-bisim-invariance}}

We begin by introducing another central notion of bisimulation, which is associated with $\GMLG$.

\begin{definition}[Global Graded Bisimulation]
A \textbf{global graded bisimulation} $Z$ between $\Pi$-labeled graphs $\cG_1
= (V_1, E_1, \lambda_1)$ and $\cG_2 = (V_2, E_2, \lambda_2)$ is a graded bisimulation
such that the following additional conditions are satisfied:

\begin{description}
    \item[global forth] for all $v_1 \in V_1$ there is a $v_2 \in V_2$ with $(v_1, v_2) \in Z$.
    
    \item[global back] for all $v_2 \in V_2$ there is a $v_1 \in V_1$ with $(v_1, v_2) \in Z$.
\end{description}

If there is a global graded bisimulation $Z$ between $\Pi$-labeled graphs
$\cG_1$ and $\cG_2$ with $(v_1, v_2) \in Z$, we write $(\cG_1, v_1) \sim_{G}
(\cG_2, v_2)$.
\end{definition}

In this section, we aim to show that every $\FO$-formula that is invariant under
$\sim_{G\%}$ is equivalent to a $\GMLG$-formula. This implies that these FO-formulae are also invariant under $\sim_{G}$, hence, relative to FO, $\sim_{G\%}$
and $\sim_{G}$ are the same.
To show this, we apply techniques by Otto~\citep{otto2004,otto2019}.

The \textbf{graded bisimulation game} that is associated with $\GML$ is played between two players,
\mn{spoiler} and \mn{duplicator} over two $\Pi$-labeled graphs $\cG_1 = (V_1,
E_1, \lambda_1)$ and $\cG_2 = (V_2, E_2, \lambda_2)$. Positions are pairs $(v_1, v_2) \in V_1
\times V_2$ and a single round played from this position allows \mn{spoiler}
to challenge \mn{duplicator}.
The \textbf{$c$-graded $\ell$-round bisimulation game}, consists of $\ell$
rounds. Each round consists of the following moves:
\begin{description}
    \item[graded down] \mn{spoiler} chooses a non-empty subset of $\neigh_{\cG_1}(v_1)$
    or $\neigh_{\cG_2}(v_2)$ of size at most $c$. \mn{duplicator} must respond with a
    matching subset of $\neigh_{\cG_2}(v_2)$ or $\neigh_{\cG_1}(v_1)$ on the opposite side of
    the same size.

    \mn{spoiler} then picks a vertex in the set proposed by \mn{duplicator} and
    \mn{duplicator} must respond by picking a vertex in the set proposed by
    \mn{spoiler}.
\end{description}
Either player loses in this round if stuck, and \mn{duplicator} loses as soon as
the current position $(v_1, v_2)$ violates $\lambda_1(v_1) = \lambda_2(v_2)$.

If \mn{duplicator} has a winning strategy for the  $c$-graded $\ell$-round
bisimulation game on $\Pi$-labeled graphs $\cG_1$, $\cG_2$ on starting position
$(v_1, v_2)$, we say that $(\cG_1, v_1), (\cG_2, v_2)$ are \textbf{$c$-graded
$\ell$-bisimilar}, written $(\cG_1, v_1) \sim^{c, \ell} (\cG_2, v_2)$. We say that $(\cG_1, v_1), (\cG_2, v_2)$ are \textbf{globally $c$-graded
$\ell$-bisimilar}, written $(\cG_1, v_1) \sim_G^{c, \ell} (\cG_2, v_2)$, if (i)~$(\cG_1, v_1) \sim^{c, \ell} (\cG_2, v_2)$ and (ii)~for every $u_1 \in V(\Gmc_1)$ there is a $u_2 \in V(\Gmc_2)$ such 
that $(\cG_1, u_1) \sim^{c, \ell} (\cG_2, u_2)$
and vice versa.

\medskip

Global $c$-graded $\ell$-bisimilarity is closely connected to $\GMLG$ as the following observations show.

\begin{lemma}\label{lem:bisim-implies-winning-strategy}
    Let $(\Gmc_1, v_1)$, $(\Gmc_2, v_2)$ be $\Pi$-labeled pointed graphs. If $(\Gmc_1, v_1)
    \sim (\Gmc_2, v_2)$, then, for all $c, \ell \geq 0$, $(\Gmc_1, v_1) \sim^{c,
    \ell} (\Gmc_2, v_2)$.
\end{lemma}

\begin{lemma}\label{lem:gmlg-invariance}
    Let $\varphi$ be a $\GMLG$-formula. Then, there are $c, \ell \geq 0$ such that
    $\varphi$ is invariant under $\sim_{G}^{c, \ell}$.
    If $\varphi$ is a $\PLG$-formula, then $\varphi$ is invariant under $\sim_{G}^{0, 0}$.
\end{lemma}

\begin{lemma}\label{lem:gmlg-bisim-equivalence}
    For every $c, \ell \geq 0$, the equivalence relation $\sim_{G}^{c, \ell}$ has finite
    index and for each equivalence class of $\sim_{G}^{c, \ell}$, there is a
    $\GMLG$-formula $\varphi$ that defines this equivalence class.
    If $\ell = 0$, then $\varphi$ is a $\PLG$-formula.
\end{lemma}

For the first step of our proof, we now introduce the up-and-down variant $\approx$ of
graded bisimulation, and its global variant $\approx_{G}$.  In a $\Pi$-labeled graph $\cG = (V, E, \lambda)$,
we use $\pred_\cG(v)$ to denote the set of \textbf{predecessors} of $v \in V$, that is, $\{v' \mid (v', v) \in E\}$.

\begin{definition}[Up-Down Graded Bisimulation]
    An \textbf{up-down graded bisimulation} $Z$ between $\Pi$-labeled graphs   $\cG_1
= (V_1, E_1, \lambda_1)$ and $\cG_2 = (V_2, E_2, \lambda_2)$ is a graded bisimulation
that satisfies the following additional conditions:

    \begin{description}
        \item[graded up forth] for all $(u_1, u_2) \in Z$, for all $k \geq 1$: for
        pairwise distinct $v_1, \ldots v_k \in \pred_{\cG_1}(u_1)$ there are
        pairwise distinct $v'_1, \ldots, v'_k \in \pred_{\cG_2}(u_2)$ such that
        $(v_1, v'_1), \ldots$, $(v_k, v'_k) \in Z$.
        \item[graded up back] for all $(u_1, u_2) \in Z$, for all $k \geq 1$: for
        pairwise distinct $v'_1, \ldots, v'_k \in \pred_{\cG_2}(u_2)$ there are
        pairwise distinct $v_1, \ldots, v_k \in \pred_{\cG_1}(u_1)$ such that $(v_1, v'_1), \ldots$, $(v_k, v'_k) \in Z$.
    \end{description}

If there is an up-down graded bisimulation $Z$ between $\Pi$-labeled
graphs $\cG_1$ and $\cG_2$ with $(v_1, v_2) \in Z$, we write $(\cG_1, v_1)
\approx (\cG_2, v_2)$.
\end{definition}

In analogy to global graded bisimulations, we also define \textbf{global up-down
graded bisimulations} $\approx_G$ in the obvious way using the
\textbf{global forth} and \textbf{global back} conditions.

Furthermore, we also define an up-down variant of graded bisimulation games. The
\textbf{up-down $c$-graded $\ell$-round bisimulation game} extends the $c$-graded $\ell$-round bisimulation game by equipping \mn{spoiler} with
the following additional move:

\begin{description}
    \item[graded up] \mn{spoiler} chooses a non-empty subset of $\pred_{\cG_1}(v_1)$
    or $\pred_{\cG_2}(v_2)$ of size at most $c$. \mn{duplicator} must respond with a
    matching subset of $\pred_{\cG_2}(v_2)$ or $\pred_{\cG_1}(v_1)$ on the opposite side of
    the same size.

    \mn{spoiler} then picks a vertex in the set proposed by \mn{duplicator} and
    \mn{duplicator} must respond by picking a vertex in the set proposed by
    \mn{spoiler}.
\end{description}

If \mn{duplicator} has a winning strategy for the up-down $c$-graded $\ell$-round
bisimulation game on $\Pi$-labeled graphs $\cG_1$, $\cG_2$ on starting position
$(v_1, v_2)$, we say that $(\cG_1, v_1), (\cG_2, v_2)$ are \textbf{up-down $c$-graded
$\ell$-bisimilar}, written $(\cG_1, v_1) \approx^{c, \ell} (\cG_2, v_2)$.
We say that $(\cG_1, v_1), (\cG_2, v_2)$ are \textbf{globally up-down $c$-graded
$\ell$-bisimilar}, written $(\cG_1, v_1) \approx_G^{c, \ell} (\cG_2, v_2)$, if (i)~$(\cG_1, v_1) \approx^{c, \ell} (\cG_2, v_2)$ and (ii)~for every $u_1 \in V(\Gmc_1)$ there is a $u_2 \in V(\Gmc_2)$ such 
that $(\cG_1, u_1) \approx^{c, \ell} (\cG_2, u_2)$
and vice versa.

\medskip

We now start with a small but useful observation about $\sim_{G\%}$.
For a graph $\Gmc = (V, E, \lambda)$, let $q \cdot \Gmc$, for $q \geq 1$ denote the $q$-fold disjoint union of $\Gmc$ with itself, that is,
\begin{align*}
    V(q \cdot \Gmc) &= V \times [q], \\
    E(q \cdot \Gmc) &= \{ ((v, i), (v', i)) \mid (v, v') \in E, 1 \leq i \leq q\}, \\
    \lambda(q \cdot \Gmc) &= \{ (v, i) \mapsto \lambda(v) \mid  (v, i) \in V(q \cdot \Gmc) \}.
\end{align*}
\begin{proposition}
\label{prop:disjointcopy}
    $\sim_{G\%}$ is preserved under disjoint copies, that is, for all $q \geq 1$ and $i$ with $1 \leq i \leq q$, 
    \[(\Gmc, v)
    \sim_{G\%} (q \cdot \Gmc, (v, i)).\]
\end{proposition}

This observation allows us to show the following first lemma.

\begin{lemma}\label{lem:ratio-upgrade-inverse-bisim-to-fo}
For every $\FO$-formula  $\varphi(x)$ that is invariant under $\sim_{G\%}$,
there are $c, \ell \geq 0$ such that $\varphi$ is invariant under
$\approx_{G}^{c, \ell}$.
\end{lemma}
\begin{proof}
  This is essentially already proved by Otto in \citep{otto2004}. What is shown there is
  that if an $\FO$-formula $\varphi(x)$  is invariant under $\hat\sim_{G}$,
  then there are $c, \ell \geq 0$ such that $\varphi$ is invariant under
  $\approx_{G}^{c, \ell}$, where
  $\hat \sim_G$ is defined like $\sim_G$, but with \textbf{graded back} and \textbf{graded forth}
  replaced by non-graded versions. Note that this is a stronger precondition
  since $(\Gmc,v) \sim_{G\%} (\Gmc',v')$ implies   $(\Gmc,v)
  \mathrel{\hat\sim_G} (\Gmc',v')$.
Nevertheless, exactly the same constructions also establish Lemma~\ref{lem:ratio-upgrade-inverse-bisim-to-fo}. Below, we describe Otto's proof in some more detail, and discuss why it yields Lemma~\ref{lem:ratio-upgrade-inverse-bisim-to-fo}. For full details, 
we refer the reader to~\citep{otto2004}. 

We start with introducing some relevant notions. Let $\Gmc=(V,E,\lambda)$ be a $\Pi$-labeled graph. Then
\begin{itemize}

    \item 
    the \textbf{distance} between vertices $v,v' \in V$, denoted $d(v,v')$ is the length of the shortest path between $v$ and $v'$ in \Gmc viewed as an undirected graph, and $\infty$ if no such path exists;
    
    \item the \textbf{neighborhood of radius $\ell \geq 0$} around a vertex $v$ in a graph \Gmc is
$$
   N^\ell(v)=\{ v' \in V \mid
   d(v,v') \leq \ell \}.
$$
\end{itemize}
We next define $\FO$-formulae in Gaifman form and their dimensions:
\begin{itemize}

    \item  an $\FO$-formula $\varphi(x)$ is 
\textbf{$\ell$-local} if it is equivalent to its relativization to $N^\ell(v)$;

    \item an $\FO$-formula
$\varphi(x)$ is  a
\textbf{simple $\ell$-local Gaifman formula} if it is
  a Boolean combination of (i)~$\ell$-local formulae $\psi(x)$ and
(ii)~formulae $\exists y \, \vartheta(y)$ with $\vartheta$ an $\ell$-local formula; the \textbf{local quantifier rank} of $\varphi$ is the maximum
quantifier rank of all constituting formulae $\psi,\vartheta$.
\end{itemize}
Let $\varphi(x)$ be an $\FO$-formula that is 
invariant under $\sim_{G\%}$. Then by Proposition~\ref{prop:disjointcopy} $\varphi(x)$ is invariant under disjoint copies. By Proposition~19 of \citep{otto2004}, this implies
that $\varphi(x)$ is equivalent to a simple $\ell$-local Gaifman formula
$\psi(x)$, for some $\ell \geq 0$. 

We can then use the proof of Lemma~35 of \citep{otto2004}, without modifying any constructions, to show that together with invariance under $\sim_{G\%}$,  this implies that $\varphi(x)$ is invariant under $\approx^{c,\ell}_G$ where $c$ is the local quantifier rank of $\psi$. This implies Lemma~\ref{lem:ratio-upgrade-inverse-bisim-to-fo}. We next give some details.
Assume to the contrary of what we have to show that there
are $\Pi$-labeled pointed graphs $(\Gmc_1,v_1)$ and $(\Gmc_2,v_2)$ such
that $(\Gmc_1,v_1) \approx^{c,\ell}_G (\Gmc_2,v_2)$,
$\Gmc_1 \models \varphi(v_1)$, and $\Gmc_2 \not\models \varphi(v_2)$. We show
how to transform $(\Gmc_1,v_1)$ and $(\Gmc_2,v_2)$ into $\Pi$-labeled pointed
graphs $(\Hmc_1,u_1)$ and $(\Hmc_2,u_2)$ such that the following conditions are
satisfied:
\begin{enumerate}

    \item $(\Gmc_i,v_i) \sim_{G\%} (\Hmc_i,u_i)$ for $i \in \{1,2\}$;

    \item there is no simple $\ell$-local Gaifman formula $\chi(x)$ such that $\Hmc_1 \models \chi(u_1)$ and $\Hmc_2 \not\models \chi(u_2)$.
    
\end{enumerate}
This clearly yields the desired contradiction: from
$\Gmc_1 \models \varphi(v_1)$ and $\Gmc_2 \not\models \varphi(v_2)$ and Point~1, we obtain 
$\Hmc_1 \models \varphi(u_1)$, and $\Hmc_2 \not\models \varphi(u_2)$, in contradiction to Point~2 and $\varphi$ being equivalent to $\psi$.

In \citep{otto2004}, $(\Hmc_1,u_1)$ and
$(\Hmc_2,u_2)$ are obtained from $(\Gmc_1,v_1)$ and $(\Gmc_2,v_2)$ in three steps, which sufficiently increase the length of all cycles in $\Gmc_1$ and $\Gmc_2$. Recall that the \textbf{girth} of a graph is the minimal length of cycles in that graph. The first step introduces extra vertices to ensure that all cycles have at least length 3, the second step takes the product with the Caley graph of a group that has high girth which increases the length of all cycles of length at least 3, and the last step removes the extra vertices introduced by the first step.

Formally, the first step is to transform $(\Gmc_1,v_1)$ and $(\Gmc_2,v_2)$ into $\Pi \cup \{ X,Y\}$-labeled pointed graphs $(\Imc_1,v_1)$ and $(\Imc_2,v_2)$ where $X,Y$ are fresh vertex labels. For $i \in \{1,2\}$, define\footnote{We replace every edge with a path of length~3. In \citep{otto2004}, a path of length~2 is used. We believe that this is a mistake as  a reflexive edge in
$\Gmc_i$ then results in two symmetric edges in $\Imc_i$, but this is not allowed in a graph that is simple according to Definition~27 in \citep{otto2004}.}
$$
\begin{array}{rcl}
    V(\Imc_i) &=& V(\Gmc_i) \cup \{ v_e,u_e \mid e \in E(\Gmc_i) \} \\[1mm]
    E(\Imc_i) &=& \{ (v,v_e),(v_e,u_e),(u_e,u) \mid \\[1mm]
    && \qquad \qquad e = (v,u) \in E(\Gmc_i) \} \\[1mm]
    \lambda(\Imc_i) &=& \{v \mapsto L \mid v \mapsto L \text{ in } \lambda(\Gmc_i) \} \; \cup \\[1mm]
    && \{v_e \mapsto \{X\},
    u_e \mapsto \{Y\} \mid e \in E(\Gmc_i) \}.
\end{array}
$$
We next take the product with (the Cayley graph of) a finite group of high girth.  Let $(G, \circ)$ be a finite group of sufficiently high girth and $g \colon E(\Imc_i) \rightarrow G$
an embedding such that $\{g(e) \mid e \in E(\Imc_i)\} \cap \{g(e)^{-1} \mid e \in E(\Imc_i) \} = \emptyset$. 
We set
$$
\begin{array}{rcl}
    V(\Imc'_i) &=& V(\Imc_i) \times G \\[1mm]
    E(\Imc'_i) &=& \{ ((v,a),(u,a \circ g(e)))  \mid \\[1mm]
    && \qquad \qquad e=(v,u) \in E(\Imc_i), \ a \in G \} \\[1mm]
    \lambda(\Imc'_i) &=&  \{(v,a) \mapsto L \mid v \mapsto L \text{ in } \lambda(\Imc_i),\ a \in G \}.
\end{array}
$$
In the third step, we `reverse' the effect of the first step and move back to $\Pi$-labeled pointed graphs by setting
$$
\begin{array}{rcl}
    V(\Hmc_i) &=& \{ (v,a) \mid 
    (v,a) \in V(\Imc'_i)\text{ with } v \in V(\Gmc_i) \}  \\[1mm]
    E(\Hmc_i) &=& \{ ((v,a_1),(u,a_4)) \mid \\[1mm] 
    && \quad ((v,a_1),(v_e,a_2)), ((v_e,a_2),(u_e,a_3)), \\[1mm]
    && \quad \quad ((u_e,a_3),(u,a_4)) \in E(\Imc'_i), \\[1mm]
    && \hspace{-3em} X \in \lambda(\Imc'_i)((v_e, a_2)), \text{ and } Y \in \lambda(\Imc'_i)((u_e, a_3)) 
    \} 
\end{array}
$$
and taking $\lambda(\Hmc_i)$ to be the restriction of $\lambda(\Imc'_i)$ to the vertices in $\Hmc_i$. We choose an $a \in G$ and set $u_1=(v_1,a)$
and $u_2 = (v_2,a)$.

It is proved in \citep{otto2004} (Lemma~35) that $(\Gmc_1,v_1) \approx^{c,\ell}_G (\Gmc_2,v_2)$
implies Point~2 above. It thus remains to prove Point~1, that is, for $i \in \{1,2\}$
(i)~$(\cG_i, v_i)
    \sim
    (\cH_i, u_i)$ and (ii)~there exists
    a rational $q > 0$ such that for each graded bisimulation type $t$ over $\Pi$,
    \[
        \begin{aligned}
            &|\{ v \in V(\Gmc_i) \mid \mn{tp}_{\Gmc_i}(v)=t\}| \\
       & \qquad \qquad = q \cdot  |\{ v \in V(\Hmc_i) \mid \mn{tp}_{\Hmc_i}(v)=t\}|.
        \end{aligned}
    \]
    It is easy to verify that
$$
 Z = \{ (v,(v,a)) \mid (v,a) \in V(\Hmc_i) \}
$$
is a graded bisimulation. Together with the choice of $u_1$ and $u_2$, this yields Point~(i). It also implies
$\mn{tp}_{\Gmc_i}(v)=\mn{tp}_{\Hmc_i}(v,a)$ for all $(v,a) \in V(\Hmc_i)$
and thus yields Point~(ii) for $q=\frac{1}{|G|}$.
\end{proof}

In the next step of the proof, we show invariance under an up-down bisimulation relation that can count
out-neighbors, but not predecessors.

\begin{definition}[Up-Ungraded Down-Graded Bisimulation]
    An \textbf{up-ungraded down-graded bisimulation} $Z$ between $\Pi$-labeled graphs   $\cG_1
= (V_1, E_1, \lambda_1)$ and $\cG_2 = (V_2, E_2, \lambda_2)$ is a graded bisimulation
that satisfies the following additional conditions:

    \begin{description}
        \item[ungraded up forth] for all $(u_1, u_2) \in Z$ and
        $v_1 \in \pred_{\cG_1}(u_1)$, there is a $v_2 \in \pred_{\cG_2}(u_2)$
        such that $(v_1,v_2) \in Z$.
        \item[ungraded up back] for all $(u_1, u_2) \in Z$ and
        $v_2 \in \pred_{\cG_2}(u_2)$, there is a $v_1 \in \pred_{\cG_1}(u_1)$
        such that $(v_1,v_2) \in Z$.
    \end{description}

If there is an up-ungraded down-graded bisimulation $Z$ between $\Pi$-labeled
graphs $\cG_1$ and $\cG_2$ with $(v_1, v_2) \in Z$, we write $(\cG_1, v_1)
\approx^\downarrow (\cG_2, v_2)$.
\end{definition}
In analogy to global graded bisimulations, we also define \textbf{global up-ungraded down-graded bisimulations} $\approx_{G}^\downarrow$ in the obvious way using the conditions
\textbf{global forth} and \textbf{global back}.

Furthermore, we define a variant of graded bisimulation games that corresponds to $\approx^\downarrow$. The
\textbf{up-ungraded down-$c$-graded $\ell$-round bisimulation game} extends the 
$c$-graded $\ell$-round bisimulation game by equipping \mn{spoiler} with
the following additional move:

\begin{description}
    \item[ungraded up] \mn{spoiler} chooses an element of $\pred_{\cG_1}(v_1)$
    or $\pred_{\cG_2}(v_2)$. \mn{duplicator} must respond with an element of
    $\pred_{\cG_2}(v_2)$ or $\pred_{\cG_1}(v_1)$ on the opposite side.
\end{description}

If \mn{duplicator} has a winning strategy for the up-ungraded down-$c$-graded $\ell$-round
bisimulation game on $\Pi$-labeled graphs $\cG_1$, $\cG_2$ on starting position
$(v_1, v_2)$, we say that $(\cG_1, v_1), (\cG_2, v_2)$ are \textbf{up-ungraded down-$c$-graded
$\ell$-bisimilar}, written $(\cG_1, v_1) \approx^{c\downarrow, \ell} (\cG_2, v_2)$.
We say that $(\cG_1, v_1), (\cG_2, v_2)$ are \textbf{globally up-ungraded down-$c$-graded
$\ell$-bisimilar}, written $(\cG_1, v_1) \approx_G^{c\downarrow, \ell} (\cG_2, v_2)$, if (i)~$(\cG_1, v_1) \approx^{c\downarrow, \ell} (\cG_2, v_2)$ and (ii)~for every $u_1 \in V(\Gmc_1)$ there is a $u_2 \in V(\Gmc_2)$ such 
that $(\cG_1, u_1) \approx^{c\downarrow, \ell} (\cG_2, u_2)$
and vice versa.
\begin{lemma}\label{lem:bisim-upgrade-up-ungraded-to-up-graded}
Every $\FO$-formula $\varphi(x)$ that is invariant under $\sim_{G\%}$ and
under $\approx_{G}^{c, \ell}$, with $c,\ell \geq 0$, is also invariant under
$\approx_{G}^{c\downarrow, \ell}$.
\end{lemma}
\begin{proof}
Let $\varphi(x)$ be an $\FO$-formula that is invariant under $\sim_{G\%}$ and
$\approx_{G}^{c, \ell}$ with $c, \ell \geq 0$.
Assume to the contrary of what we have to show that there exist $\Pi$-labeled pointed graphs
$(\Gmc_1, v_1)$ and $(\Gmc_2, v_2)$ such that $(\Gmc_1, v_1)
\approx_{G}^{c\downarrow, \ell} (\Gmc_2, v_2)$, $\Gmc_1 \models \varphi(v_1)$
and $\Gmc_2 \not\models \varphi(v_2)$.
We construct new graphs $\Hmc_1, \Hmc_2$ as follows. 

Let $W_{c, \ell}$ be the finite  set of all words over the alphabet $\{1, \ldots, c\}$ that are of
length at most $\ell$, and let $\cdot$ denote concatenation.
For $i \in \{1, 2\}$ and $\Gmc_i = (V_i, E_i, \lambda_i)$, set $\Hmc_i = (V'_i, E'_i, \lambda'_i)$, where
\begin{align*}
    V'_i = {} & \{ (v, w) \mid v \in V_i, w \in W_{c, \ell}\}, \\
    E'_i = {} & \{ ((v, w \cdot b), (u, w)) \mid (v, u) \in E_i, \\
    & \qquad \qquad \qquad b \in [c], w \cdot b \in W_{c, \ell}\} \cup {} \\
         &\{ ((v, \varepsilon), (u, \varepsilon)) \mid (v, u) \in E_i\}, \\
    \lambda'_i = {} & \{ (v, w) \mapsto \lambda_i(v) \mid (v, w) \in V'_i\}.
\end{align*}

We show that $\Hmc_1, \Hmc_2$ satisfy the following conditions:
\begin{enumerate}
    \item $(\Gmc_i, v_i) \sim_{G\%} (\Hmc_i, (v_i, w))$ for all $w \in W_{c,
    \ell}$ and $i \in \{1, 2\}$;
    
    \item $(\Hmc_1, (v, w_1)) \approx^{c, k} (\Hmc_2, (u, w_2))$ for
    all $k \leq \ell$ and $w_1, w_2 \in W_{c, \ell}$ such that
    \begin{itemize}
        \item[$(*)$] if $|w_1| \neq |w_2|$, then $\max( |w_1|, |w_2|) \leq \ell - k$,
    \end{itemize} 
    and all $v, u$ such that $(\Gmc_1, v) \approx^{c\downarrow, k} (\Gmc_2, u)$;

    \item $(\Hmc_1, (v_1, \varepsilon)) \approx_{G}^{c, \ell} (\Hmc_2, (v_2, \varepsilon))$.
\end{enumerate}

This yields the desired contradiction: from Point~1 we can conclude that
$\Hmc_1 \models \varphi(v_1, \varepsilon)$ and $\Hmc_2 \not\models
\varphi(v_2, \varepsilon)$, but this together with Point~3 contradicts that $\varphi$ is invariant
under $\approx_{G}^{c, \ell}$. We use Point~2 to show Point~3.

\medskip
\noindent
\textit{Proof of Point~1}. Observe that for all $i \in \{1, 2\}$,
\[
    \{(v, (v, w)) \mid (v, w) \in V'_i\}
\]
is a graded (downwards-only) bisimulation between $\Gmc_i$ and $\Hmc_i$. Hence, for
all $(v, w) \in V'_i$, $\mn{tp}_{\Gmc_i}(v) = \mn{tp}_{\Hmc_i}(v, w)$. Note
that, as before, $\mn{tp}$ refers to the \emph{downwards-only} graded bisimulation
type.
Point~1 then follows from the fact that for every vertex $v \in V_i$, there are
exactly $|W_{c, \ell}|$ vertices $(v, w) \in V'_i$.

\medskip
\noindent
\textit{Proof of Point~2}. We show Point~2 by induction on $k$. In the case $k = 0$, $(\Hmc_1, (v, w_1))
\approx^{c, k} (\Hmc_2, (u, w_2))$ is immediate by definition of $\lambda'_i$.
Now let $k > 0$, and assume that the up-down $c$-graded $k$-round bisimulation
game starts in position $((v, w_1), (u, w_2))$ such that $(*)$ is satisfied.

If \mn{spoiler} makes a \textbf{graded down} move in $\Hmc_1$, all vertices
they select must be of shape $(v', w_1')$ with either $w_1 = w_1' =
\varepsilon$, or $w_1 = w_1' \cdot b$ for some $b \in [c]$.
As $(\Gmc_1, v) \approx^{c\downarrow, k}(\Gmc_2, u)$ and by construction of $\Hmc_1$ and $\Hmc_2$,   for each
$(v', w_1')$ selected by \mn{spoiler}, there must be a distinct $(u', w_2') \in V'_2$ with
$w_2 = w_2' = \varepsilon$ or $w_2 = w_2' \cdot b$ for some $b \in [c]$, such that
$(\Gmc_1, v') \approx^{c\downarrow, k - 1}(\Gmc_2, u')$.
Note that $\max(|w_1'|, |w_2'|) \leq \max(|w_1|, |w_2|) \leq \ell - k$ and
therefore $(*)$ holds for every possible new position. The induction hypothesis
then yields that \mn{duplicator} has a winning strategy for the $k - 1$-round
game from any possible new position, as required. The symmetric argument applies if
\mn{spoiler} makes a \textbf{graded down} move in $\Hmc_2$.

If \mn{spoiler} makes a \textbf{graded up} move in $\Hmc_1$, they select at most
$c$ distinct predecessors of $(v, w_1)$, which are of the form $(v', w_1 \cdot b)$ with $b \in [c]$, or possibly $(v', \varepsilon)$ if $w_1 = \varepsilon$. As
$(\Gmc_1, v) \approx^{c\downarrow, k}(\Gmc_2, u)$, there must be, for each
such $v'$ that occurs in the set selected by \mn{spoiler}, a corresponding predecessor
 $u'$ of $u$ such that $(\Gmc_1, v') \approx^{c\downarrow, k -
1}(\Gmc_2, u')$. Now, for each $(v', w_1 \cdot b)$ or $(v', \varepsilon)$ in the set selected by
\mn{spoiler}, \mn{duplicator} can select an answer $(u', w_2 \cdot a)$, starting with
$a = 1$ and increasing $a$ such that all choices are distinct. This is
possible because $w_1$ and $w_2$ satisfy $(*)$:
If $|w_1| = |w_2|$, then the existence of predecessors of $(v, w_1)$ implies that
$|w_1| = |w_2| < \ell$ by construction of $\Hmc_1$, and thus
$(u, w_2)$ has predecessors of the shape $(u', w_2 \cdot a)$ by
construction of $\Hmc_2$.
If $|w_1| \neq |w_2|$, then $(*)$ implies that $\max(|w_1|, |w_2|)
\leq \ell - k \leq \ell - 1$ and therefore $(u, w_2)$ also has vertices of the shape
$(u', w_2 \cdot a)$ as predecessors by construction of $\Hmc_2$.  
Furthermore, \mn{spoiler} selects at most $c$ vertices and $(u', w_2 \cdot a) \in V'_2$
for $a \in [c]$.

Now \mn{spoiler} selects one vertex from \mn{duplicator}'s set, and \mn{duplicator}
can respond accordingly. Let $((v', w_1'), (u', w_2'))$ be the resulting
position. If $|w_1'| = |w_2'|$, then $(*)$ holds. If $|w_1'| \neq |w_2'|$, then
$\max(|w_1'|, |w_2'|) \leq \max(|w_1|, |w_2|) + 1 \leq \ell - (k - 1)$. Thus
$(*)$ also holds for $k - 1$. As furthermore $(\Gmc_1, v') \approx^{c\downarrow, k -
1}(\Gmc_2, u')$ the induction hypothesis implies that \mn{duplicator} has a
winning strategy for the $k - 1$ round game from this position. A symmetric argument applies if
\mn{spoiler} makes a \textbf{graded up} move in $\Hmc_2$.

\medskip
\noindent
\textit{Proof of Point~3}. 
$(\Hmc_1, (v_1, \varepsilon)) \approx^{c, \ell} (\Hmc_2, (v_2, \varepsilon))$ follows
from Point~2 and the fact that $(\Gmc_1, v_1) \approx^{c\downarrow, \ell} (\Gmc_2,
v_2)$.
Let $(v, w)$ be any vertex in $\Hmc_1$. As $(\Gmc_1, v_1) \approx_{G}^{c
\downarrow, \ell} (\Gmc_2, v_2)$, there must be a vertex $u$ in $\Gmc_2$ such that
$(\Gmc_1, v) \approx^{c\downarrow, \ell} (\Gmc_2, u)$. Consider the vertex $(u, w)$ in $\Hmc_2$.
As the second components of $(v, w)$ and $(u, w)$ are identical, $(*)$ of Point~2 is satisfied. Thus,
$(\Hmc_1, (v, w)) \approx^{c, \ell} (\Hmc_2, (u, w))$ as required.
The symmetric argument applies to any vertex $(u, w)$ in $\Hmc_2$.
\end{proof}

In the last step of our proof, we finally show invariance under $\sim_{G}^{c, \ell}$.

\begin{lemma}\label{lem:ratio-upgrade-bisim-to-inverse-bisim}
Every $\FO$-formula $\varphi(x)$ that is invariant under $\sim_{G\%}$ and
under $\approx_{G}^{c\downarrow, \ell}$, with $c,\ell \geq 0$, is also invariant under
$\sim_{G}^{c, 2 \ell}$.
\end{lemma}

\begin{proof}
    The proof is inspired by the unique history construction from Otto's proof of Lemma 40 in~\citep{otto2004}. We can, however, not use that
    construction as is because given a $\Pi$-labeled pointed graph $(\Gmc,v)$ we want to produce a
    $\Pi$-labeled pointed graph $(\Hmc,u)$ such that $(\Gmc,v) \sim_{G\%} (\Hmc,u)$, and Otto's
    construction does not achieve this. 
    We thus adapt it in a suitable way. 

    Let $\varphi(x)$ be an $\FO$-formula that is
    invariant under $\sim_{G\%}$ and  $\approx_{G}^{c\downarrow, \ell}$, with $c,\ell \geq 0$.
    Assume to the contrary of what we have to 
    show that there exist pointed graphs  $(\Gmc_1,v_1)$ and $(\Gmc_2,v_2)$ such
     that $(\Gmc_1,v_1) \sim_{G}^{c, 2 \ell} (\Gmc_2,v_2)$,
    $\Gmc_1 \models \varphi(v_1)$, and $\Gmc_2 \not\models \varphi(v_2)$. 
    We construct new graphs $\Hmc_1,\Hmc_2$ as follows. 

    Let $\Pi$ be a finite set of vertex label symbols.
    A \textbf{$c$-graded $\ell$-bisimulation type over $\Pi$  is a maximal set $t$ of $\Pi$-labeled pointed graphs such
that $(\Gmc,v) \sim^{c, \ell} (\Gmc',v')$ for all $(\Gmc,v),(\Gmc',v') \in t$.}

For a $\Pi$-labeled pointed
graph $(\Gmc,v)$, we use $\mn{tp}^{c, \ell}_\Gmc(v)$ to denote the unique $c$-graded
$\ell$-bisimulation type $t$ over $\Pi$ such that $(\Gmc,v) \in t$.
If $u_1,\dots,u_k$ is a path in $\Gmc$ with $k \leq \ell$, then we call the
sequence $h=\mn{tp}^{c, \ell}_{\Gmc}(u_1),\ldots,\mn{tp}^{c, \ell}_{\Gmc}(u_k)$ an
\textbf{$\ell$-history} of $u_k$ in~$\Gmc$. We say that $h$ is \textbf{maximal} if
$k=\ell$ or $u_1$ does not have any predecessors in $\Gmc$.\footnote{Our
histories are defined in a more liberal way than in \citep{otto2004}, where they
must be 
maximal.}
We use $\mn{tail}(h)$ to denote $\mn{tp}^{c, \ell}_{\Gmc}(u_k)$. Note that a 
vertex may have multiple $\ell$-histories, of varying lengths, and that every
vertex has at least one maximal $\ell$-history.
We say that an $\ell$-history $h_2$ is a \textbf{continuation} of an $\ell$-history $h_1=t_1,\dots,t_k$ if one of the following conditions holds:
 \begin{itemize}
   
    \item $k < \ell$ and  $h_2=t_1,\dots,t_k,t$ for some $t$;

    \item $k=\ell$  and $h_2=t_2,\dots,t_k,t$ for some $t$.

\end{itemize}
An \textbf{$\ell$-history choice} for \Gmc is a function $\chi$ that assigns to
each vertex $v$ in \Gmc a (not necessarily maximal) $\ell$-history of $v$ in~\Gmc.

For $i \in \{1,2\}$,
 the graph $\Hmc_i$ is defined as follows:
\begin{itemize}

    \item $V(\Hmc_i)$ consists of all pairs $(v,\chi)$ with $v \in V(\Gmc_i)$
    and $\chi$ a history choice for~$\Gmc_i$;

    \item for every edge $(v,u) \in E(\Gmc_i)$ and every vertex  $(v,\chi) \in V(\Hmc_i)$,
    we find an edge $((v,\chi),(u,\chi')) \in E(\Hmc_i)$ where
    $\chi' =\chi$ except that $\chi'(u)$ is the 
 unique continuation of $\chi(v)$ with $\mn{tail}(\chi'(u))=\mn{tp}^{c, \ell}_{\Gmc_i}(u)$;

    \item each vertex  $(v,\chi) \in V(\Hmc_i)$ carries the same vertex labels as the
    vertex $v \in V(\Gmc_i)$.
    
\end{itemize}
Informally, a main purpose of $\ell$-history choices is to ensure that every
vertex from $\Gmc_i$ gets duplicated exactly the same number of times in $\Hmc_i$. This number is $n_H$, the number of $\ell$-history choices for $\Gmc_i$.

For $i \in \{1,2\}$, choose $u_i$ to be some vertex $(v_i,\chi) \in V(\Hmc_i)$ such that $\chi(v_i) = \mn{tp}^{c, \ell}_{\Gmc_i}(v_i)$.
We show the following, for all $i \in \{1,2\}$:
\begin{enumerate}

   \item  $\mn{tp}_{\Gmc_i}(v)=\mn{tp}_{\Hmc_i}(v,\chi)$  for all $(v,\chi) \in V(\Hmc_i)$;
  
    \item $(\Gmc_i,v_i) \sim_{G\%} (\Hmc_i,u_i)$;

    \item $(\Gmc_1,v) \sim^{c, 2 \ell} (\Gmc_2,v')$
      implies
      $(\Hmc_1,(v,\chi)) \sim^{c, 2 \ell} (\Hmc_2,(v',\chi'))$
    for all $(v,\chi) \in V(\Hmc_1)$ and $(v',\chi') \in V(\Hmc_2)$;

    \item each vertex $(v,\chi) \in V(\Hmc_i)$
    has unique maximal $\ell$-history $\chi(v)$;

     \item $(\Hmc_1,(v_1,\chi)) \approx_{G}^{c\downarrow, \ell}
       (\Hmc_2,(v_2,\chi'))$ whenever
       $\chi(v_1)=\chi'(v_2) = \mn{tp}^{c, \ell}_{\Gmc_1}(v_1) = \mn{tp}^{c, \ell}_{\Gmc_2}(v_2)$
       for $i \in \{1,2\}$. 
 
\end{enumerate}
Note that the equality $\mn{tp}^{c, \ell}_{\Gmc_1}(v_1) =
\mn{tp}^{c, \ell}_{\Gmc_2}(v_2)$
in Point~5 holds since $(\Gmc_1, v_1)
\sim_{G}^{c, 2\ell} (\Gmc_2, v_2)$.

The above gives the desired contradiction. From Point~2 and $\varphi(x)$
being invariant under $\sim_{G\%}$ it follows that $\Hmc_1 \models \varphi(u_1)$ and
$\Hmc_2 \not\models \varphi(u_2)$. But this is a contradiction to Point~5
because $\varphi$ is invariant under $\approx_{G}^{c\downarrow, \ell}$. The purpose of
Points~1,~3 and~4 is to support the proof of Points~2 and~5. 

\medskip \noindent 
\emph{Proof of Point~1}. Point~1 follows from the observation that
\[\{(v,(v,\chi)) \mid (v,\chi) \in V(\Hmc_i) \}\]
is a graded bisimulation between $\Gmc_i$ and $\Hmc_i$. 

\medskip \noindent 
\emph{Proof of Point~2}. Point~2
follows from Point~1 and the fact that  for every vertex $v \in V(\Gmc_i)$, there are
exactly $n_H$
vertices $(v,\chi) \in V(\Hmc_i)$.

\medskip
\noindent
\emph{Proof of Point~3}. 
Let $v \in V(\Gmc_1)$ and $v' \in V(\Gmc_2)$ such that $(\Gmc_1,v) \sim^{c, 2 \ell} (\Gmc_2,v')$. Then \mn{duplicator} has a winning strategy $S$ for the  $c$-graded $2\ell$-round bisimulation game on $\Gmc_1,\Gmc_2$ from
starting position $(v,v')$. Let $(v,\chi) \in V(\Hmc_1)$ and
$(v',\chi') \in V(\Hmc_2)$. We can use $S$ to
identify a winning strategy $S'$ for \mn{duplicator}  in the $c$-graded $2\ell$-round bisimulation game on $\Hmc_1,\Hmc_2$ from
starting position $((v,\chi),(v',\chi'))$. In fact, $S'$ simply mimics the responses of \mn{duplicator} in $S$.

Let the current game
position be $((u,\xi),(u',\xi'))$ in the second game and assume
that \mn{spoiler} plays $X \subseteq \neigh_{\Hmc_1}((u,\xi))$ with $|X| \leq c$ in a {\bf graded down} move. 
Consider
the corresponding position $(u,u')$ in the first game and set $X_0 = \{ w \mid (w,\rho) \in X \}$. By construction
of $\Hmc_1$, we have $X_0  \subseteq \neigh_{\cG_1}(u)$ and $|X| = |X_0|$.
Thus  \mn{spoiler} can play $X_0$ in the first game using the same
kind of move, and \mn{duplicator} has a response $Y_0 \subseteq \neigh_{\cG_2}(u')$ with $|Y_0| = |X_0|$.
In the second game, \mn{duplicator} then plays the set $Y$ that consists
of all elements $(w,\rho)$ such that
$w \in Y_0$ and $\rho=\xi'$ except that
$\rho(w)$ is the unique continuation of $\xi'(u')$ with $\mn{tail}(\rho(w))=\mn{tp}^{c, \ell}_{\Gmc_{2}}(w)$.
Clearly, $|Y| = |X|$.

As part of the same move, \mn{spoiler}
then chooses an element $(w',\xi'') \in Y$. In the first game, they may choose
$w' \in Y_0$, and \mn{duplicator} has a response $w \in X_0$. By definition of $X_0$ there is
a $(w,\rho) \in X$. By construction of $\Hmc_1$ and choice of $X$,
there is in fact a unique such $(w,\rho)$ with
$\rho=\xi$ except that
$\rho(w)$ is the unique continuation of $\xi(u)$ with $\mn{tail}(\rho(w))=\mn{tp}^{c, \ell}_{\Gmc_{1}}(w)$.
In the second game, \mn{duplicator}
replies with choosing $(w,\rho)$.
The case where \mn{spoiler} plays a set $X \subseteq \neigh_{\Hmc_2}((u',\xi'))$ is symmetric.
It can be
verified that the described strategy $S'$ is indeed winning.

\medskip \noindent
\emph{Proof of Point~4}. We have to argue that if  $(v,\chi) \in V(\Hmc_i)$ and
$h$ is a maximal $\ell$-history of  $(v,\chi)$ in $\Hmc_i$, then $h=\chi(v)$.
Thus let $(v,\chi) \in V(\Hmc_i)$, and let $\chi(v)=t_1,\dots,t_k$. 
Take any path $(u_1, \chi_1),\dots, (u_m, \chi_m)=(v,\chi)$ in
$\Hmc_i$ that gives rise to a maximal $\ell$-history
$h=t'_1,\dots,t'_m$ of $(v,\chi)$. 
Note that $\chi_m(u_m)=\chi(v)$ and, by construction of $\Hmc_i$, $\mn{tp}^{c, \ell}_{\Hmc_i}((v, \chi)) = \mn{tail}(\chi(v))$.

For $0 \leq z \leq k$, let $\mn{head}_z(\chi(v))$ denote
the result of dropping from $\chi(v)$ all but the trailing $z$ types
and let $\mn{tail}_z(\chi_j(u_j))$ denote the result of dropping
from $\chi_j(u_j)$ all but the leading $z$ types. 
Using the construction of $E(\Hmc_i)$ and the fact that $\chi_m(u_m)=\chi(v)$, one
may verify that one of the following must hold:
\begin{enumerate}

\item[(a)] $m < \ell$ and $\chi_z(u_z) = \mn{head}_{z}(\chi(v))$ for $1 \leq z \leq m$;

\item[(b)] $m= \ell$ and $\mn{tail}_z(\chi_z(u_z)) = \mn{head}_{z}(\chi(v))$ for $1 \leq z \leq m$. 

\end{enumerate}
In both cases, by definition of $V(\Hmc_i)$ and Point~1, this implies
$t_{k-z}=t'_{m-z}$ for $0 \leq z < \min\{k,m\}$. 
It thus remains
to show that $m=k$.

First assume to the contrary that $m>k$. Then we must clearly have $k < \ell$.
Since $\chi_m(u_m)=\chi(v)$, the length of $\chi_m(u_m)$ is strictly smaller than $\ell$. But then the length of $\chi_j(u_j)$ is strictly smaller than
that of $\chi_{j+1}(u_{j+1})$ for $1 \leq j < m$. As a consequence, $m < \ell$. We are thus in Case~(a)
above and therefore $\chi_{m-(k-1)}(u_{m-(k-1)})$ must be a sequence of length~1. By construction
of $E(\Hmc_i)$, the vertex $(u_{m-{(k-1)}}, \chi_{m - (k - 1)})$ therefore has no predecessor in
$\Hmc_i$, in contradiction to the fact that $(u_{m - k}, \chi_{m - k})$ is such a
predecessor.

Now assume that $k>m$. 
Then clearly $m < \ell$. We first observe that since $h$
is a maximal $\ell$-history, $\chi_1(u_1)$ is a sequence of length~1.  Assume
otherwise. By definition of $V(\Hmc_i)$, this implies that $u_1$ has a
predecessor $w$ in $\Gmc_i$ which has $\mn{head}_1(\chi_1(u_1))$ as an
$\ell$-history. But then by definition of $E(\Hmc_i)$,
$(w,\rho)$ is a predecessor of $(u_1, \chi_1)$ in $\Hmc_1$ for every $\rho$ such that
$\rho=\chi_1$ except that $\chi_1(u_1)$ is the unique continuation of $\rho(w)$ with $\mn{tail}(\chi_1(u_1))=\mn{tp}^{c, \ell}_{\Gmc_i}(u_1)$. There is clearly at least one such $\rho$. But since $m < \ell$
this implies that $\mn{tp}^{c, \ell}_{\Gmc_i}(w),h$ is 
an $\ell$-history of $(v,\chi)$ in
$\Hmc_i$, in contradiction to $h$ being maximal. We have thus shown that $\chi_1(u_1)$ is a sequence
of length~1. But then in both Case~(a) and Case~(b) above, $\mn{head}_{m-1}(\chi(v))$ is a
sequence of length~1, implying that $k=m$. Contradiction.

\medskip
\noindent
\emph{Proof of Point~5}. Let  $i \in \{1,2\}$. We show that
\begin{enumerate}

\item[(a)] for every vertex $(v,\chi) \in V(\Hmc_i)$, there is a vertex
  $(v',\chi') \in V(\Hmc_{3-i})$ such that \mn{duplicator} has a winning
  strategy in the up-ungraded down-$c$-graded $\ell$-round
  bisimulation game on $\Hmc_1,\Hmc_2$ starting from position $((v,\chi),(v',\chi'))$, and
  
\item[(b)] if $v=v_1$ and $\chi(v_1)=\mn{tp}^{c, \ell}_{\Gmc_1}(v_1)$, then we can
  choose $(v',\chi')$ such that
  $v'=v_2$ and $\chi'(v_2)=\mn{tp}^{c, \ell}_{\Gmc_2}(v_2)$. 

\end{enumerate}
This clearly yields Point~5. Thus let $(v,\chi) \in V(\Hmc_i)$. We want to choose a vertex $(v',\chi') \in V(\Hmc_{3-i})$ that satisfies the conditions in Points~(a) and~(b).
Importantly, we want to choose $(v',\chi')$ so that 
\begin{enumerate}

    \item[(c)] $(\Gmc_i,v) \sim^{c,\ell} (\Gmc_{3-i},v')$ and

    \item[(d)] $\chi(v)=\chi'(v')$.
    
\end{enumerate}
By Point~4, $(v,\chi)$ has unique maximal $\ell$-history $\chi(v)$ in $\Hmc_i$.
Let $\chi(v)=t_1,\dots,t_k$. We thus find a path $x_1,\dots,x_k=(v,\chi)$ in
$\Hmc_i$ such that $\mn{tp}^{c, \ell}_{\Hmc_i}(x_r)=t_r$ for
$1 \leq r \leq k$, and $k=\ell$ or $x_1$ has no predecessor in $\Hmc_i$.
Let
$x_r = (u_r,\chi_r)$ for $1 \leq r \leq k$. Since
$(\Gmc_1, v_1) \sim_{G}^{c, 2\ell} (\Gmc_2, v_2)$, we can
find in $\Gmc_{3-i}$ a vertex $w$ such that
$(\Gmc_i, u_1) \sim^{c, 2\ell} (\Gmc_{3-i}, w)$. Note that by definition of $\ell$-histories, $h'=\mn{tp}_{\Gmc_{3-i}}^\ell(w)$ is a (possibly not maximal) $\ell$-history of $w$ in
$\Gmc_{3-i}$. 
By definition of $\Hmc_{3-i}$, we thus find a vertex
$y_1=(w,\chi')$ in $\Hmc_{3-i}$ such 
that $\chi'(w)=h'$.  From
Point~3, we obtain
$(\Hmc_1,x_1) \sim^{c, 2\ell} (\Hmc_2,y_1)$. Exercising the
\textbf{graded forth} property~$k$ times, following the path $x_1,\dots,x_k$ in
$\Hmc_i$, we find in $\Hmc_{3-i}$ a path $y_1,\dots,y_k$ such that
$(\Gmc_i, x_r) \sim^{c, 2\ell-r} (\Gmc_{3-i}, y_r)$ for
$1 \leq r \leq k$. As the vertex $(v',\chi')$ in Point~(a) above, we
choose $y_k$. Note that Point~(b) above is satisfied automatically.
Also note that Points~(c) and~(d) are satisfied. For Point~(d), this is the
case because, by choice of $(v',\chi')$, this vertex has $\ell$-history
$\chi(v)$ in $\Hmc_{3-i}$, no matter
whether $k=\ell$ or $x_1$ has no predecessor in $\Hmc_i$. By Point~(4),
this implies that $\chi'(v')=\chi(v)$, as required.

\smallskip

It remains to argue that \mn{duplicator} has a winning strategy $S'$ in the up-ungraded down-$c$-graded $\ell$-round bisimulation game on
$\Hmc_i,\Hmc_{3-i}$ from position $(v,\chi),(v',\chi')$. Key to this 
strategy is to ensure the following invariants: if
$(u_1,\xi_1),(u_2,\xi_2)$ is the position reached after $m \in
  \{0,\dots,\ell \}$ rounds, then 
\begin{itemize}

\item[(i)] $(\Gmc_1,u_1) \sim^{c,\ell-m} (\Gmc_2,u_2)$ and
  
\item[(ii)] the length $\ell-m$-suffixes 
  of $\xi_1(u_1)$ and $\xi_2(u_2)$ are identical.

\end{itemize}
Note that these invariants are satisfied at the initial position due to Points~(c) and~(d). We now describe~$S'$. Let the current game
position be $(u_1,\xi_1),(u_2,\xi_2)$ and assume
that \mn{spoiler} plays $X \subseteq \neigh_{\Hmc_i}((u_i,\xi_i))$ with $|X| \leq c$ in a {\bf graded down} move, $i \in \{1,2\}$. By Invariant~(i), \mn{duplicator} has a winning strategy $S$ in the down-$c$-graded $\ell-m$-bisimulation game starting in position $(u_1, u_2)$.  
Consider
the set $X_0 = \{ w \mid (w,\xi) \in X \}$. By construction
of $\Hmc_i$, we have $X_0\subseteq \neigh_{\Gmc_i}(u_i)$ and $|X_0| \leq c$
and thus  \mn{spoiler} can play $X_0$ in $S$ using the same
kind of move.  \mn{duplicator} has a response $Y_0 \subseteq \neigh_{\Gmc_{3 - i}}(u_{3-i})$.
For $S'$, \mn{duplicator} then plays the set $Y$ that consists
of all elements $(w,\rho)$ such that
$w \in Y_0$ and $\rho=\xi_{3 - i}$ except that
$\rho(w)$ is the unique continuation of $\xi_{3-i}(u_{3-i})$ with $\mn{tail}(\rho(w))=\mn{tp}^{c, \ell}_{\Gmc_{3-i}}(w)$. 
As part of the same move, \mn{spoiler}
then chooses an element $(w',\xi') \in Y$. In $S$, they may choose
$w' \in Y_0$, and \mn{duplicator} has a response $w \in X_0$. By definition of $X_0$ there is
a $(w,\rho) \in X$. 
In $S'$, \mn{duplicator}
replies with choosing such a $(w,\rho)$.

The second case
is that  \mn{spoiler} plays $(w_i,\rho_i) \in V(\Hmc_i)$, in an {\bf
  ungraded up} move, $i \in \{1,2\}$. By Invariant~(ii), the length
$\ell-m$-suffixes
of $\xi_1(u_1)$ and $\xi_2(u_2)$ are identical. By Point~4
and since the choice of $(w_i,\rho_i)$ shows that
$(u_i,\xi_i)$ has a predecessor in $\Hmc_i$,
also $(u_{3-i},\xi_{3-i})$ has a predecessor
$(w_{3-i},\rho_{3-i})$ in $\Hmc_{3-i}$. Moreover, by (ii) and Point~4, the length $\ell-(m+1)$-suffix of $\rho_{3-i}(w_{3-i})$ is independent of which
of the possibly many predecessors we choose
and in particular, Invariant~(ii) is again satisfied.
It is also easy to see that Invariant~(ii) being satisfied implies that so is Invariant~(i):
as the $0$-suffixes of $\rho_1(w_1)$ and $\rho_2(w_2)$ must be identical, it
follows from Point~4 that $\mn{tp}^{c, \ell}_{\Hmc_1}( (w_1, \rho_1) ) =
\mn{tp}^{c, \ell}_{\Hmc_2}((w_2, \rho_2))$ and from Point~1 and Lemma~\ref{lem:bisim-implies-winning-strategy} that
$\mn{tp}^{c, \ell}_{\Gmc_1}(w_1) = \mn{tp}^{c, \ell}_{\Gmc_2}(w_2)$. 

At this point, it is not hard to 
verify that the described strategy $S'$ is indeed winning.
\end{proof}

Using Lemmas~\ref{lem:ratio-upgrade-inverse-bisim-to-fo}, \ref{lem:bisim-upgrade-up-ungraded-to-up-graded}, and~\ref{lem:ratio-upgrade-bisim-to-inverse-bisim}, we can now complete the proof of Theorem~\ref{lem:ratio-bisim-invariance}.

\lemratiobisiminvariance*

\begin{proof} ``$1 \Rightarrow 2$''.
    If $\varphi$ is invariant under $\sim_{G\%}$, then by
    Lemma~\ref{lem:ratio-upgrade-inverse-bisim-to-fo}, there are $c, \ell$ such
    that $\varphi$ is invariant under $\approx_{G}^{c, \ell}$. 
    By Lemma~\ref{lem:bisim-upgrade-up-ungraded-to-up-graded}, $\varphi$ is then also invariant under
    $\approx_{G}^{c\downarrow, \ell}$.   
    It then follows
    from Lemma~\ref{lem:ratio-upgrade-bisim-to-inverse-bisim} that there are
    also $c', \ell'$ such that $\varphi$ is invariant under $\sim_{G}^{c',
    \ell'}$.

    By Lemma~\ref{lem:gmlg-bisim-equivalence}, $\sim_{G}^{c', \ell'}$ has a finite
    number of equivalence classes, and each equivalence class can be defined
    using a $\GMLG$-formula.
    Thus, a $\varphi' \in \GMLG$ with $\varphi \equiv
    \varphi'$ can be obtained by taking the disjunction of the formulae that
    define the equivalence classes which are models of $\varphi$.

    \medskip
    
    \noindent
    ``$2 \Rightarrow 1$'':
    If $\varphi \equiv \varphi'$ for some $\varphi' \in
    \GMLG$, then by Lemma~\ref{lem:gmlg-invariance},
    there are $c, \ell$ (which can be determined from $\varphi'$)  such that 
    $\varphi$ is invariant under $\sim_{G}^{c, \ell}$.
    Thus, by Lemma~\ref{lem:bisim-implies-winning-strategy}, $\varphi$ is also invariant under $\sim_{G}$ and $\sim_{G\%}$.
\end{proof}

A slight variation of the above proof also shows the following.

\begin{corollary}\label{prop:standard-bisim-invariance}
    For every $\FO$-formula $\varphi(x)$ over $\Pi$, the following are equivalent:
    \begin{enumerate}
        \item $\varphi$ is invariant under $\sim_{G}$;
        \item $\varphi$ is equivalent to a $\GMLG$-formula over all
          (finite!) $\Pi$-labeled pointed graphs.
        \end{enumerate}
\end{corollary}

To see this, note that every $\FO$-formula that is invariant under $\sim_{G}$ is also
invariant under $\sim_{G\%}$, so the ``$1 \Rightarrow 2$'' direction follows from Theorem~\ref{lem:ratio-bisim-invariance}. The ``$2 \Rightarrow 1$'' direction
simply follows from the fact that all $\GMLG$-formulae are invariant under $\sim_{G}$.

\medskip
We note that Theorem~\ref{lem:ratio-bisim-invariance} and Corollary~\ref{prop:standard-bisim-invariance} also have  interesting applications regarding GNNs with global readout. In particular,   \cite{Barcelo_GNNs} left open the non-trivial question of the precise expressive power of GNNs with global readout, relative to FO; see also \cite{DBLP:journals/corr/abs-2508-06091}. It is easy to see that GNN+Gs are invariant
under $\sim_G$ (recall that GNN+G refers to GNNs with \emph{non-counting} global readout), and thus Corollary~\ref{prop:standard-bisim-invariance} together with a straightforward translation of $\GMLG$ into $\GNNG$ in the style of \cite{Barcelo_GNNs} (see also the proof of Lemma~\ref{lem:gmlg-to-real-gps}) yields the
following.
\begin{corollary}
   Relative to $\FO$, $\GNNG$ and $\GMLG$  have the same expressive power.
\end{corollary}
Another relevant variation is to consider arithmetic mean as the aggregation function for global readout and to stick with sum for local aggregation. Let us use GNN+GM to refer to this version of GNNs. One can verify that GNN+GM is invariant under  $\sim_{G\%}$. Thus, 
Theorem~\ref{lem:ratio-bisim-invariance} together with a 
straightforward variation of the proof of Lemma~\ref{lem:gmlg-to-real-gps} yields the
following.
\begin{corollary}
   Relative to $\FO$, \textnormal{GNN+GM} and $\GMLG$  have the same expressive power.
\end{corollary}
The expressive power of $\GNNGC$ remains an interesting open problem.

\subsection{Proof of Theorem~\ref{thm:real-GT-PLG}}

\thmrealgtplg*

To capture the $\FO$ vertex properties that can be expressed by $\GT$s, we define an
equivalence relation on graphs under which $\GT$s are invariant.

\begin{definition}[Label-ratio equivalence]\label{def:label-ratio-equiv}
    Two $\Pi$-labeled pointed graphs $(\Gmc_1, v_1)$, $(\Gmc_2, v_2)$ are
    \textbf{label-ratio equivalent}, written $(\Gmc_1, v_1) \sim_{\lambda\%}
    (\Gmc_2, v_2)$, if $\lambda_1(v_1) = \lambda_2(v_2)$ and there exists a
    rational number $q > 0$ such that for all $t \subseteq \Pi$,
    \[
        |\{v \in V(\Gmc_1) \mid \lambda_1(v) = t\}| = q \cdot |\{v \in V(\Gmc_2) \mid \lambda_2(v) = t\}|.
    \]
\end{definition}

Intuitively, Definition~\ref{def:label-ratio-equiv} is that of global-ratio graded
bisimilarity without the \textbf{graded forth} and
\textbf{graded back} conditions of graded bisimulations. Directly from the definitions of the two relations, one can show that
invariance under $\sim_{\lambda\%}$ implies invariance under $\sim_{G\%}$.

In the proof of Lemma~\ref{lem:gps-invariant-under-grg-bisim}, the 
\textbf{graded forth} and \textbf{graded back} conditions are only used in the part of the
proof that is concerned with message passing modules. As $\GT$s do not contain
 message passing modules, the proof of Lemma~\ref{lem:gps-invariant-under-grg-bisim} also yields the following.

\begin{proposition}\label{prop:gts-invariant-under-label-ratio}
Let $T$ be a soft-attention or average hard-attention $\GT$. Then $T$ is
invariant under $\sim_{\lambda\%}$.
\end{proposition}

We now show a counterpart of Theorem~\ref{lem:ratio-bisim-invariance}
for~$\sim_{\lambda\%}$.

\begin{lemma}\label{lem:upgrade-label-ratio-to-bisim}
    Every $\FO$-formula $\varphi(x)$ that is invariant under $\sim_{\lambda\%}$ is also
    invariant under $\sim_G^{0, 0}$.
\end{lemma}
\begin{proof}
    Let $\varphi(x)$ be an $\FO$-formula that is
    invariant under $\sim_{\lambda\%}$. Further assume, to the contrary of what
    we have to show, that there are $\Pi$-labeled pointed graphs $(\Gmc_1, v_1)$, $(\Gmc_2,
    v_2)$ such that $(\Gmc_1, v_1) \sim_G^{0, 0} (\Gmc_2, v_2)$,
    $\Gmc_1 \models \varphi(v_1)$, and $\Gmc_2 \not\models \varphi(v_2)$.

    Since $\varphi$ is invariant under $\sim_{\lambda\%}$, it is also invariant
    under $\sim_{G\%}$. We can thus first apply
    Lemma~\ref{lem:ratio-upgrade-inverse-bisim-to-fo} to conclude that there are
    $c, \ell$ such that $\varphi$ is invariant under $\approx_G^{c, \ell}$, Lemma~\ref{lem:bisim-upgrade-up-ungraded-to-up-graded} to conclude that $\varphi$ is invariant under $\approx_G^{c\downarrow, \ell}$ and
    then Lemma~\ref{lem:ratio-upgrade-bisim-to-inverse-bisim} to conclude that
    there are $c', \ell'$ such that $\varphi$ is invariant under $\sim_G^{c',
    \ell'}$.

    From $\Gmc_1$ and $\Gmc_2$, we now construct graphs $\Hmc_1$, $\Hmc_2$ such that
    \begin{enumerate}
        \item $(\Gmc_i, v_i) \sim_{\lambda\%} (\Hmc_i, v_i)$ for all $i \in \{1, 2\}$;
        \item $(\Hmc_1, v_1) \sim_G^{c', \ell'} (\Hmc_2, v_2)$.
    \end{enumerate}

    This shows the desired contradiction. Point~1 and the fact that $\varphi(x)$ is
    invariant under $\sim_{\lambda\%}$ implies that $\Hmc_1 \models
    \varphi(v_1)$ and $\Hmc_2 \not\models \varphi(v_2)$. This then contradicts that
    $\varphi$ is invariant under $\sim_G^{c', \ell'}$ and Point~2.

    For $\Gmc_1 = (V_1, E_1, \lambda_1)$ and $\Gmc_2 = (V_2, E_2, \lambda_2)$,
    the graphs $\Hmc_1$ and $\Hmc_2$ can simply be obtained by setting
    $\Hmc_1 = (V_1, \emptyset, \lambda_1)$ and $\Hmc_2 = (V_2, \emptyset,
    \lambda_2)$, that is, removing all edges from $\Gmc_1$ and $\Gmc_2$.
    As $\sim_{\lambda\%}$ only considers $\lambda_1$ and $\lambda_2$, one can
    then verify that Point~1 holds, via the rational number $q = 1$.
    For Point~2, observe that no vertex in $\Hmc_1$ and $\Hmc_2$ has any
    successors, which means that
    there is no position in
    which \mn{spoiler} can make any moves. As in addition $(\Gmc_1, v_1)
    \sim_G^{0, 0} (\Gmc_2, v_2)$ implies that $(\Hmc_1, v_1) \sim_G^{0, 0}
    (\Hmc_2, v_2)$, Point~2 follows.
\end{proof}

\begin{lemma}\label{lem:label-ratio-invariant-fragment-is-plg}
For any $\FO$-formula $\varphi(x)$ over $\Pi$, the following are equivalent:
\begin{enumerate}
    \item $\varphi$ is invariant under $\sim_{\lambda\%}$;
    \item $\varphi$ is equivalent to a $\PLG$-formula $\varphi'$ over all $\Pi$-labeled pointed graphs.
\end{enumerate}
\end{lemma}

\begin{proof} ``$1 \Rightarrow 2$''. If $\varphi$ is invariant under 
    $\sim_{\lambda\%}$, then by Lemma~\ref{lem:upgrade-label-ratio-to-bisim},
    $\varphi$ is also invariant under $\sim_G^{0, 0}$.
    By Lemma~\ref{lem:gmlg-bisim-equivalence}, $\sim_G^{0, 0}$ has a finite number of equivalence
    classes, and each equivalence class can be defined using a $\PLG$-formula.
    Thus, a $\PLG$-formula $\varphi'$ with $\varphi \equiv \varphi'$ can be obtained
    by taking the disjunction of the formulae that define the equivalence
    classes that contain pointed graphs which satisfy $\varphi$.

    \medskip

    \noindent ``$2 \Rightarrow 1$''.
    If $\varphi \equiv \varphi'$ for some $\varphi' \in \PLG$, then by
    Lemma~\ref{lem:gmlg-invariance}, $\varphi$ is invariant under $\sim_G^{0,
    0}$. Thus, $\varphi$ is also invariant under $\sim_{\lambda\%}$.
\end{proof}

The first direction of Theorem~\ref{thm:real-GT-PLG} can now be proved by
combining Proposition~\ref{prop:gts-invariant-under-label-ratio} with
Lemma~\ref{lem:label-ratio-invariant-fragment-is-plg}. For the second direction,
observe that in the proof of Lemma~\ref{lem:gmlg-to-real-gps}, message passing
layers are only required to express subformulae of the form $\Diamond_{\geq k}
\varphi$. All other types of subformulae can be expressed solely with MLPs and
(soft-attention or average hard-attention) self-attention heads. Hence, for
every $\PLG$-formula, there is also an equivalent $\GT$.

\section{Floating-point preliminaries}

\subsection{Floating-point numbers and arithmetic}\label{appendix: floats}

Here we define more formally the floating-point arithmetic operations used in this paper. 

Let $\cF$ be a floating-point format.
The arithmetic operations $+$, $-$, $\cdot$ and $\div$ over a floating point format $\cF$ are functions of the form $\cF \times \cF \to \cF$ and they are computed as follows. Let $\star$ be one of these operations. First, as discussed in the main section, if one of the inputs is $\NaN$, the output is also $\NaN$. Otherwise, we take the precise operation of $\star$ w.r.t. the real arithmetic extended with $\infty$ and $-\infty$ and then round the precise result by using the ``round to nearest, ties to even'' method, which means that we round to the nearest number in the format $\cF$ as though there was no upper bound for the exponent (i.e., we allow numbers that exceed the maximum exponent in $\cF$). With ties we round to the number with an even least significant digit. If the maximum exponent is exceeded, we round to $\infty$ or $-\infty$ depending on the sign. 
Analogously, we define the arithmetic operation $\sqrt{x}$ over $\cF$ which is a function of the type $\cF \to \cF$. Moreover, in the case where the operation leads to an undefined number, i.e., $\frac{\pm \infty}{\pm \infty}$, $\frac{\pm \infty}{\mp \infty}$, $0 \cdot \pm \infty$, $\pm \infty \mp \infty$, the output is $\NaN$.
In the IEEE754 standard, it is suggested that these operations, $+$, $-$, $\cdot$, $\div$ and $\sqrt{x}$, are taken as basic operations, and other operations can be defined in terms of these operations (or taken directly as basic operations).

We already discussed an implementation of the average hard-attention function with floats, where the denominator is obtained via the floating-point sum. Another possible technique would be to calculate the denominator by rounding the real size of the set into $\cF$ directly, but it is not clear whether integers greater than any float in $\cF$ should round down to the greatest non-infinite float in $\cF$ or result in overflow. Our characterizations hold in the former case, which we will discuss in the proofs, and we leave the latter as an interesting open question.

We note that our definition of floating-point formats includes two floats that represent the number zero: one with a positive sign and one with a negative sign. For the operations, we preserve the sign in the preceding calculations if possible, and otherwise default to positive zero. Whether we include one or both zeros in formats does not change our results on expressive power.

We assume that the exponential function $\exp(x)$ over $\cF$ is defined by using range reductions and polynomial approximations; this is done in a similar way as in \citep{DBLP:journals/toms/Tang89} and in the popular math library fdlibm as follows.\footnote{Moreover, the exponent function was analyzed from the circuit complexity perspective and defined in an analogous way in \citep{HESSE2002695, chiang2025transformersuniformtc0}.}
Informally, the algorithm consists of four steps:
\begin{enumerate}
    \item \textbf{Check for exceptions:} If $x$ is too large, the output is $\infty$, and if $x$ is too small, the output is $0$. When $x$ is too close to zero, the output is $1$. If $x$ is $\NaN$, the output is $\NaN$.
    \item \textbf{Range reduction:} Given an $x \in \cF$, we first compute an integer $k = \lfloor \frac{x}{\ln 2} \rfloor \in \cF$ and $r = x - k \log 2$. 
    \item \textbf{Polynomial approximation:} We approximate $\exp(r)$ by using the Taylor approximation of $\exp(x)$ of a small degree. As shown in \citep{taylor-approximation}, even the Taylor series of $\exp(x)$ of degree $6$ suffices for accurate results with small input values. Another popular choice is the Remes algorithm to find a small polynomial that approximates $\exp(x)$ with small input values. Thus, we can assume that the polynomial used for the approximation is fixed and has some constant degree $c \in \N$.
    \item \textbf{Combination:} We set that $\exp(x) = 2^k \cdot \exp(r)$.
\end{enumerate}

Step $3$ is often evaluated by using addition, multiplication and Horner's rule, i.e., a polynomial 
\[
a_0 + a_1 x + \cdots + a_n x^n,
\]
over $\cF$, where $a_0, \dots, a_n \in \cF$,
is evaluated as 
\[
a_0 + x\big(a_1 + x(a_2 + \cdots + x( a_{n-1} + x a_n) \cdots ) \big).
\]
In the last step of the algorithm, $2^k$ is trivial to compute, since $k$ is an integer and the base of the format is $2$. More complicated implementations are also possible for us, e.g., we could directly implement the source code of the exponent function of the math library fdlibm.

\subsection{Interpreting labeling functions as float feature maps}\label{appendix: Interpreting labeling functions as float feature maps}

In this section, we discuss another type of translation for floats, where the labeling of a graph may be interpreted in a different way.

Given a floating-point format $\cF(p, q)$, an \textbf{$(\cF^d, \Pi)$-labeled graph} 
$(V, E, \lambda)$ 
refers to a $\Pi$-labeled graph with $\abs{\Pi} = d(p + q + 1)$; in this case, the labeling function $\lambda$ can be identified with a floating-point feature map $\lambda_\cF \colon V \to \cF^d$ in the following natural way.
We first split $\Pi$ into $d$ subsets $P_1, \dots, P_d$ of equal size such that $P_1$ contains the $p+q+1$ least elements of $\Pi$ (with respect to $<^\Pi$), $P_2$ contains the next $p+q+1$ elements, and so forth. For each $P_i$, we can order its elements with respect to $<^\Pi$ into a sequence $s_i$ and interpret this sequence in each vertex as a floating-point number over $p$ and $q$ based on which symbols the vertex is labeled with. For each such float (which may be neither normalized, subnormalized, $\infty$, $-\infty$ nor $\NaN$), there is a corresponding float in $\cF$ that is interpreted as the same real value. Thus, $(s_1, \dots, s_d)$ can be interpreted as a vector in $\cF^d$ depending on which symbols a given vertex is labeled with.
We can leave $\Pi$ implicit and refer to $(\cF^d, \Pi)$-labeled graphs as $\cF^d$-labeled graphs.
Thus, a computing model over $\cF$ (e.g., a $\GPSF$-network or a $\GTF$) can run over such a graph 
by interpreting its labeling function as a floating-point feature map, instead of transforming the labeling function into the corresponding binary valued feature map.

\subsection{Other classification heads}\label{appendix: Other classification heads}

Here we discuss definitions for types of classification other than Boolean vertex classification. In particular, we consider Boolean graph classification, non-Boolean vertex classification and non-Boolean graph classification.

First, we consider graph classification as opposed to vertex classification. 
A \textbf{Boolean graph classification head} 
is a readout gadget where the $\MLP$ is a Boolean vertex classification head.
As with Boolean vertex classifiers, note that the $\MLP$ of the readout gadget is not assumed to be $\ReLU$-activated, meaning that it can use, e.g., step functions. For example, \citep{grohe_transformers} use readout gadgets as final classification heads when graph classification tasks are considered. 

Both Boolean vertex classifiers and Boolean graph classifiers can be generalized further. Any $\MLP$ can be considered a \textbf{(general) vertex classifier}. Likewise, any readout gadget can be considered a \textbf{(general) graph classifier}. Here instead of $0$ or $1$, we classify vertices and graphs using feature vectors, i.e., the output dimension can be any positive integer.

Now we can modify all the variants of $\GPS$-networks, $\GT$s and $\GNN$s by replacing the Boolean vertex classifiers with any of the classifiers discussed here with the same input dimension as the Boolean vertex classifier. The above classifiers are defined analogously for floats, and we can modify $\GTF$s, $\GPSF$-networks and $\GNNF$s analogously.

We also define fragments of our logics that exclusively define graph properties.
Let $\cL$ be one of the logics discussed in Section~\ref{section: logics}. The set of \textbf{$\Pi$-formulae $\psi$ of $\cL^*$ }is defined according to the following grammar:
\[
    \psi \coloncolonequals \DiamondG_{\geq k} \varphi \,|\, \neg \, \psi \,|\, \psi \land \psi \,|\, \DiamondG_{\geq k} \psi,
\]
where $\varphi$ is a $\Pi$-formula of $\cL$. The semantics of $\cL^*$ is defined in the natural way.

\subsection{Equivalence of graph classifiers and general classifiers}
\label{appendix: Other notions on equivalence}

Here we introduce concepts of equivalence and expressive power that account for graph classification and non-Boolean classification.

A \textbf{graph property} is simply a vertex property $\lambda \colon V \to \{0,1\}$ such that for some $b \in \{0,1\}$, $\lambda(v) = b$ for all $v \in V$.
The concepts of equivalence between learning models and logics from Section~\ref{sec: equivalence} extend in a natural way for graph classification by considering graph properties and graph classifiers instead of vertex properties and vertex classifiers.

An \textbf{$m$-ary feature update over $\Pi$} is simply an isomorphism invariant mapping $U$ that takes a $\Pi$-labeled graph $\cG = (V, E, \lambda)$ as input and outputs a new feature map $\lambda' \colon V \to \{0,1\}^m$.
In the case of floats, given a floating-point format $\cF$ and $p, q \in \N$, an \textbf{$m$-ary feature update over $\cF^p$} is simply an isomorphism invariant mapping $U_{\cF}^p$ that takes an $\cF^p$-labeled graph $\cG = (V, E, \lambda)$ as input and outputs a new feature map $\lambda' \colon V \to \cF^m$. In the case with floats we can leave the set of vertex labels implicit and omit it. From the perspective of logics, a feature update is just a query.

Note that our computing models with general classification heads (e.g., $\GPS$-networks and $\GT$s with non-Boolean vertex classification heads) are essentially just feature updates. Analogously, sequences of formulae of our logics can be seen as classes of feature updates, i.e., a sequence $(\varphi_1, \ldots, \varphi_k)$ of $\Pi$-formulae of a logic $\cL$ defines a $k$-ary feature update over $\Pi$.

Let $\mathrm{id}_{\cF} \colon \{0,1\}^* \to \cF^*$ be a function that maps each binary string $\bb$ to a floating point string $\bbf$ of equal length such that $\bbf(i) = b_\cF$ iff $\bb(i) = b$, where $b_\cF$ denotes the corresponding float string of $b$.
Given a feature update $U_1$ over $\Pi$ and a feature update $U_2$ over $\cF^d$, we say that $U_2$ is \textbf{equivalent to $U_1$} (w.r.t. $\mathrm{id}_\cF$), if for each $\Pi$-labeled graph $(V, E, \lambda)$,
the feature map $U_2(V, E, \mathrm{id}_\cF(\lambda)) $ is the same as
$\mathrm{id}_\cF(U_1(V, E, \lambda))$. 
Respectively, a feature update $U_1$ over $\Pi$ is \textbf{equivalent to} a feature update $U_2$ over $\cF^d$, if $U_1$ defines the same feature update as $U_2$ over $\Pi$-labeled graphs (recall that floats are just binary strings).

Given a class $\cL$ of feature updates over $\Pi$ (e.g. sequences of formulae of a logic) and a class $\cC$ of feature updates over $\cF^d$ for any $d$ (e.g. $\GT$s with non-Boolean vertex classification heads), we say that $\cL$ and $\cC$ \textbf{have the same expressive power} (w.r.t. the general classification), if for each feature update $U \in \cL$ there is an equivalent feature update $U' \in \cC$, and vice versa.
Analogously, two classes $\cC_1$ and $\cC_2$ of computing models \textbf{have the same expressive power}, if for each feature update $U_1 \in \cC_1$ there is a feature update $U_2 \in \cC_2$ that defines the same feature update, and vice versa.

\section{Proofs for Section \ref{Characterizing float-based transformers}}\label{appendix: Proofs for float section}

In this section, we give the full proofs of the results in Section~\ref{Characterizing float-based transformers}.
The translations in this section that are given in terms of vertex classification also generalize for graph classification. In the case of logics, this means replacing the logic $\cL$ appearing in a result with the logic $\cL^*$\footnote{Alternatively, we could simply restrict the concept of equivalence to those formulae of $\cL$ that define graph properties, but $\cL^*$ gives us a proper syntax.} (see Appendix~\ref{appendix: Other classification heads} for the definition of $\cL^*$ and Appendix \ref{appendix: graph classification} for the formal analysis how our results generalize for graph classification). The translations that are given in terms of Boolean classification also generalize for general classifiers (see Appendix \ref{appendix: graph classification}). In the case of logics, this means considering sequences of formulae instead of a single formula.

Regarding the results of this section where we translate logic formulae into $\GTF$s, $\GPSF$-networks, etc., we make one very important assumption, namely that that the inputs of $\GTF$s, transformer layers, $\GPSF$-networks, $\GPS$-layers, $\GNNF$s, message-passing layers, $\MLPF$s, perceptron layers, etc. \emph{are always non-negative}. Their outputs may contain negative values, but such values are always turned non-negative before they are given as input to anything else. Note that this does not restrict our translations from logics to computing models, since when simulating formulae, our architectures start with floats corresponding to $1$s and $0$s.

\subsection{Proof of Theorem \ref{theorem: PLGC = SGT = AHGT}}\label{Appendix: PLGC = SGT = AHGT}

In this section, we give the proof of Theorem \ref{theorem: PLGC = SGT = AHGT}. 

\thmPLGCSGTAHGT*

We start by proving the following lemma. 

\begin{restatable}{lemma}{lemMLPPL}\label{lemma: MLP PL}
    The following have the same expressive power (w.r.t. general vertex classification): sequences of $\PL$-formulae, $\MLPF$s and $\ReLU$-activated $\MLPF$s.
\end{restatable}

We first show the translation from $\MLPF$s to $\PL$. This translation is given in terms of general classification, as described in Appendix~\ref{appendix: Other classification heads}, as $\MLP$s do not give Boolean vertex classifications by default.

\begin{lemma}\label{lemma: MLP to logic}
    For each $\MLPF$, we can construct an equivalent 
    sequence of $\PL$-formulae. For each Boolean vertex classifier on floats, we can define an equivalent $\PL$-formula.
\end{lemma}
\begin{proof}
    This follows from the Boolean completeness of $\PL$.
    First, consider that an $\MLP$ only performs floating-point operations locally, i.e., it does not involve communication between vertices. Thus, the $\MLP$ can be expressed as a function $f_\cF \colon \cF^n \to \cF^m$, where $n$ is the input dimension and $m$ the output dimension of the $\MLP$. Since floats are bit strings, this means we can interpret $f_\cF$ as a function $f_\B \colon \B^{kn} \to \B^{km}$, where $k$ is the number of bits in the floats in $\cF$. Since each bit in the output of $f_\B$ can be expressed as a Boolean combination of the bits in the input, it follows that a sequence of $\PL$-formulae can simulate the $\MLP$. In the case of Boolean vertex classifiers, a single formula is enough.
\end{proof}

We next show the translation from $\PL$ to $\ReLU$-activated $\MLP$s, which can also be derived from Theorem 15 in \citep{ahvonen_CSL}.

\begin{lemma}\label{lemma: PL to MLP}
    For each $\Pi$-formula $\varphi$ of $\PL$,
    we can construct an equivalent real or floating-point
    $\ReLU$-activated Boolean vertex classification head.
    Moreover, for each sequence of $\PL$-formulae, we can construct an equivalent $\ReLU$-activated $\MLP$.
\end{lemma}
\begin{proof}
    The idea is to split $\varphi$ into its subformulae and calculate them  one at a time in successive layers of the $\MLP$.
    
    Let $\psi_1, \dots, \psi_d$ be an enumeration of the subformulae of 
    $\varphi$
    such that 
    $\psi_1 = \varphi$.
    We construct the $\MLP$ $M$ as follows. All hidden dimensions of $M$ are $d$; the $i$th component intuitively corresponds to $\psi_i$. The number of hidden layers is the 
    formula depth of $\varphi$
    and we calculate one formula depth per layer.
    
    The first layer is constructed as follows. If $\psi_i$ is a proposition symbol, then $M$ performs the identity transformation to that proposition symbol while placing it in the correct component. If $\psi_i = \top$, then the $i$th component becomes $1$. All other components become zero.

    Next, consider the $\ell$th hidden layer. If the formula depth of $\psi_i$ is not $\ell$, then the $i$th component is copied from the previous layer. Otherwise, if $\psi_i = \neg \psi_j$ for some $j$, then the $i$th component is obtained by multiplying the $j$th component by $-1$ and adding bias $1$. If $\psi_i = \psi_j \land \psi_k$ for some $j$ and $k$, then the $i$th component is obtained by multiplying the $j$th and $k$th components and adding them together with the bias $-1$.
    It is simple to prove by induction that the vector component corresponding to the truth value of a particular formula becomes correctly calculated after a particular layer of the $\MLP$, with the first component (which corresponds to $\varphi$) becoming correctly calculated after the second-to-last layer.

    For the final layer, we copy only the 
    value of $\varphi$
    from the previous layer (i.e., we copy the first 
    component
    and ignore the rest).
\end{proof}

Now we move our focus to transformers.
We start by establishing that the $\softmax$ function saturates over floats in a similar way as the sum of a multiset of floats (recall Proposition \ref{proposition: floating-point saturation}).
For a vector $\bv$, let $\bv_k^+$ denote the set of all vectors obtained from $\bv$ by adding additional components (to the right of $\bv$) containing elements that already appear at least $k$ times in $\bv$. For example, if $\bv = (2,3,2,1,3,3)$ and $k = 2$, then $(2,3,2,1,3,3,2,2)$ and $(2,3,2,1,3,3,3,3,2)$ are in $\bv^+_2$, but $(2,3,2,1,3,3,1,2,2)$ is not in $\bv^+_2$. We let $\softmax_{\cF}$ and $\AH_\cF$ denote the implementations of $\softmax$ and $\AH$ with floats in the floating-point format $\cF$ as described in Section \ref{section: Transformers with floats}.

\begin{proposition}\label{proposition: saturation of softmax}
    For all floating-point formats $\cF$, there exists a $k \in \N$ such that for all vectors $\bv$ over floats in $\cF$, we have $\softmax_\cF(\bv)_i = \softmax_\cF(\bu)_i$ and $\AH_\cF(\bv)_i = \AH_\cF(\bu)_i$ for all $\bu \in \bv_k^+$, where $i$ ranges over the components of $\bv$.
\end{proposition}
\begin{proof}
    This follows in a straightforward way from Proposition \ref{proposition: floating-point saturation}, since adding more elements to $\bv$ only affects the denominator in the equation of $\softmax$ and $\AH$, and the denominator consists of the saturating sum $\SUM_{\cF}$.

    Next, we discuss the alternate method of implementing the average hard-attention with floats discussed in Appendix~\ref{appendix: floats},
    where the denominator is obtained by rounding the real size of the set $\cI_{\bv}$ into the floating-point format directly.
    We only need to know one of each float to determine which of the elements of $\bv$ should map to a non-zero value. Furthermore, to determine the values in these positions, only a bounded number of positions in $\bv$ can have the largest value in $\bv$ before the denominator in the description of $\AH$ becomes the largest finite float in $\cF$, at which point further instances of that number will not change $\AH_{\cF}(\bv)$.
\end{proof}

As we have already given the translation from $\MLPF$s to $\PL$, we next show a translation from attention modules on floats to $\PLGC$. The translation is given in terms of general vertex classification as defined in Section~\ref{appendix: Other classification heads}, since attention modules do not give Boolean vertex classifications.

\begin{lemma}\label{lemma: attention module to logic}
    For each floating-point soft or average hard-attention module, there exists an equivalent sequence of $\PLGC$-formulae.
\end{lemma}
\begin{proof}
    Intuitively, to simulate an attention module $\SA$, due to Proposition~\ref{proposition: floating-point saturation}, it is enough for each vertex to know how many of each float feature vector occurs in the graph up to some bound that only depends on the floating-point format, not on the size of the graph. We will analyze the steps of the attention module.
    During the construction, we will also use $\FO$ instead of $\PLGC$ to simulate small steps of a given attention head, but ultimately we will get rid of $\FO$ and turn its formulae into $\PLGC$.

    Let $\cG = (V, E, \lambda)$ be a graph, let $X$ be the feature matrix of $\cG$ and let $\cF$ be the floating-point format of $\SA$. Let $n$ be the number of vertices and $d$ the dimension of feature vectors in $\cG$, let $d_h$ be the hidden dimension and $k$ the number of attention heads in $\SA$ and let $W_O$ be the output matrix as it appears in the definition of attention modules.

    We start by considering a single attention head $H$ in $\SA$ with weight matrices $W_Q$, $W_K$ and $W_V$ as they appear in the definition of attention heads.
    First, consider the matrix products $XW_Q$, $XW_K$ and $XW_V$; since floating-point formats are finite and these matrix products do not involve communication between vertices, they can be expressed as $\PL$-formulae by the Boolean completeness of $\PL$ by having a single formula per column of the product. 
    For brevity, we let $Q = XW_Q$, $K = (X W_K)^\T$, $V = X W_V$ and $O = W_O$.
    Now, let $Y$ denote $(X W_Q) (X W_K)^\T$, let $Z$ denote $\frac{Y}{\sqrt{d_h}}$ and let $S$ denote $\softmax(Y)$.

    First, consider the matrix product $Y$ of $Q$ and $K$. We know that $Y_{i,j} = \sum_{\ell = 1}^{d_h} Q_{i,\ell} K_{\ell, j}$. Intuitively, the matrix $Y$ can be viewed as an edge-labeled graph with vertices $[n]$ obtained from $Q$ and $K$, where the edge $(i,j)$ is weighted with $Y_{i,j}$.
    In terms of $\FO$, each edge-weight can be expressed as a sequence $(\varphi_1(x,y), \ldots, \varphi_m(x,y))$ of quantifier-free formulae, where $m$ is the length of floating-point numbers in $\cF$. Intuitively, $\varphi_i(x,y)$ encodes the $i$th bit of the float between the vertices $x$ and $y$. Since quantifier-free $\FO$ is Boolean-complete, such a sequence of formulae clearly exists. In later steps, we will turn this sequence of $\FO$-formulae into a sequence of $\PLGC$-formulae.

    Second, consider the matrix $Z$. We obtain $Z_{i,j}$ as $\frac{Y_{i,j}}{\sqrt{d_h}}$, and this division can be simulated by a sequence of quantifier-free $\FO$-formulae (with two free variables) due to the Boolean completeness of $\FO$.

    Third, we consider the matrix $S$. We note that $S_{i,*}$ is obtained as $\softmax(Z_{i,*})$, which in turn is obtained as follows. Let $b_i = \max\{\, Z_{i,j} \mid j \in [n] \}$. We get that $Z_{i,j} = \frac{e^{Z_{i,j} - b_i}}{\sum_{\ell = 1}^{n} e^{Z_{i,\ell} - b_i}}$. (Note that due to Proposition~\ref{proposition: saturation of softmax}, only a bounded number of each float in a vector needs to be known to calculate the output of $\softmax$ for each vertex.) Due to the Boolean completeness of $\FO$, this step is possible to simulate by a sequence of quantifier-free $\FO$-formulae (with two free variables).
    This step of the analysis can be done analogously for the case where the attention module uses average hard-attention instead of soft-attention.

    Fourth, we consider the output $H(X) = SV$ of the attention head $H$. We obtain $H(X)_{i,j} = \sum_{\ell = 1}^{n} S_{i,\ell} V_{\ell,j}$. Due to Proposition~\ref{proposition: floating-point saturation}, it is enough to know a bounded number of pairs of floats $S_{i, \ell}$ and $V_{\ell,j}$ to compute $H(X)_{i,j}$. In terms of $\FO$, this means that we quantify the sequences of formulae that encode $S$. However, since the number of floats is bounded, this is also possible by using the counting global modality. Thus, since $S$ can be expressed as a sequence of quantifier-free $\FO$-formulae (with two free variables), we can express $H(X)$ by using a sequence of formulae of $\PLGC$ instead of $\FO$. 

    Fifth, we consider the final output $\SA(X) = \cH(X)W_O$ of the attention module, where $\cH(X)$ is the concatenation of $H_1(X), \dots, H_k(X)$ where $H_1, \dots, H_k$ are the attention heads of $\SA$. 
    Just like with $XW_Q$, $XW_K$ and $XW_V$, this is easy to express due to the Boolean completeness of $\PL$.
\end{proof}

Combining the translations from $\MLP$s and attention modules to logics, we obtain a translation from $\GTF$s to $\PLGC$.

\begin{theorem}\label{theorem: PLGC to GT}
    For each floating-point soft or average hard-attention graph transformer, there exists an equivalent $\PLGC$-formula.
\end{theorem}
\begin{proof}
    This follows directly from Lemmas~\ref{lemma: MLP to logic} and~\ref{lemma: attention module to logic}.
\end{proof}

Next, we show our translation from $\PLGC$ to $\GTF$s. First note that while we have given a translation from $\PL$ to $\MLP$s, there is no restriction on the number of layers of the $\MLP$. Typically the $\MLP$s appearing in graph transformers and $\GPS$-networks are assumed to be simple as defined in Section~\ref{sec: GTs and GNNs}, and the hidden dimension is restricted to at most twice the input/output dimension, though the dimension restriction is easy to work around by increasing the dimension of the surrounding architecture (for instance, if we want to build a $\GTF$-layer with an $\MLP$ of hidden dimension $d$, we simply construct a $\GTF$ of dimension $2d$). Moreover, while we can translate $\PL$ to $\MLPF$s by Lemma~\ref{lemma: PL to MLP}, we also have to give a translation from $\PL$ to $\GTF$s, as the architecture of the $\GTF$s could conceivably ruin the translation. However, we will show that there is no problem, as increasing the inner dimension of the computing models makes it possible to simulate them in the sense defined next.

We start by discussing the notion of \emph{shifting a feature update} (for discussion on feature updates, see Appendix~\ref{appendix: Other notions on equivalence}). The intuition is that we construct a feature update of higher dimension that simulates the lower-dimension feature update but moves the result from the first elements of the input vector to the last elements of the output vector (or the other way around); the remaining elements of inputs and outputs are assumed to be zeros.
For the formal definition, let $\cF$ be a floating-point format, let $\fG[\cF,d]$ be the class of $\cF^d$-labeled graphs, let $L[\cF,d]$ be the class of $d$-dimensional feature maps over $\cF$ and let $d' \geq d$. For each $\cG \in \fG[\cF,d]$, let $\cG_{r}$ (resp. $\cG_{\ell}$) denote the $\cF^{d'}$-featured graph obtained from $\cG$ by adding $d' - d$ columns of zeros to the right (resp. left) of the feature matrix of $\cG$. 
For each $\lambda \in L[\cF,d]$, we define $\lambda_r$ and $\lambda_\ell$ analogously.
Now, let $f \colon \fG[\cF,d] \to L[\cF,d]$ and $f' \colon \fG[\cF,d'] \to L[\cF,d']$ be feature updates.
If for each $\cG \in \fG[\cF,d]$ we have $f'(\cG_{r}) = f(\cG)_{\ell}$, then we say that $f'$ \textbf{shifts $f$ to the right} (by $d'-d$). Likewise if $f'(\cG_{\ell}) = f(\cG)_{r}$ for each $\cG \in \fG[\cF,d]$, then we say that $f'$ \textbf{shifts $f$ to the left} (by $d'-d$).
On the other hand, if $f'(\cG_{r}) = f(\cG)_{r}$, then we say that $f'$ is \textbf{prefix equivalent} to $f$.

Now, we show that a transformer layer can shift an $\MLP$.

\begin{lemma}\label{lemma: simulating MLPs with transformer layers}
    For each $\MLP$ $M$ of I/O dimension $d$, we can construct a soft-attention or average hard-attention transformer layer $T$ of dimension $2d$ that shifts $M$ to the right (or left). If $M$ is simple, then $T$ is simple.
\end{lemma}
\begin{proof}
    We ``skip'' the attention module by using the skip connection wrapped around it and simulate $M$ with the $\MLP$ while leveraging the increased dimension to neutralize the effect of the second skip connection.

    We construct a transformer layer $T = (\SA, \FF)$ that shifts $M$ to the right, as shifting to the left is analogous. The module $\SA$ simply outputs a zero matrix (this is possible by setting $W_O$ to be a zero matrix). Due to the skip connection wrapped around the attention module, $\FF$ now receives the same input as $T$. The $\MLP$ $\FF$ is obtained from $M$ as follows. The number of layers of $\FF$ is the same as $M$, and we add $d$ to the dimensions of each layer. The output of $M$ is computed in identical fashion by $\FF$, but the output is placed in the last $d$ components of the output vector (which is possible via simple manipulations of the weight matrices). The extra dimensions in each layer are used to remember the first $d$ components of the input of $\FF$ in each layer (i.e., we perform an identity transformation to them in each layer, which is possible because we have assumed that all inputs are non-negative and they are thus unaffected by $\ReLU$). In the final layer these remembered values are multiplied by $-1$ and placed in the first $d$ components of the output vector (there is no $\ReLU$ on the final layer, so these values of the output are non-positive). Now, if the last $d$ columns of the feature matrix of the input were zero columns, then the skip connection wrapped around the $\MLP$ does not affect the last $d$ components of the output and cancels out the first $d$ components of the output.
\end{proof}

Now we can translate $\PL$ to $\GTF$s, but not yet to simple $\GTF$s. To translate to this simpler architecture, we need a way of breaking an $\MLP$ down into multiple simple $\MLP$s, which can be carried out in a sequence. This simply means that the output of an $\MLP$ is given as input to the next $\MLP$ in the sequence.

\begin{lemma}\label{lemma: splitting MLPs}
    For each $\ReLU$-activated $\MLP$ $M$ with $k$ layers, I/O-dimension $d$ and maximum hidden dimension $d_h$, we can construct a sequence $(M_1, \dots, M_{k-1})$ of simple $\MLP$s with I/O-dimension $d' \colonequals \max\{d, d_h\}$ that are, as a sequence, prefix equivalent to $M$.
\end{lemma}
\begin{proof}
    We simply separate the layers of the $\MLP$ and transform each one into a simple $\MLP$. 
    
    The first layer of $M_i$ performs the same transformation from the prefix of its input vector to the prefix of its output vector as the $i$th layer of $M$ would (the I/O dimensions of the $i$th layer of $M$ matching the lengths of the prefixes). The remaining components of the output vector of the first layer are zeros. The second layer then simply performs an identity transformation to each component, with the exception of $M_{k-1}$, where the second layer instead performs the same operation as the final layer of $M$ (again w.r.t. prefixes).
\end{proof}

Lemmas \ref{lemma: simulating MLPs with transformer layers} and \ref{lemma: splitting MLPs} together mean that a $k$-layer $\MLP$ can be simulated by $k-1$ simple transformer layers. This is achieved through alternation by the odd layers shifting the simple $\MLP$s to the right and even layers shifting them to the left or vice versa.

Now we are almost ready to show our translation from $\PLGC$ to $\GTF$s. Before this, we require a couple of lemmas showing that $\MLP$s can check some simple binary conditions on inputs, i.e., whether an element of the input vector is greater than or equal to some specific float.

\begin{lemma}\label{lemma: checking > with an MLP}
    Let $\cF$ be a floating-point format and let $F \in \cF$. We can construct a $4$-layer 
    Boolean vertex classifier of I/O dimension $(d,1)$
    that for an input vector $(x_1, \dots, x_d)$ outputs $(1)$ if $x_i \geq F$ and $(0)$ otherwise.
\end{lemma}
\begin{proof}
    We use $\ReLU$ and negative weights to flatten all values $x \geq F$ to $1$ and all values $x < F$ to $0$. With $\R$, this would not be possible, since there is no greatest $x \in \R$ such that $x < F$, but such a float exists in $\cF$.

    For the first layer of the $\MLP$, we want to squish all values at least $F$ to a single value. Working with $\R$ we could use weight $-1$ and bias $F$ for the $i$th component (and $0$ for others) to squish all values greater than $F$ to $0$, while all values less than $F$ would become positive values. However, with floats this only works for some values of $F$, as performing this operation for the smallest number in the format (or greatest if $F$ is negative) may result in overflow\footnote{Overflow means that the maximum exponent of the floating-point format is exceeded. In these cases, the operation will output $\infty$ or $-\infty$.} for large enough values of $F$. If $F$ is a large enough value to cause overflow, then we divide both the weight and bias by $2$; since floats are in base $2$, the resulting numbers are precisely representable in the format.
    After this first layer, the value is $0$ if $x_i \geq F$ and otherwise some positive value. 
    
    Next, let $f$ be the smallest positive floating-point number in $\cF$. For the second layer, we want to squish positive values (i.e. the case where $x_i < F$) to a single value. We use the weight $-1$ and bias $f$. After this second layer, the value is $f$ if $x_i \geq F$ and $0$ otherwise. 
    
    It is now easy to define two more successive layers transforming $f$ into $1$ by using positive weights and biases $0$. For example, if $e_{\max}$ denotes the greatest possible (non-biased) exponent in $\cF$, then the weight of the first layer might be $0.10 \cdots 0 \times 2^{e_{\max}}$ transforming $f$ into $0.10\cdots0 \times 2^{-p}$. The weight of the second layer could then be $0.10 \cdots 0 \times 2^{p+2}$ transforming $0.10\cdots0 \times 2^{-p}$ to $1$.
\end{proof}

\begin{lemma}\label{lemma: checking = with an MLP}
    Let $\cF$ be a floating-point format and let $F \in \cF$. We can construct a $6$-layer 
    Boolean vertex classifier of I/H/O-dimension $(d,2,1)$
    that for an input vector $(x_1, \dots, x_d)$ outputs $(1)$ if $x_i = F$ and $(0)$ otherwise.
\end{lemma}
\begin{proof}
    We use the construction from the proof of Lemma \ref{lemma: checking > with an MLP} in the two hidden components. The first checks if $x_i \geq F$ and the second checks if $x_i \leq F$. We add one more layer to the $\MLP$ that takes the sum of these two components with bias $-1$ (and one more layer with an identity transformation because there is no $\ReLU$ on the final layer).
\end{proof}

We ease into our translation from $\PLGC$ to $\GTF$s by first considering a translation from $\PLG$ to unique hard-attention $\GTF$s as it uses a similar general strategy as the case for soft or average hard-attention, but is also both simpler and shorter. Note that, in the below theorem (and others like it), the final simple Boolean vertex classifier does not have to use any step functions, as the $\ReLU$ will suffice.

\begin{theorem}\label{theorem: logic to hard-attention GT}
    For each $\PLG$-formula, we can construct an equivalent simple unique hard-attention $\GTF$.
\end{theorem}
\begin{proof}
The idea is to use a single transformer layer to simulate a single subformula of the $\PLG$-formula. We use the $\MLP$s to simulate the operators of $\PL$ and transformer layers to simulate the non-counting global modality.

Let 
$\varphi$ be a $\PLG$-formula
and let $\psi_1, \dots, \psi_d$ be an enumeration of the subformulae of 
$\varphi$
including proposition symbols and $\top$ such that 
$\psi_1 = \varphi$.
We construct an $\ordo(d)$-layer graph transformer over any floating-point format $\cF$.

First, we give a general idea of the construction.
The initial $\MLP$ transforms the feature vectors into vectors that have two components $i$ and $2i$ for each subformula $\psi_i$. For the first of these, the $\MLP$ preserves the truth values of proposition symbols, setting the component corresponding to $\top$ to $1$. All other components are set to $0$ (including the last $d$ components). The transformer layers calculate the truth values of the subformulae starting from simple subformulae and moving to more complex ones. The layer we construct depends on the subformula $\psi_i$ under evaluation, i.e., on whether $\psi_i$ is a $\PL$-formula or of the type $\DiamondG \psi_j$ for some $j \in [d]$. Both before and after each evaluation of a subformula, the feature matrix is a binary matrix where subformulae not yet calculated have a corresponding column of $0$s and those already calculated have a corresponding column of $1$s and $0$s. (Additionally, either the first or last $d$ columns are zero columns, as the evaluated values may shift from left to right and vice versa.) Finally, the classification head just copies the first element of the feature vector.

Now, we start the construction.
The operators of $\PL$ are handled by the $\MLP$s; this is possible due to Lemmas \ref{lemma: PL to MLP}, \ref{lemma: simulating MLPs with transformer layers} and \ref{lemma: splitting MLPs} and is also what necessitates the hidden dimension $2d$ as the values shift from left to right or vice versa. All that is left is to define a transformer layer that simulates the non-counting universal modality. 
Assume that we have to simulate a formula of type $\psi_i \colonequals \DiamondG \psi_j$. Given a matrix $X \in \mathbb{B}^{n \times d}$, assume that column $k$ contains the truth values of $\top$ (i.e., column $k$ is a column of only $1$s). Now, let $W_Q$ be the Boolean-valued $(d \times 1)$-matrix where only the $k$th row is $1$. Likewise, $W_K$ is the Boolean-valued matrix where only the $j$th row is $1$. We obtain that
\[
Y = \frac{XW_Q (XW_K)^\T}{\sqrt{1}}
\]
is an $(n \times n)$-matrix, where each row contains the transpose of column $j$ of $X$. Now, $Z = \mathrm{UH}(Y)$ gives an $(n \times n)$-matrix, where a single column contains only ones and others are zero columns; if there is at least one vertex where $\psi_j$ is true, then the column of $1$s corresponds to one such vertex. Let $W_V = W_K$. Now, $Z (X W_V)$ gives an $(n \times 1)$-matrix that contains only ones if $\psi_j$ is true in at least one vertex and otherwise the vector is a zero vector. Finally, the matrix $W_O$ is the $(1 \times d)$-matrix where exactly the $i$th element is $1$ and others are $0$s. Thus, the attention module outputs a matrix where the $i$th column contains the truth values of $\psi_i$ and other columns are zero columns. The skip connection is used to recover all previously calculated columns. Finally we ``skip'' the $\MLP$ of the layer, i.e., the $\MLP$ multiplies every component by $0$ and thus outputs a zero vector; the so-far calculated columns are then recovered via the skip connection of the $\MLP$.
\end{proof}

Next, we show our translation from $\PLGC$ to soft and average hard-attention $\GTF$s.
This case is naturally more involved than the case with unique hard-attention; the key insight involves leveraging floating-point underflow.
We first recall the related Proposition~\ref{proposition: floating-point underflow 2}.

\floatunderflow*

Using the above proposition, we now give the translation from $\PLGC$ to $\GTF$s.

\begin{theorem}\label{theorem: logic to GPS}
    Given a 
    formula
    of $\PLGC$, we can construct an equivalent simple soft-attention or simple average hard-attention $\GTF$. 
\end{theorem}
\begin{proof}
    First, we describe the proof strategy informally.
    The idea is to use the $\MLP$ to handle Boolean connectives and self-attention to simulate counting global modalities $\DiamondG_{\geq k}$. Due to Proposition~\ref{proposition: floating-point underflow 2}, we can handle half the possible values of $k$ by simply checking if underflow occurs in the output of a particular attention head; if the output is $0$ (i.e., underflow occurs), it signals that the column of the input matrix has at least $k$ $1$s. 
    For the other half of the possible values of $k$, the above attention head cannot identify if there are at least $k$ $1$s in a column, but it can check a nearby upper and lower bound for the number of $1$s. For the two or three values falling between the bounds, we can use some additional numerical analysis to identify the number of $1$s.

    Now, we may start the formal proof.
    In this proof, addition, subtraction, multiplication and division are always assumed to be exact (i.e., $\frac{1}{k}$ refers to the precise value even if $k$ is a floating-point number), and the rounding operations inherent in floating-point arithmetic are always made explicit wherever they need to be performed. Thus, for each $x \in \R$ we let $\round(x)$ denote the rounded value of $x$ in $\cF$.

    \textbf{General architecture of the graph transformer}

    Let 
    $\varphi$ be a $\PLGC$-formula.
    We mostly follow the same general architecture as in the proof of Lemma \ref{theorem: logic to hard-attention GT} but with the following distinctions. We choose $\cF$ to be a floating-point format such that for the maximum grade $K$ appearing in the global modalities of 
    $\varphi$,
    all integers $k \in [K]$ can be represented precisely in $\cF$ and the number $\frac{1}{k}$ rounds to a different value for each $k \in [K]$. The hidden dimension is $2(d+4)$ (where $d$ is the number of subformulae of 
    $\varphi$
    including $\top$). The factor $2$ is because we use Lemma~\ref{lemma: simulating MLPs with transformer layers} to simulate $\MLP$s, which requires shifting. However, this means that there is always one half of the input which consists of $0$s that do not cause problems in attention heads, so we may treat the construction as having dimension $d+4$ except where Lemma~\ref{lemma: simulating MLPs with transformer layers} is applied. The last four of the $d+4$ columns are auxiliary and used to help compute some of the other columns. The initial $\MLP$ sets these auxiliary positions to $0$. 
    
    As in the proof of Lemma \ref{theorem: logic to hard-attention GT}, each transformer layer again only focuses on a single subformula $\psi_i$ of 
    $\varphi$,
    but we may require more than a single transformer layer (in a row) per subformula of the type $\DiamondG_{\geq k} \psi_j$. 
    Again the transformer layer depends on whether the subformula $\psi_i$ under evaluation is a $\PL$-formula or of the form $\DiamondG_{\geq k} \psi_j$ for some $j, k$.
    The operators of $\PL$ are again handled by the $\MLP$s, made possible by Lemmas \ref{lemma: PL to MLP}, \ref{lemma: simulating MLPs with transformer layers} and \ref{lemma: splitting MLPs}, so all that is left is to define transformer layers that simulate the counting universal modality.

    \textbf{Simulating counting universal modalities:}

    If $\psi_i \colonequals \DiamondG_{\geq k} \psi_j$ for some previously computed subformula $\psi_j$, then we make use of multiple transformer layers in a row. 
    The argument has two main cases: the simple case is the one where $k$ is an even number and $\round(\frac{1}{k}) \leq \frac{1}{k}$ (i.e., $\frac{1}{k}$ does not round upward in $\cF$), and the more complicated case is the one where $\round(\frac{1}{k}) > \frac{1}{k}$. The cases where $k$ is an odd number can be reduced to the two cases for even numbers, so we will only consider them briefly.

    \textbf{The case $\round(\frac{1}{k}) \leq \frac{1}{k}$:}
    
    First, consider the case where $k$ is an even number other than zero and $\round(\frac{1}{k}) \leq \frac{1}{k}$. We start by constructing a single transformer layer consisting of two self-attention heads. The first of these checks if there are at least $k$ $1$s in column $j$ of the feature matrix, and it is constructed as follows.
    \begin{enumerate}
        \item Let $F \in \cF$ be the greatest floating-point number in $\cF$ such that $\round(F^2) \neq \infty$. The query matrix $W_Q$ and key matrix $W_K$ are identical $((d+4) \times 1)$-matrices (i.e., vectors), where only the $j$th element is $F$ and others are $0$s. Thus, before $\softmax$, we have an $(n \times n)$-matrix where each row is either a zero vector or the $j$th column of the Boolean input matrix $X$ (multiplied by $\round(F^2)$).
        \item Next, before $\softmax$, the rows are biased according to the maximum element on the row. For the zero rows, there is no change. For the other rows, each $0$ is replaced with $-\round(F^2)$ and each $\round(F^2)$ is replaced with $0$.
        \item Now, applying $\softmax$ to a row of $0$s gives a row of $\round(\frac{1}{n'})$ where $n'$ is the sum of a multiset of $n$ $1$s, where $n$ is the number of rows in $X$. For the other rows, let $\ell$ be the number of $1$s in column $j$ of $X$, let $M$ be a multiset of $\ell$ $1$s and let $\ell'$ denote $\mathrm{SUM}_{\cF}(M)$ where $\mathrm{SUM}_{\cF}$ is the saturating sum from Proposition \ref{proposition: floating-point saturation}. The application of $\softmax$ will then give a row where each $-\round(F^2)$ is replaced with $0$ and each $0$ is replaced with $\round(\frac{1}{\ell'})$. This is because $e^{-\round(F^2)}$ rounds to $0$ and thus $\softmax$ gives an even probability distribution for the remaining positions, calculated by first taking the saturating sum of values $e^0 = 1$.
        \begin{itemize}
            \item The analysis in step 3 above is identical for average hard-attention, when the denominator of $\AH$ is obtained as a multiset of $1$s. On the other hand, consider the implementation of $\AH$ discussed in Appendix~\ref{appendix: floats}, where the size of the set in the denominator is rounded into the floating-point format directly. In this case, the analysis is also identical except that $n' = \min\{\round(n), F_{\max}\}$ and $\ell' = \min\{\round(\ell), F_{\max}\}$, where $n$ and $\ell$ are the same as above and $F_{\max}$ is the largest non-infinite floating-point number in the format. This is because, according to the implementation, $n$ and $\ell$ are rounded into the format, but not to $\infty$.
        \end{itemize}
        \item Let $f$ be the smallest positive floating-point number in $\cF$. The value matrix $W_V$ is defined in the same way as $W_Q$ and $W_K$, except that we use $\frac{k}{2}f$ instead of $F$. Thus, the matrix product $V = X W_V$ gives the $j$th column of $X$ where each $1$ is replaced with $\frac{k}{2}f$.
        \item The product of the $(n \times n)$-matrix $N$ resulting from step $3$ and the $(n \times 1)$-matrix $V$ from step $4$ thus multiplies each column where $\round(\frac{1}{\ell'})$ appears with $\frac{k}{2}f$ and then takes the sum of each row. 
        Now the matrix product gives us $\round(\frac{1}{\ell'})\frac{k}{2}f$
        which by Proposition~\ref{proposition: floating-point underflow 2} rounds to $0$ exactly when $\round(\frac{1}{\ell'}) \leq \frac{1}{k}$. Recall that $\round(\frac{1}{k}) \leq \frac{1}{k}$ and since $\frac{1}{h}$ rounds to a different value for all $h \in [K]$ we must have $\round(\frac{1}{k-1}) > \frac{1}{k}$.
        This means that $\round(\frac{1}{\ell'}) \leq \frac{1}{k}$ is equivalent to $\ell' \geq k$.
        Now $NV$ is a zero vector if and only if the formula $\psi_j$ is true in at least $k$ vertices of the graph (or if the formula is true in zero of the $n$ vertices and $n \geq k$).
    \end{enumerate}
    
    From the above constructed attention head, we see that we still have to check if zero vertices satisfy $\psi_j$; we do this with a second attention head. This second head will output a zero vector if and only if there are zero formulae that satisfy $\psi_j$; otherwise, the elements of the vector will all be some other singular value. For this attention head, we define the query and key matrices as zero matrices, which results in the output of $\softmax$ being a matrix where each element is $\round(\frac{1}{n'})$ for some $n' \leq n$ (this is because $\softmax$ calculates $n'$ as a sum of a multiset of $n$ $1$s, and the sum of a multiset of floats saturates by Proposition \ref{proposition: floating-point saturation}). The value matrix $W_V'$ is the $((d+4) \times 1)$-matrix where the $j$th row is $1$ and others are zero; the matrix $XW_V'$ is thus the $j$th column of $X$. The output of the attention head is thus a vector where each element is $\mathrm{SUM}_{\cF}(M)$ where $M$ is a multiset of $\ell$ copies of $\round(\frac{1}{n'})$, where $\ell$ is the number of vertices satisfying $\psi_j$. Since $\round(\frac{1}{n'})$ is never zero, this means that the output is a vector of $0$s if $\ell = 0$ and otherwise each component is some positive value. 
    By the same analysis, this construction works for average hard-attention. For the alternate implementation of $\AH$ discussed in Appendix~\ref{appendix: floats}, the same analysis holds with the difference that $n' = \min\{\round(n), F_{\max}\}$ where $F_{\max}$ is the greatest float in the format other than $\infty$.

    Now, to check if $\ell \geq k$ (where $\ell$ is the number of $1$s in column $j$ of $X$), we simply have to check if the output of the first attention head is a zero vector and the output of the second attention head is a non-zero vector. Unfortunately, if $0 < \ell < k \leq n$, then the output of the first attention head may contain both $0$s and non-zero values, meaning that some rows may currently be under a false impression that $\ell \geq k$. We will use a second transformer layer to distribute the necessary information between all vertices. First, we finish the current layer by defining the weight matrix $W_O$ to store the output vectors of the two attention heads in two separate columns $d+1$ and $d+2$. Next, the $\MLP$ of this layer normalizes these two columns such that all positive values are replaced with $1$s (this is possible by Lemma \ref{lemma: checking > with an MLP}; if we desire a simple graph transformer, then we can simply use multiple consecutive layers to accomplish this task by Lemmas \ref{lemma: simulating MLPs with transformer layers} and \ref{lemma: splitting MLPs}.)

    The second layer uses the second attention head of the first layer, except that it now checks column $d+1$ instead of column $j$, i.e., the $(d+1)$th row of the value matrix is $1$ and others are $0$. If the column $d+1$ has at least one $1$, then all elements of the output vector of the attention head become non-zero. The weight matrix $W_O'$ places these values in the column $d+3$. Now, the $\MLP$ once again normalizes positive values in column $d+3$ to $1$. Then, we use a third layer that does not do anything in 
    the self-attention module, but the $\MLP$ checks that the values in the columns $d+2$ and $d+3$ are both $1$. A simple sum of the values is enough; if $\ell \geq k$, then the sum is $1$. Otherwise the sum is $0$ or $2$, because it is not possible that $\ell = 0$ and $0 < \ell < k$. This result is placed in the $i$th component. The $\MLP$ also resets columns $d+1$ through $d+4$ by multiplying them in the final layer of the $\MLP$ with $-1$ (whence they will be eliminated by the skip connection) such that these columns can be used again in later layers.

    On the other hand, assume that $k$ is an odd number and $\round(\frac{1}{k-1}) > \frac{1}{k-1}$. Then $\round(\frac{1}{k-1})\frac{k-1}{2}f \geq \frac{1}{2}f$ but $\round(\frac{1}{k})\frac{k-1}{2}f < \frac{1}{2}f$ because we assumed that $\frac{1}{k-1}$ and $\frac{1}{k}$ round to different values in $\cF$ and $\frac{1}{k} < \frac{1}{k-1}$. Thus, the above construction can now be used to check if at least $k$ vertices satisfy $\psi_j$.

    \textbf{The case $\round(\frac{1}{k}) > \frac{1}{k}$:}

    Now assume that $k$ is an even number such that $\round(\frac{1}{k}) > \frac{1}{k}$ (the case where $k$ is an odd number such that $\round(\frac{1}{k-1}) \leq \frac{1}{k-1}$ is similar). We can use the above construction to calculate $\DiamondG_{\geq k+1} \psi_j$. We can also calculate either $\DiamondG_{\geq k-1} \psi_j$ or $\DiamondG_{\geq k-2} \psi_j$ depending on whether $\round(\frac{1}{k-2}) > \frac{1}{k-2}$. Thus, in order to calculate $\DiamondG_{\geq k} \psi_j$, we can simply calculate the upper bound ($k+1$) and the lower bound ($k-1$ or $k-2$), and then separate the remaining cases, i.e., we need a way of distinguishing the cases where the formula is satisfied in $k$, $k-1$ and $k-2$ vertices. To do this, we simply use the above construction for $k$. 
    Recall that $\round(\frac{1}{k}) > \frac{1}{k}$ and $\frac{1}{h}$ rounds to a different value for all $h \in [K]$. 
    If $k-1$ is the lower bound, then we have $\frac{1}{k-1} < \round(\frac{1}{k-1}) < \frac{1}{k-2}$. Thus for all $k > 3$,
    \[
        \begin{aligned}
            \frac{1}{2} = \frac{k}{2} \cdot \frac{1}{k} &< \frac{k}{2} \cdot \round\bigg{(}\frac{1}{k}\bigg{)} \\
            &< \frac{k}{2} \cdot \round\bigg{(}\frac{1}{k-1}\bigg{)} < \frac{k}{2} \cdot \frac{1}{k-2} < \frac{3}{2}.
        \end{aligned}
    \]
    Likewise, if $k-2$ is the lower bound, then we have $\round(\frac{1}{k-2}) \leq \frac{1}{k-2}$ and thus for all $k > 3$,
    \[
        \begin{aligned}
            \frac{1}{2} = \frac{k}{2} \cdot \frac{1}{k} &< \frac{k}{2} \cdot \round\bigg{(}\frac{1}{k}\bigg{)} \\
            &< \frac{k}{2} \cdot \round\bigg{(}\frac{1}{k-1}\bigg{)} \\
            &< \frac{k}{2} \cdot \round\bigg{(}\frac{1}{k-2}\bigg{)} \leq \frac{k}{2} \cdot \frac{1}{k-2} < \frac{3}{2}.
        \end{aligned}
    \]
    Thus, multiplying $\round(\frac{1}{k-1})$ or $\round(\frac{1}{k})$ with $\frac{k}{2}f$ gives $f$ in each case (because we round to the nearest number) and likewise for $\round(\frac{1}{k-2})$ when $k-2$ is the lower bound. Accordingly, the output vector is then $(k-2)f$, $(k-1)f$ or $kf$. 
    The only even values of $k$ this analysis does not account for are $k = 0$ and $k = 2$. Because $\round(\frac{1}{2}) = \frac{1}{2}$, the case $k = 2$ does not need to be examined here, and we already showed in the previous case how to construct an attention head that checks for $k = 0$.

    For this last case, we thus end up with four separate columns $d+1$ through $d+4$ that need to be combined; column $d+1$ determines if zero vertices satisfy $\psi_j$,column $d+2$ checks that at least $k-2$ or $k-1$ vertices satisfy $\psi_j$, column $d+3$ checks that at least $k+1$ vertices satisfy $\psi_j$ and column $d+4$ gives a distinct value for the cases $k-2$, $k-1$ and $k$. To output $1$, we thus want the column $d+1$ to be positive, the column $d+2$ to be $0$, and either the column $d+3$ to be $0$ or the column $d+4$ to be the distinguished value corresponding to $k$. Checking each of these conditions individually can be done by $\MLP$s that only use the $\ReLU$ activation function by Lemmas \ref{lemma: checking > with an MLP} and \ref{lemma: checking = with an MLP} (i.e., for each condition the output will be $1$ if it is cleared and $0$ otherwise). Now, one more $\MLP$ layer can check that the conditions are met.
    by taking the sum of the components and adding the bias $-2$.
    We also add one more transformer layer that, as in the easier case, distributes the information (whether $\ell \geq k$) to all vertices, since only vertices where $\psi_j$ is true know this. 
    Finally, the $\MLP$ again resets the columns $d+1$ through $d+4$. This concludes the description for handling the counting global modality as all cases have been covered.
\end{proof}

As the final piece of Theorem~\ref{theorem: PLGC = SGT = AHGT}, we characterize $\GNNGCF$s with constant local aggregation functions. We utilize Theorem~\ref{theorem: GNN GML} and though Theorem~\ref{theorem: GNN GML} is only proven in the next appendix, the proof of Theorem~\ref{theorem: GNN GML} is independent of the result below.

\begin{theorem}
    $\GNNGCF$s with constant local aggregation functions have the same expressive power as $\PLGC$. This also applies when the $\GNNGCF$s are simple.
\end{theorem}
\begin{proof}
    Theorem~\ref{theorem: GNN GML} states that $\GNNF$s with counting global readout have the same expressive power as $\GMLGC$.
    If the $\GNNF$s have constant aggregation functions, then the translation in \citep{ahvonen_neurips} (which we apply for Theorem~\ref{theorem: GNN GML}) results in a $\PLGC$-formula. Likewise, the translation from $\GMLGC$ to $\GNNF$s results in a $\GNNF$ with a constant local aggregation function if the formula is from $\PLGC$.
\end{proof}

\subsection{Proof of Theorem \ref{theorem: GNN GML}}\label{appendix: GML = GNNs}

In this section, we provide the proof of Theorem~\ref{theorem: GNN GML}. 

\thmGNNGML*

First, we note that Theorem 3.2 in \citep{ahvonen_neurips} showed that recurrent $\GNNF$s and the logic $\GMSC$ have the same expressive power. A simple modification of the proofs of \citep{ahvonen_neurips} gives us the following result.

\begin{lemma}\label{lemma: GNNF = GML}
    $\GNNF$s, simple $\GNNF$s and $\GML$ have the same expressive power.
\end{lemma}
\begin{proof}
    Note that some of the concepts used in this proof are taken from \citep{ahvonen_neurips}.
    First, we translate a $\GNNF$ into $\GML$ as follows. 
    First, we use Proposition B.17 in \citep{ahvonen_neurips} which states that a $\GNN$ with $n$ layers can be modified into a recurrent $\GNN$ that gives the same output after exactly $n$ rounds of iteration.
    The construction in the proposition is for reals, but it is easy to see that it also works for floating-point $\GNN$s (i.e., we can turn an $n$-layer floating-point $\GNN$ into an equivalent recurrent floating-point $\GNN$ that uses the same layer repeatedly).
    We must note that the $\GNNF$s in \citep{ahvonen_neurips} do not include skip connections, but they do allow arbitrary combination functions, which means that we can simply treat the skip connections as part of the combination function when using Proposition B.17.
    By Lemma B.3 and Theorem 3.2 of \citep{ahvonen_neurips}, a recurrent $\GNNF$ can be translated into a GMSC-program, which are a recursive generalization of $\GML$. We can run this program for exactly $n$ rounds, which gives us an equivalent GML-formula, obtained as the disjunction of the $n$th iteration formulae of the appointed predicates of the program.

    For the translation from $\GML$ to $\GNNF$s, we can simply use the proof of Proposition 4.1 in \citep{Barcelo_GNNs}, as floating-point numbers and the saturating floating-point sum are sufficient for the construction given there. This direction also follows from Lemma B.5 in \citep{ahvonen_neurips}. Note that the $\MLP$s of the $\GNNF$ use the truncated $\ReLU$ as the activation function, i.e., the function $\ReLU^*(x) = \min\{1,\max\{0,x\}\}$. We could also use the ordinary $\ReLU$ by adding a couple of additional layers to the $\MLP$s, and we could obtain a $\GNNF$ using simple $\MLP$s by splitting the computation of the non-simple $\ReLU$-activated $\MLP$ into multiple message-passing layers.

    One important thing to note is that the accepting condition for the $\GNNF$s in \citep{ahvonen_neurips} specifies a set of so-called accepting feature vectors. Trivially this accepting condition can simulate any $\MLP$ that gives a Boolean vertex classification by simply listing all the feature vectors that are accepted by the $\MLP$. For the other direction, note that the constructions in Proposition 4.1 in \citep{Barcelo_GNNs} and Lemma B.5 in \citep{ahvonen_neurips} are made such that it is possible to determine whether a feature vector is accepting or not simply by checking if a single element of the feature vector is $0$ or $1$, which means that the Boolean vertex classifier of our $\GNNF$s only needs to project this single element.
\end{proof}

Next we show the generalization of Lemma \ref{lemma: GNNF = GML} for counting global readouts and counting global modalities.

\begin{lemma}\label{lemma: GNNF+GC = GML+GC}
    $\GNNGCF$s, simple $\GNNGCF$s and $\GMLGC$ have the same expressive power.
\end{lemma}
\begin{proof}
    This follows via a simple modification to the proof of Theorem 3.2 in \citep{ahvonen_neurips}. We simply have to include a counting global readout in the recurrent float $\GNN$s and a counting global modality in the logic $\GMSC$, and the two can be used to simulate one another in the proofs. In the case of the R-simple recurrent $\GNN$s in \citep{ahvonen_neurips}, this simply involves an additional sum aggregation of all feature vectors in the graph and an additional weight matrix in the combination functions.
\end{proof}

Finally, we prove the last part of Theorem~\ref{theorem: GNN GML}. In the lemma below, note that simple $\GNNGF$s use a version of the floating-point sum that only sums one of each feature vector in the multiset.

\begin{lemma}\label{lemma: GNNF+G = GML+G}
    $\GNNGF$s, simple $\GNNGF$s and $\GMLG$ have the same expressive power.
\end{lemma}
\begin{proof}
    When modifying the constructions in \citep{ahvonen_neurips}, we notice that if the aggregation function $\AGG$ in the global readout is set-based, i.e., $\AGG(M) = \AGG(M_{|1})$ for all multisets $M$, then only the non-counting global modality is required to simulate it and vice versa, only a set-based aggregation is required to simulate a non-counting global modality.
\end{proof}

\subsection{Proof of Theorem \ref{theorem: GMLGC = SGPS = AHGPS = GNNG}}\label{Appendix: GMLGC = SGPS = AHGPS = GNNG}

In this section, we give the proof of Theorem \ref{theorem: GMLGC = SGPS = AHGPS = GNNG}.

\thmGMLGCSGPSAHGPSGNNG*

For the direction from $\GPSF$-networks to logic, we only need to append the analysis of $\GTF$s to account for the message-passing layers. This is rather trivial since $\GNNF$s translate into $\GML$ and the skip connection and architecture do not get in the way.

\begin{theorem}
    For each floating-point soft or average hard-attention $\GPS$-network, there exists an equivalent sequence of $\GMLGC$-formulae.
\end{theorem}
\begin{proof}
    This follows directly from Theorem \ref{theorem: PLGC to GT} and Lemma~\ref{lemma: GNNF = GML}.
\end{proof}

Because we have already translated $\PLGC$ to $\GTF$s, $\PL$ to $\MLPF$s and $\GML$ to $\GNNF$s, our proof will consist of translating transformer layers, $\MLPF$s and message-passing layers into $\GPSF$-networks. We start with the translation from $\MLPF$s to $\GPSF$-networks. Note that the lemmas of this section apply for both reals and floats.

\begin{lemma}\label{lemma: simulating MLPs with GPS-layers}
    For each $\MLP$ $M$ of I/O dimension $d$, we can construct a $\GPS$-layer $G$ of dimension $3d$ that shifts $M$ to the right (or left). If $M$ is simple, then $G$ is simple.
\end{lemma}
\begin{proof}
    The proof is very similar to the proof of Lemma \ref{lemma: simulating MLPs with transformer layers}. We construct a $\GPS$-layer $G = (\SA, \MP, \FF)$ that shifts $M$ to the right as shifting to the left is analogous. We construct $\MP$ to copy the first $d$ elements of a vertex's feature vector and place them in the middle $d$ components of the output vector, other components being $0$. The module $\SA$ again outputs a zero matrix. Now the middle $d$ components of the input of $\FF$ are the same as the first $d$ components of the input of $G$. The $\MLP$ $\FF$ is constructed the same as before, placing the output of $M$ (w.r.t. the middle $d$ components of the input) into the last $d$ vector components of the output vector, while the other components of the output cancel out the input (which is achieved by remembering these inputs in each layer).
\end{proof}

Lemmas \ref{lemma: splitting MLPs} and \ref{lemma: simulating MLPs with GPS-layers} together mean that a $k$-layer $\MLP$ can be simulated by $k-1$ simple $\GPS$-layers.

Next, we show how $\GPS$-layers can simulate a message-passing layer.

\begin{lemma}
    For each message-passing layer $\mathrm{M}$ of dimension $d$, we can construct a $\GPS$-layer $G$ of dimension $3d$ that shifts $\mathrm{M}$ to the right (or left). If $\mathrm{M}$ uses sum aggregation, then $G$ is simple.
\end{lemma}
\begin{proof}
    At this point the proof is almost routine by similarity to the proofs of Lemmas \ref{lemma: simulating MLPs with transformer layers} and \ref{lemma: simulating MLPs with GPS-layers}.
    We construct a $\GPS$-layer $G = (\SA, \MP, \FF)$ that shifts $\mathrm{M}$ to the right, as shifting to the left is analogous.  We define $\MP$ to place the output of $\mathrm{M}$ into the middle $d$ components, the others being $0$; this is possible by trivial matrix manipulations in the simple case. Meanwhile, $\SA$ simply outputs a zero matrix as in previous proofs. Finally, $\FF$ places the middle $d$ components of the input into the last $d$ components of the output while the other components of the output cancel out the input.
\end{proof}

As the last result of this section, we show how a $\GPS$-layer can simulate a transformer layer.

\begin{lemma}\label{lemma: transformer layers to GPS-layers}
    For each transformer layer $T$ of dimension $d$, we can construct a $\GPS$-layer $G'$ of dimension $3d$ that shifts $T$ to the right (or left) (with the same attention mechanism). If $T$ is simple, so is $G'$.
\end{lemma}
\begin{proof}
    Let $T = (\SA, \FF)$ be a transformer layer of dimension $d$. We construct an equivalent $\GPS$-layer $G' = (\SA', \MP', \FF')$. The proof is similar to previous proofs, but special attention is paid to the attention module.
    
    We define the attention module $\SA'$ by modifying $\SA$ as follows. We use the usual notation $W_Q$, $W_K$, $W_V$ and $W_O$ for the involved matrices of $\SA$. For each $W \in \{W_Q, W_K, W_V\}$, we construct $W'$ by adding rows of zeros to the bottom of $W$. Likewise, we obtain $W_O'$ by adding the same number of columns of zeros to the left and right of $W_O$ (meaning that $W_O$ is in the middle $d$ columns of $W_O'$). Now, if $\SA(\cG) = \cG'$, then $\SA'(\cG_r)$ is the graph obtained from $\cG'$ by adding $d$ zero columns to the left and right of the feature matrix (recall the meaning of $\cG_r$ from Appendix~\ref{Appendix: PLGC = SGT = AHGT}).
    
    The message-passing layer places the first $d$ components of a feature vector to the middle $d$ elements in the output vector. The first and last $d$ elements of the output vector are simply $0$s.

    Now (assuming that the last $2d$ columns of input feature matrices are zero columns) the skip connections of both $\SA'$ and $\MP'$ only affect the first $d$ columns of their respective output matrices. When the outputs of $\SA'$ and $\MP'$ are added together in the $\GPS$-layer, the middle $d$ columns are the same as they would be after the skip connection is applied to $\SA$ in the transformer layer $T$.

    Now the middle $d$ elements of the input vectors of $\FF'$ are the same as the input vectors of $\FF$. As in earlier proofs, we define $\FF'$ to do the same transformations as $\FF$ to the middle $d$ elements, but placing the output to the last $d$ elements of the output vector, while the other elements of the output vector are defined to cancel out the input after the skip connection is applied.
\end{proof}

Now Theorem~\ref{theorem: GMLGC = SGPS = AHGPS = GNNG} follows directly from the results of this section. As a final note, we give a one-sided translation from $\GMLG$ to unique hard-attention $\GPSF$-networks.

\begin{theorem}\label{theorem: logic to hard-attention GPS}
    For each $\GMLG$-formula, we can construct an equivalent simple unique hard-attention $\GPSF$-network.
\end{theorem}
\begin{proof}
    This follows from Theorem~\ref{theorem: logic to hard-attention GT}, Lemma~\ref{lemma: GNNF = GML} and Lemma~\ref{lemma: transformer layers to GPS-layers}.
\end{proof}

\subsection{Characterizations with graph classification and non-Boolean classification}\label{appendix: graph classification}

In this section, we explain in more detail how Theorems~\ref{theorem: PLGC = SGT = AHGT}, \ref{theorem: GNN GML} and \ref{theorem: GMLGC = SGPS = AHGPS = GNNG} also generalize for graph classification and non-Boolean classification. 

Note that our results generalize, in a sense, trivially for graph classification if we simply restrict our analysis to those learning models and logic formulae that define graph properties. What makes the generalization to graph classification interesting is that we have fixed the syntax of the logics and we have also fixed learning models that exclusively define graph properties via a graph classification head.

For graph classification, we replace the logic $\cL$ appearing in the result with $\cL^*$ and Boolean vertex classifiers with Boolean graph classifiers. Recall that the syntax of the logic $\cL^*$ is defined as
\[
    \psi \coloncolonequals \DiamondG_{\geq k} \varphi \,|\, \neg \, \psi \,|\, \psi \land \psi \,|\, \DiamondG_{\geq k} \psi,
\]
where $\varphi$ is a $\Pi$-formula of $\cL$.
When translating from a logic $\cL$ to, e.g., a $\GPSF$-network, note first that each formula of $\cL^*$ is also a formula of $\cL$,\footnote{Note that this is not the case logics that do not have the counting global modality; for these logics, we can simply remove $\DiamondG_{\geq k} \psi$ from the syntax and handle the remaining counting global modalities with the readout gadget.} and we can use the same construction as we do for vertex classification. 
Recall that a computing model that is a graph classifier uses a readout gadget instead of an $\MLP$ as a final classification head as defined in Appendix \ref{appendix: Other classification heads}.  
Thus, the only difference is that the final classifier is a readout gadget instead of an $\MLP$, and so we use a readout gadget that uses the max aggregation function (this also works for the sum aggregation function by Lemma \ref{lemma: checking > with an MLP}) and where the $\MLP$ is the same as the Boolean vertex classifier in the construction. Since the final $\MLP$ in the construction is simply a projection of a single feature of the feature vector, and since this component of the vector has the same value regardless of vertex, the max aggregation suffices. For the other direction, note that prior to the classification head, we have an equivalent formula of $\cL$ for each bit of a feature vector of a vertex. After this, we can simulate the aggregation function of the readout gadget with the global diamond $\DiamondG_{\geq k}$ and the $\MLP$ via the connectives $\neg$ and $\land$, which gives us a formula of $\cL^*$.

For non-Boolean classification, we simply use a sequence of formulae instead of a single formula and a general classification head instead of a Boolean classification head. From logics to learning models, we simply have a single component of the output feature vector for each formula in the sequence. From learning models to logics, we have a formula in the sequence per each bit of each float in the output feature vector of the learning model.

\section{Transformers on Words and Positional Encodings}\label{appendix: PEs and words}

In this section, we go through the follow-up results explained in the conclusion section \ref{sect:concl} for floating-point based positional encodings and transformers over words.

\subsection{Transformers on Words}\label{appendix: Transformers on Words}

We start by considering transformers over word-shaped graphs (see Appendix \ref{appendix: words} for the definition of word-shaped graphs).
As already noted, a $\GT$ over word-shaped graphs is just an `encoder-only transformer without causal masking'. For example, the popular BERT \citep{devlin2019bertpretrainingdeepbidirectional} is such a model inspired by the work of \cite{NIPS2017_3f5ee243}. 

It is straightforward to see that our proofs of Theorem \ref{theorem: PLGC = SGT = AHGT} and Theorem \ref{theorem: GMLGC = SGPS = AHGPS = GNNG} generalize when restricted to word-shaped graphs. Our results hold in all four cases: Boolean and general vertex classification and Boolean and general graph classification. 
Boolean vertex classification is perhaps the most commonly used for word-shaped graphs, or more precisely, a word $\bw$ is said to be accepted if its pointed word-shaped graph $(\cG_\bw, 1)$ is classified to $1$.
Note that GTs and GPS-networks in the corollaries below do not use causal masking.

\begin{corollary}
    When restricted to word-shaped graphs, the following have the same expressive power: 
    %\begin{itemize}
        %\item 
        $\PLGC$, 
        %\item 
        soft-attention $\GTF$s and
        %\item 
        average hard-attention $\GTF$s (and
        %\item 
        $\GNNGCF$s with constant local aggregation functions).
        This also holds when the $\GTF$s and $\GNNGCF$s are simple.
\end{corollary}

\begin{corollary}
    When restricted to word-shaped graphs, the following have the same expressive power: 
    %\begin{itemize}
        %\item 
        $\GMLGC$, 
        %\item 
        soft-attention $\GPSF$-networks and
        %\item 
        average hard-attention $\GPSF$-networks (and
        %\item 
        $\GNNGCF$s).
        This also holds when the $\GPSF$-networks and $\GNNGCF$s are simple.
\end{corollary}

In order to characterize floating-point unique hard-attention graph transformers on words, we need to introduce a new logic. We start by discussing a known fragment of linear-time temporal logic (LTL). The set of $\Pi$-formulae of \textbf{linear-time temporal logic with past} (LTLP) is defined by the grammar
\[
\varphi \coloncolonequals \top \mid p \mid \neg \varphi \mid \varphi \land \varphi \mid \mathsf{P} \varphi,
\]
where $p \in \Pi$. The (full) linear-time temporal logic also includes the ``since'' operator $\mathsf{S}$, the ``until'' operator $\mathsf{U}$ and the ``future'' operator $\mathsf{F}$, but we do not consider them here. The semantics of LTLP is defined over words as follows. Given a $\Pi$-labeled pointed word-shaped graph $(\cG_\bw, i)$ and a $\Pi$-formula $\varphi$ of LTLP, we define the truth of $\varphi$ in $(\cG_\bw, i)$ as follows.
\begin{itemize}
    \item $\cG_\bw, i \models p$ iff $p \in \lambda(\cG_\bw)_i$.
    \item The semantics of Boolean operators are defined in the usual way.
    \item $\cG_\bw, i \models \mathsf{P} \varphi$ iff there is some $j < i$ such that $\cG_\bw, j \models \varphi$.
\end{itemize}
Now, since we have defined LTLP, we can define the logic $\cL$ used for the characterization. The logic $\cL$ is a mix of $\PLG$ and LTLP, and obtained from the logic PL as follows: we extend PL with the $2$-ary operator $\mathsf{GB}$ called \textbf{global-before}. The semantics of the formula $\varphi \,\mathsf{GB}\, \psi$ is the same as that of the formula 
\[
\DiamondG ( \psi \land \neg \mathsf{P} \varphi).
\]
Intuitively, the formula states that ``there is a state where $\psi$ is true and $\varphi$ does not precede $\psi$''. We call this logic \textbf{PL with global-before} (or PL+GB).

We start with the following observation. Below, we let $\mathrm{LTLP+G}$ denote the logic obtained from $\mathrm{LTLP}$ by extending its syntax with the global non-counting modality. Below, ``$\cC_1$ is strictly more expressive than $\cC_2$'' means that for every $x \in \cC_2$, there is an equivalent object $y \in \cC_1$, but there is an object $z \in \cC_1$ that is not equivalent to any object in $\cC_2$. Moreover, ``$\cC_1$ and $\cC_2$ have orthogonal expressive power'' means that there is an object in $\cC_1$ that is not equivalent to any object in $\cC_2$, and vice versa.
\begin{proposition}
    Over word-shaped graphs, the following holds.
    \begin{itemize}
        \item $\mathrm{PL+GB}$ is strictly more expressive than $\PLG$.
        \item $\mathrm{LTLP}$ and $\mathrm{PL+GB}$ have orthogonal expressive power.
        \item $\mathrm{LTLP}+\mathrm{G}$ is strictly more expressive than $\mathrm{PL+GB}$.
    \end{itemize}
\end{proposition}
\begin{proof}
    First, we show that $\mathrm{PL+GB}$ is strictly more expressive than $\PLG$.
    We notice that any formula of the type $\DiamondG \varphi$ can be expressed as 
    $\bot \,\mathsf{GB}\, \varphi$, which means that each formula of $\PLG$ can be expressed in $\mathrm{PL+GB}$. 
    Clearly, $\mathrm{PL+GB}$ is strictly more expressive than $\PLG$ because, e.g., the formula 
    $\top \mathsf{GB} p$, 
    which intuitively states that ``the first letter in a word is $p$'', is not expressible in $\PLG$. 
    
    Secondly, it is easy to see that $\mathrm{LTLP}$ and $\mathrm{PL+GB}$ have orthogonal expressive power by considering the following invariance. $\mathrm{LTLP}$ is invariant over word-shaped graphs under extending a word from the right with a unary string, i.e., given a pointed $\Pi$-word $(\bw, i)$, for every $\Pi$-formula $\varphi$ in $\mathrm{LTLP}$, we have $\cG_\bw, i \models \varphi$ iff $\cG_{\bw\bu}, i \models \varphi$, where $\bu$ is a unary word (i.e. $\bu \in \{p\}$ for some letter $p$) and $\bw\bu$ denotes the concatenation of $\bw$ with $\bu$. However, clearly, $\mathrm{PL+GB}$ is not invariant under extending a word from the right.
    On the other hand, $\mathrm{PL+GB}$ can express formulae of the type $\DiamondG \varphi$ which $\mathrm{LTLP}$ cannot.
    
    Lastly, it is trivial that $\mathrm{PL+GB}$ is contained in $\mathrm{LTLP}+\mathrm{G}$. To see that $\mathrm{LTLP}+\mathrm{G}$ is strictly more expressive, note that $\mathrm{LTLP+G}$ contains $\mathrm{LTLP}$ that can already express properties not expressible in $\mathrm{PL+GB}$.
\end{proof}

Before we characterize unique hard-attention, we prove the following helpful lemma.
\begin{lemma}\label{lemma: attention identical rows and columns}
    Let $\cF$ be a float format, let $X \in \cF^{n \times d}$ and $W_Q, W_K \in \cF^{d \times d_h}$. If the $i$th and the $j$th row of $X$ are identical, then the $i$th row (resp. column) and the $j$th row (resp. column) of the matrix
    \[
    \frac{(XW_Q)(XW_K)^\T}{\sqrt{d_h}},
    \]
    are identical.
\end{lemma}
\begin{proof}
    Since the $i$th row and the $j$th row of $X$ are identical, the $i$th row and the $j$th row of $XW_Q$ are identical and the same is true for $XW_K$. To see that the $i$th row and the $j$th row of $Y \colonequals (XW_Q)(XW_K)^\T$ are identical, note that for each column $k$, we have $Y_{i,k} = (XW_Q)_{i,*} ((XW_K)_{k,*})^\T$ and $Y_{j,k} = (XW_Q)_{j,*} ((XW_K)_{k,*})^\T$. By an analogous argument, we can also see that the $i$th column and the $j$th column of $Y$ are identical. Dividing all elements in $Y$ with $\sqrt{d_h}$ naturally does not change which rows or columns are identical.
\end{proof}

We will show that, over words, PL+GB has the same expressive power as float-based graph transformers with unique hard-attention (without causal masking). 
\begin{theorem}\label{thrm: unique GT}
    Over word-shaped graphs, $\mathrm{PL+GB}$ has the same expressive power as unique hard-attention $\GTF$s. This also applies when $\GTF$s are simple.
\end{theorem}
\begin{proof}
    Let $T$ be a $\GTF$ with unique hard-attention over $(\Pi, d)$ over the floating-point format $\cF$.  We construct an equivalent $\Pi$-formula of $\mathrm{PL+GB}$ as follows. The construction is almost the same as the case with soft-attention and average hard-attention. The hardest part is to simulate the unique hard-attention function in each attention head. We first intuitively describe the simulation of an attention head
    \[
    H(X) = \mathrm{UH}\Big( \frac{(X W_Q)(XW_K)^\T}{\sqrt{d_h}} \Big)(XW_V).
    \]
    First, we observe that the unique hard function $\mathrm{UH}$ in $H$ outputs a square matrix $S$ that is just a selection matrix, i.e., each row contains one $1$ and all other elements of the row are $0$s. Secondly, by Lemma \ref{lemma: attention identical rows and columns} and the former observation we notice that if the $i$th row and the $j$th row of the input matrix $X$ are identical, then in $H$ the selection matrix $S$ will select the same row from $XW_V$ to be the $i$th row and the $j$th row of $H(X)$. 
    To find out which row of $XW_V$ will be selected by the selection matrix $S$ of $H(X)$, we have to look at the values of $XW_Q$ and $(XW_K)^\T$ as follows.
    For each row $\vec{q}$ of $XW_Q$, we find the leftmost column $\vec{k}$ of $(XW_K)^\T$ that maximizes the value $\vec{q}\vec{k}$ across all the columns of $(XW_K)^\T$. If the leftmost column of $(XW_K)^\T$ that gives the maximum value with the $i$th row of $XW_Q$ is the $j$th column, then the selection matrix $S$ selects the $j$th row of $XW_V$ for the $i$th row of $H(X)$. In terms of words, if $X$ represents the feature matrix of a word, this means that for each node $v$ (in the word) we have to find the leftmost node $w$ that maximizes the value of 
    \[
    \frac{(XW_Q)_{v,*} ((XW_K)^\T)_{*,w}}{\sqrt{d_h}}. 
    \]

    Now, with this intuition, we start the formal construction. 
    Let $\cG$ be the studied graph.
    The construction is similar to the case with $\GTF$s with soft-attention and average hard-attention except for the simulation of unique hard-attention heads.
    Assume that for each node $u$ the current feature vector $\vec{f}_u$ of $u$ is encoded by the formula sequence $\vec{\psi}$ as in the proof of Theorem \ref{theorem: PLGC to GT}. Then we define the formulae $\vec{\psi}^Q$, $\vec{\psi}^K$ and $\vec{\psi}^V$ which encode for each node $u$ the vectors $\vec{q}_u \colonequals \vec{f}_u W_Q$, $\vec{k}_u \colonequals \vec{f}_u W_K$ and $\vec{v}_u \colonequals \vec{f}_u W_V$, respectively. 
    By using these formulae, we can write down the formula sequence $\vec{\theta}$ which encodes for each node $u$ the feature vector $\vec{g}_u$ that is obtained from $XW_V$ by using the selection matrix $S$ described above. 
    For each node $u$, we define the floating-point number 
    \[
    M_u \colonequals \max\left\{\, \frac{(\vec{q}_u) (\vec{k}_w)^\T}{\sqrt{d_h}} \,\bigg{|}\, w \in V(\cG) \,\right\}.
    \]
    For each feature vector pair $(\vec{f}, \vec{g}) \in \cF^d \times \cF^d$, we can define a formula $\psi_{\vec{f}, \vec{g}}$ (by using the operator $\mathsf{GB}$ and the auxiliary formulae defined above), which intuitively checks two things when evaluated at a given node $u$:
    \begin{enumerate}
        \item First, the formula checks that there is a node $w$ such that 1) $\vec{k}_w = \vec{f}$, 2) $\vec{q}_u \vec{k}_w = M_u$ and 3) $\vec{v}_w = \vec{g}$.
        \item At the same time the formula checks that there is no node $w'$ on the left of $w$ with $\vec{q}_u \vec{k}_{w'} = M_u$ (this part requires the operator $\mathsf{GB}$).
    \end{enumerate}
    By using these formulae, it is easy to define the formula sequence $\vec{\theta}$ as follows. 
    At each node, precisely one of the formulae $\psi_{\vec{f}, \vec{g}}$ is true. We construct $\vec{\theta}$ such that for each feature vector pair $(\vec{f}, \vec{g}) \in \cF^d \times \cF^d$, if $\psi_{\vec{f}, \vec{g}}$ is the true formula, then $\vec\theta$ encodes the feature vector $\vec g$, which is the feature vector selected by the selection matrix for that node.

    For the converse, let $\theta$ be a $\Pi$-formula of $\mathrm{PL+GB}$. We show by induction on the construction of a formula how to simulate $\theta$. Proposition symbols and Boolean connectives are analogous to the case with soft-attention and average hard-attention. To simulate the operator $\varphi \,\mathsf{GB}\, \psi$, we use the unique hard-attention module. Assume that the truth values of $\varphi$ and $\psi$ are encoded into a matrix $X \in \mathbb{B}^{n \times d}$ in the columns $i$ and $j$, respectively. Moreover, we assume that the truth of $\top$ is encoded into the column $k$. Now, we build an attention head $H$ as follows. The matrix $W_Q$ is just the $d \times 2$ Boolean-valued matrix, where in both columns the $k$th row is $1$ and others are $0$. The matrix $W_K$ is the $d \times 2$ matrix where in the first column the $i$th row is $\infty$ and in the second column the $j$th row is $\infty$, and other entries are $0$. Now, the resulting matrix $\frac{(X W_Q)(XW_K)^\T}{\sqrt{2}}$ is a matrix where each column is filled with some float 
    $f$ 
    that represents 
    either $\infty$ of $0$.
    Therefore, $\mathrm{UH}$ will assign $1$ to the same column in all rows. 
    Now, we define that $W_V$ is the Boolean-valued $(d \times 2)$-matrix, where in the first column the $j$th row is $1$ and other entries in both columns are zero. Now, it is easy to verify that $H(X)$ is the $(n\times 2)$-matrix where each row contains the same pair $(x,0)$, where $x$ is 
    either $0$ or $1$.
    Now, we have 
    $x = 1$ if $\DiamondG ( \psi \land \neg \mathsf{P} \varphi )$ is true and $x = 0$ otherwise.
    The rest of the construction is easy to complete.
\end{proof}

By using the result above and Lemma \ref{lemma: GNNF = GML}, it is straightforward to obtain the following theorem for GPS-networks on word-shaped graphs (without causal masking).

\begin{theorem}\label{thrm: unique GPS}
    Over word-shaped graphs, $\mathrm{GML+GB}$ has the same expressive power as unique hard-attention $\GPSF$-networks. This also applies when $\GPSF$-networks are simple.
\end{theorem}

\subsection{Positional Encodings}\label{appendix: Positional Encodings}

In this section, we shall consider our computing models with positional encodings. In the end of Section \ref{sect: float characterization}, we had a preliminary discussion about positional encodings, but let us go through the related concepts on positional encodings in a little bit more detail.

Often each $\GNN$, $\GT$ or $\GPS$-network $A = (P, L^{(1)}, \ldots, L^{(k)}, C)$ based on reals (with input dimension $\ell$) is associated with a \textbf{positional encoding} (or $\mathrm{PE}$) $\pi$, that is, a mapping that gives for a graph $\cG$ a function $\pi(\cG) \colon V(\cG) \to \R^\ell$.\footnote{In the preliminary draft of this paper \citep{ahvonen2025transformersv1}, there is a typo stating that positional encodings are isomorphism‑invariant. For example, \cite{grohe_transformers} notes that real-valued LapPE is not isomorphism‑invariant.} 
A popular $\mathrm{PE}$ is LapPE \citep{Rampasek}.
Now, $A$ with $\pi$ computes over $\cG$ a sequence of feature maps similarly to $A$, but for each vertex $v$ in $\cG$, we define $\lambda^{(0)}_v \colonequals P(\lambda)_v + \pi(\cG)_v$. 
Analogously to computing models based on reals, computing models based on floats can be associated with a positional encoding: 
Let $\cF(p, q)$ be a floating-point format and let $\pi$ be a positional encoding of dimension $d$ over $\cF$, i.e., for each graph $\cG$, $\pi$ gives a mapping $\pi(\cG) \colon V(\cG) \to \cF^d$.
Our positional encodings are ``absolute'', but one could also study ``relative'' positional encodings that give a weight between the vertices of the studied graph $\cG$, i.e., a relative positional encoding over $\cG$ is a function $ V(\cG) \times V(\cG) \to \cF^d$.
If a $\GNNF$, $\GNNGCF$, $\GNNGF$, $\GTF$ or $\GPSF$-network $T$ is paired with a positional encoding $\pi$, we assume that $T$ is defined over the same floating-point format $\cF$ as $\pi$.
Moreover, if $\cC$ is a class of computing models (over floats) we let $\cC[\pi]$ denote the class of computing models, where each computing model over a floating-point format $\cF$ is paired with the positional encoding $\pi$ over $\cF$.

Let $\Pi$ be a set of vertex labels.
Given a logic $\cL$, a $\Pi$-formula of the logic $\cL[\pi]$ is a $\Pi \cup \Pi_\pi$-formula of the logic $\cL$, where $\Pi_\pi = \{\, \ell^\pi_i \mid i \in [d(p+q+1)] \,\}$ is a fresh set of vertex label symbols (i.e. $\Pi \cap \Pi_\pi = \emptyset$). Now, a $\Pi$-formula $\varphi \in \cL[\pi]$ is interpreted over $\Pi \cup \Pi_\pi$-labeled graphs $\cG$ such that for each vertex $v$ in $\cG$, we have $\cG, v \models \ell^\pi_i$ iff the $i$th bit of $\pi(\cG)_v$ is $1$. That is, each $\Pi$-formula of $\cL[\pi]$ is interpreted over graphs where each graph includes vertex label symbols such that in each vertex $v$, the vertex label symbols of $v$ encode $\pi(\cG)_v$ in binary. In terms of $\FO$ this would mean adding an additional set of predicates to the logic.

Now, it is straightforward to see that Theorem \ref{theorem: PLGC = SGT = AHGT} and Theorem \ref{theorem: GMLGC = SGPS = AHGPS = GNNG} apply for graph transformers and $\GPS$-networks with positional encoding $\pi$ when the logic is extended with predicates that encode positional encodings. Also Theorem \ref{thrm: unique GT} and Theorem \ref{thrm: unique GPS} can be generalized in an analogous way.

\begin{corollary}
    Let $\pi$ be a positional encoding.
    The following have the same expressive power: 
    %\begin{itemize}
        %\item 
        $\PLGC[\pi]$, 
        %\item 
        soft-attention $\GTF$s with $\pi$ and
        %\item 
        average hard-attention $\GTF$s with $\pi$ (and
        %\item 
        $\GNNGCF$s with constant local aggregation functions and $\pi$).
        Moreover, the same holds when restricted to word-shaped graphs, and when the $\GTF$s and $\GNNGCF$s are simple.
\end{corollary}

\begin{corollary}
    Let $\pi$ be a positional encoding.
    The following have the same expressive power: 
    %\begin{itemize}
        %\item 
        $\GMLGC[\pi]$, 
        %\item 
        soft-attention $\GPSF$-networks with $\pi$, 
        %\item 
        average hard-attention $\GPSF$-networks with $\pi$ and 
        %\item 
        $\GNNGCF$s with $\pi$.
        Moreover, the same holds when restricted to word-shaped graphs, and when the $\GPSF$-networks and $\GNNGCF$s are simple.
\end{corollary}

\begin{corollary}
    Let $\pi$ be a positional encoding.
    The following have the same expressive power when restricted to word-shaped graphs: 
    %\begin{itemize}
        %\item 
        $\mathrm{PL+GB}[\pi]$ and 
        %\item 
        unique hard-attention $\GTF$s with $\pi$. This also applies when $\GTF$s are simple.
\end{corollary}

\begin{corollary}
    Let $\pi$ be a positional encoding.
    The following have the same expressive power when restricted to word-shaped graphs: 
    %\begin{itemize}
        %\item 
        $\mathrm{GML+GB}[\pi]$ and 
        %\item 
        unique hard-attention $\GPSF$-networks with $\pi$. This also applies when $\GPSF$-networks are simple.
\end{corollary}

While we primarily consider graph transformers and $\GPS$-networks, Theorem~\ref{theorem: GNN GML} concerning $\GNN$s can be similarly generalized to cover positional encodings.

\begin{corollary}
    Let $\pi$ be a positional encoding.
    The following pairs have the same expressive power: $\GNNF$s with $\pi$ and $\GML[\pi]$; $\GNNGF$s with $\pi$ and $\GMLG[\pi]$; $\GNNGCF$s with $\pi$ and $\GMLGC[\pi]$.
    This also applies when each type of $\GNNF$ is simple.
\end{corollary}

One could also characterize a positional encoding directly and combine that characterization with any of our Theorems \ref{theorem: PLGC = SGT = AHGT}, \ref{theorem: GNN GML}, \ref{theorem: GMLGC = SGPS = AHGPS = GNNG}, \ref{thrm: unique GT} and \ref{thrm: unique GPS}.
We leave this for future work.

\end{document}